\documentclass[12pt]{article}

\usepackage{amssymb,amsmath,color}
\usepackage[pdftex]{graphicx}
\usepackage{epsfig}
\usepackage{epstopdf}
\usepackage{framed}
\usepackage{amscd}
\newcommand{\rem}[1]{}
\newcommand{\fe}{\mathfrak{e}}

\newcommand{\fq}{\mathfrak{q}}

\newcommand{\fs}{\mathfrak{s}}
\newcommand{\fz}{\mathfrak{z}}
\newcommand{\fv}{\mathfrak{v}}
\newcommand{\fm}{\mathfrak{m}}

\newtheorem{theorem}{Theorem}

\newtheorem{definition}[theorem]{Definition}

\newtheorem{notation}[theorem]{Notation}

\newtheorem{proposition}[theorem]{Proposition}

\newtheorem{remark}[theorem]{Remark}

\numberwithin{theorem}{section}
\newenvironment{proof}[1][Proof]{\textbf{#1.} }{\ \rule{0.5em}{0.5em}}

\def\be{\begin{equation}}
\def\ee{\end{equation}}
\def\bea{\begin{eqnarray}}
\def\eea{\end{eqnarray}}
\def\ba{\begin{array}}
\def\ea{\end{array}}

\def\bOm{\boldsymbol{\Omega}}

\def\boldeta{\boldsymbol{\eta}}
\def\brho{\boldsymbol{\rho}}

\def\d{\delta}

\let\<\langle
\let \>\rangle

\newcommand{\de}{\delta}

\newcommand{\bm}{\boldsymbol{m}}
\newcommand{\bq}{\boldsymbol{q}}
\newcommand{\bu}{\boldsymbol{u}}
\newcommand{\bPsi}{\boldsymbol{\Psi}}
\newcommand{\bv}{\boldsymbol{v}}

\newcommand{\bK}{\boldsymbol{K}}

\newcommand{\br}{\boldsymbol{r}}
\newcommand{\bGam}{\boldsymbol{\Gamma}}
\newcommand{\bom}{\boldsymbol{\omega}}
\newcommand{\bgam}{\boldsymbol{\gamma}}
\newcommand{\bsigma}{\boldsymbol{\Sigma}}
\newcommand{\bpsi}{\boldsymbol{\Psi}}

\newcommand{\bbeta}{\boldsymbol{\beta}}

\newcommand{\bZ}{\boldsymbol{Z}}

\newcommand{\bmu}{\boldsymbol{\mu}}

\newcommand{\bXi}{\boldsymbol{\Xi}}

\newcommand{\bT}{\boldsymbol{T}}

\newcommand{\bkappa}{\boldsymbol{\kappa}}
\newcommand{\bpi}{\boldsymbol{\pi}}

\newcommand{\dd}[2]{\frac{d #1}{d #2}}
\newcommand{\dede}[2]{\frac{\delta #1}{\delta #2}}
\newcommand{\prt}{\partial}
\newcommand{\DD}[2]{\frac{D #1}{D #2}}

\newcommand{\om}{\omega}
\newcommand{\al}{\alpha}

\newcommand{\Om}{\Omega}

\newcommand{\lsb}{\left[}
\newcommand{\rsb}{\right]}

\newcommand{\lp}{\left(}
\newcommand{\rp}{\right)}

\newcommand{\scp}[2]{{\left\langle {#1}\, , \, {#2}\right\rangle}}

\newcommand{\CO}{{\mathcal O}}

\newcommand{\di}{{\diamond}}

\newcommand{\mR}{{\mathbb{R}}}

\newcommand{\Ad}{\mbox{Ad}}
\newcommand{\ad}{\mbox{ad}}

\newcommand{\mse}{\mathfrak{se}}
\newcommand{\mso}{\mathfrak{so}}

\newcommand{\id}{{\mathrm{id}}\,}

\newcommand{\con}{\overline}

\newcommand{\hor}{\mbox{Hor}}
\newcommand{\ver}{\mbox{Ver}}

\newtheorem{thm}{Theorem}[section]

\newtheorem{defi}[thm]{Definition}

\newtheorem{lem}[thm]{Lemma}
\newtheorem{rema}[thm]{Remark}

\newcommand{\bfi}{\bfseries\itshape}

\makeatletter
\@addtoreset{equation}{section}

\makeatother

\begin{document}

\title{Dynamics of charged molecular strands}
\author{David C. P. Ellis${}^1$, Fran\c{c}ois Gay-Balmaz${}^2$, Darryl D. Holm${}^{1,3}$,\\
Vakhtang Putkaradze${}^{4,5}$ and Tudor S. Ratiu${}^2$}
\addtocounter{footnote}{1}

\footnotetext{
Mathematics Department,
Imperial College London, SW7 2AZ, UK
\texttt{de102@imperial.ac.uk, d.holm@imperial.ac.uk}
\addtocounter{footnote}{1}}

\footnotetext{Section de
Math\'ematiques and Bernoulli Center, \'Ecole Polytechnique F\'ed\'erale de
Lausanne.
CH--1015 Lausanne. Switzerland. 
\texttt{Francois.Gay-Balmaz@epfl.ch, Tudor.Ratiu@epfl.ch}
\addtocounter{footnote}{1}}

\footnotetext{
Institute for Mathematical Sciences,
Imperial College London, SW7 2PG, UK
\addtocounter{footnote}{1}}

\footnotetext{
Department of Mathematics
Colorado State University,
Fort Collins, CO 80523-1874
\texttt{putkarad@math.colostate.edu}
\addtocounter{footnote}{1}}

\footnotetext{
Department of Mechanical Engineering,
University of New Mexico,  Albuquerque, NM 87131-1141
\addtocounter{footnote}{1}}

\maketitle

\makeatother

\maketitle

\begin{abstract}\noindent
Euler-Poincar\'e equations are derived for the dynamical folding of charged molecular strands (such as DNA) modeled as flexible continuous filamentary distributions of interacting rigid charge conformations. The new feature is that the equations of motion for the dynamics of such molecular strands are nonlocal when the screened Coulomb interactions, or Lennard-Jones potentials between pairs of charges are included. These nonlocal dynamical equations are derived in the \emph{convective} representation of continuum motion by using modified Euler-Poincar\'e and Hamilton-Pontryagin variational formulations that illuminate the various approaches within the framework of symmetry reduction of Hamilton's principle for exact geometric rods.
In the absence of nonlocal interactions, the equations recover the classical Kirchhoff theory of elastic rods in the \emph{spatial} representation.
The motion equations in the convective representation are shown to be affine Euler-Poincar\'e equations relative to a certain cocycle. This property relates the geometry of the molecular strands to that of complex fluids. An elegant change of variables allows a direct passage from the infinite dimensional point of view to the covariant formulation in terms of Lagrange-Poincar\'e equations.
In another revealing perspective, the convective representation of the nonlocal equations of molecular strand motion is transformed into quaternionic form.
\end{abstract}

\newpage

\tableofcontents

\section{Introduction}

\subsection{Physical Setup}\label{subsec-phys}

Many long molecules may be understood as strands of individual charged units. Generally, the dynamics of such strands of charged units depend both on th.. Generally, the dynamics of such strands of charged units depend both on the local elastic deformations of the strand and the nonlocal (screened electrostatic) interactions of charged units across the folds in the molecule. These electrostatic interactions depend on the spatial distances and relative orientations between the individual charged units in different locations along the strand. One important approach to such a complex problem is a full dynamical simulation. However, in spite of the importance of this approach for determining certain molecular properties, it provides little insight for analytical understanding of the dynamics. \smallskip

Continuum approaches to the dynamics of molecular strands offer an alternative theoretical understanding which is attractive because of the insight achieved in finding analytical solutions.
Many studies have addressed the elastic dynamics of the charged strands using Kirchhoff's approach \cite{Ki1859}.  For a historical review and citations of this approach see \cite{Di1992}. Recent advances, especially in the context of helical structures, appear in  \cite{GoTa1996, GoPoWe1998, BaMaSc1999, GoGoHuWo2000, HaGo2006, NeGoHa2008}. While many important results have been obtained by this approach, the generalization of the classical Kirchhoff theory to account for the torque caused by the long-range electrostatic interaction of molecules in different spatial locations along a flexible strand has not been achieved, although the force due to electrostatic interaction has been considered before. See, for example, the article \cite{DiLiMa1996} which reviews progress in the treatment of charged units distributed along a strand. 
In general, the lack of a consistent continuous model incorporating both torques and forces from electrostatic interactions has hampered analytical considerations, see for example \cite{BaMaSc1999}.  
\smallskip

The present framework does allow treatment of both torques and forces from electrostatic interactions.  We should note that even in the absence of a continuous model for nonlocal interactions, it is possible to obtain \emph{static} solutions using energy minimization techniques. For example, interesting helical static solutions of pressed elastic tubes using interactions that prevent self-intersection of the tubes were obtained in \cite{Ba-etal-2007}.
The difficulty in computing the dynamical effects of torque due to long-range interactions among the molecular subunits arises because the classical Kirchhoff theory is formulated in a frame moving with the strand, but deals with a mixture of quantities, some measured in the fixed spatial frame and some in the body frame. The torque due to  long-range interactions then presents a particular difficulty for the mixed representations in the Kirchhoff theory, because it is applied at base points of a curve that is moving in space.  That is, the spatial Euclidean distances and relative orientations of the molecules must be reconstructed at each time step during the sinuous motion and twisting of the strand before any self-consistent computation can be made of the forces and torques due to long-range electrostatic interactions.

\smallskip
In fact, even when electrostatic forces are not involved, the motion of realistic curves in space is inherently nonlocal because of the requirement that the curve not cross itself during the dynamics. In the  purely elastic Kirchhoff approach, such nonlocal considerations are neglected. Physically, however, self-intersections are prevented by the existence of a short-range potential (e.g., Lennard-Jones potential) that produces highly repulsive forces when two points along the curve approach each other. \smallskip

This paper casts the problem of strand dynamics for an arbitrary intermolecular potential into the {\bfi convective representation} introduced in \cite{HoMaRa1986} and applied in exact geometric rod theory in \cite{SiMaKr1988}. Its methods are also applicable to the consideration of Lennard-Jones potentials and the constrained motion of non-self-interacting curves.

\smallskip

If the curve were constrained to be rigidly fixed in space, and the attached molecules on this  curve were allowed to simply rotate freely at each position, the theory of motion based on nonlocal interaction between different molecules would be more straightforward. Of particular interest here is the work  \cite{Mezic2006} where a single charge is attached at each point along the fixed filament by a rigid rod of constant length that was allowed to rotate in a transverse plane. These charges were allowed to interact locally with other nearby charges that were similarly attached to planar pendula of constant length mounted transversely to the fixed filament.
\smallskip

This constrained motion can be generalized to allow flexible motion of the strand (time-dependent bend, twist and writhe) while also including  the degrees of freedom of molecular orientation excited during the process of, say, DNA folding. According to this class of models, a DNA molecule is represented as a flexible filament or strand, along which are attached various different types of rigid conformations of sub-molecules that may \emph{swivel} relative to each other in three dimensions under their mutual interactions. The flexibility of the filament arises physically because the electrostatic interaction between any pair of these rigid conformations either along the filament or across from one loop to another of its folds is much weaker than the internal interactions that maintain the shape of an individual charged conformation.
\smallskip

This paper considers rigid charge conformations (RCCs) mounted along a flexible filament. The RCCs are more complex than the planar pendula considered in \cite{Mezic2006}. They are allowed to interact with each other via a nonlocal (\emph{e.g.} screened electrostatic, or Lennard-Jones) potential.
Our investigation is based on the {\bfi geometrically exact rod theory} of Simo \emph{et al.} \cite{SiMaKr1988}, which is expressed in the {\bfi convective representation} of continuum mechanics.
The rotations of rigid charge conformations along the flexible filament are illustrated in Figure \ref{model-fig}.

\begin{figure} [h]

\centering

\includegraphics[width=12cm]{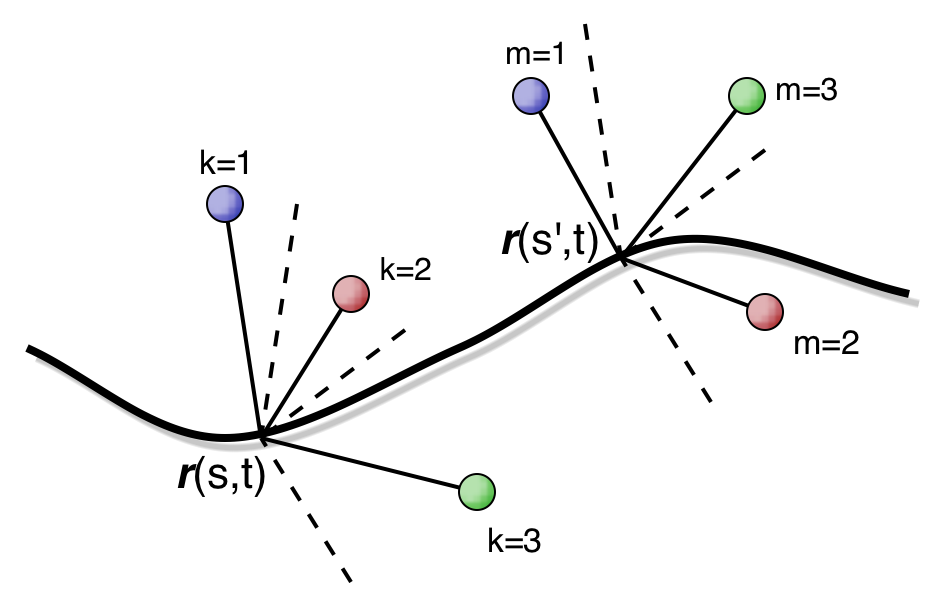}
 \caption{Rigid conformations of charges are distributed along a curve. Note that this is a spatial representation of the orientation.}
 \label{model-fig}
 \end{figure}

These rigid conformations of multiple charges are allowed to interact via an effective many-body potential representing their screened electrostatic interactions. The nonlocal interactions among these RCCs depend on their spatial separations and relative orientations, which are both allowed to evolve dynamically. Thus, the inertial motion of a pair of  RCCs mounted at any two spatial points $\br(s,t)$ and $\br(s',t)$ along the filament is governed by an effective potential interaction energy that depends on their separation and relative orientation. The filament is taken to be one-dimensional, although the orientations of the rigid charged conformations mounted along it may be three-dimensional. A practical example to which our filament approach applies is the vinylidene fluoride (VDF) oligomer \cite{No2003}, which may be approximated by a strand carrying a dipole moment whose orientation is perpendicular to the axis of the strand. The VDF oligomer strand is straight for small lengths, but it forms complex shapes due to electrostatic interactions for longer lengths.

\smallskip
The theory presented here generalizes directly to the case when the dimension of the underlying manifold (filament) is greater than unity, and so is applicable to such problems as the motion of charged sheets, or charged elastically deformable media. While we present part of the relevant geometry here, we leave its applications in higher dimensions for a later publication.

\smallskip

\paragraph{Plan.}The remainder of the paper is organized as follows.
Subsection~\ref{subsec.mathsetup} outlines the content of the paper in mathematical terms by giving an overview of the various spatial representations of filament dynamics discussed here from the canonical and covariant point of views.
Subsection~\ref{sec:PreviousStudies} connects our results to the earlier literature. Paragraph~\ref{sec:Elastic-Rod} relates the theory presented here to the classical elastic rod approach pioneered by Kirchhoff. The need to keep track of spatial separations in  long-range electrostatic interactions requires that we write the dynamics in either the spatial or convective representations, as opposed to the Kirchhoff mixed representation. 
Paragraph \ref{fixedfilament} considers the simplified case when the orientations of the RCCs along the curve may depend on time, but the position of any point $s$ along the curve is fixed, thereby connecting to earlier work in \cite{Mezic2006}. Section~\ref{moving-filament} incorporates the flexible motion of the filament into the dynamics by using the geometrically exact rod theory of Simo \emph{et al.} \cite{SiMaKr1988}.  The equations of motion are derived in convective form by using a {\bfi modified Hamilton-Pontryagin} and {\bfi modified Euler-Poincar\'e approach}  to allow for nonlocal interactions, in Subsections~\ref{sec:Hamilton-Pontryagin} and \ref{sec:Euler-Poincare}, respectively.
These equations are formulated as conservation laws along the filament in Section~\ref{sec:Conservation} and their {\bfi affine Lie-Poisson} Hamiltonian structure is elucidated in Section~\ref{sec:Hamiltonian}.
Section~\ref{sec:affineEP} explains the background for the {\bfi affine Euler-Poincar\'e} and {\bfi affine Hamilton-Pontryagin} approaches and applies this framework to the dynamics of charged strands. 
Section~\ref{sec:Coordinatechange} introduces a remarkable change of coordinates that decouples the equations into their {\bfi horizontal} and {\bfi vertical} parts.
Section~\ref{sec:Lagrange-Poincare} explains the geometric structure of this coordinate change and leads to the {\bfi covariant Lagrange-Poincar\'e} formulation.
Subsection~\ref{sec.strand.gen} and Section~\ref{sec:ApplicationLP} discuss generalizations of the molecular strand to higher dimensions. In Section~\ref{sec:ApplicationLP}, the equations of motion are obtained by an alternative covariant Lagrange-Poincar\'e approach.
Section~\ref{quatform} provides a useful representation of the convective  frame dynamics of the flexible strand using quaternions.
Section~\ref{sec:Conclusions} briefly summarizes our conclusions and sets out possible directions for further studies.

\subsection{Mathematical Setup}\label{subsec.mathsetup}

\subsubsection{Description of the variables involved}
\label{definition_variables}

In the \textit{Lagrangian representation}, the motion is described by the variables $\Lambda(s,t)\in SO(3)$ and $\br(s,t)\in\mathbb{R}^3$. The vector $\br(s,t)$ is the \textit{spatial position of the filament} and the variable $\Lambda(s,t)$ denotes the \textit{rotation of the RCC at the point $s$ along the filament at time $t$}. By taking the time and space derivatives, we find the \textit{material velocity} $(\dot\Lambda(s,t),\dot{\br}(s,t))$ and the angular and linear \textit{deformation gradients} $(\Lambda'(s,t),\br'(s,t))$, respectively. Given $\Lambda$ and $\br$, we define notation for the following \textit{reduced variables}
\begin{eqnarray}
\Om&=\, \Lambda^{-1} \Lambda'&  \in \mso(3)
\,,\nonumber
\\
\om&= \, \Lambda^{-1} \dot{\Lambda}& \in \mso(3)
\,,\nonumber
\\
\bGam&= \, \Lambda^{-1} \br'& \in \mathbb{R}^3
\,,\label{bundle.coords}
\\
\bgam&= \, \Lambda^{-1} \dot{\br}& \in \mathbb{R}^3
\,,\nonumber
\\
\brho&= \, \Lambda^{-1} \br&  \in \mathbb{R}^3
\,.\nonumber
\end{eqnarray}

\begin{remark}[Notation]
{\rm   
Quantities defined using derivatives in $s$ are denoted using capital Greek letters, whereas lower-case Greek letters (except for $\brho$) denote quantities whose definitions involve derivatives with respect to time. Bold letters, for example $\bGam$, denote vectors in $\mathbb{R}^3$ whereas $\Omega$ is a $3\times3$ skew-symmetric matrix in the Lie algebra $\mso(3)$.
}
\end{remark}

\begin{definition}
The {\bfi hat map} $\widehat{\phantom{\bOm}}: (\mathbb{R}^3, \times ) \to (\mso(3), [\,,])$ is the Lie algebra isomorphism given by $\widehat{\boldsymbol{u}}\boldsymbol{v}=\boldsymbol{u} \times \boldsymbol{v}$ for all $\boldsymbol{v}\in\mathbb{R}^3$.
\end{definition}

Thus, in an orthonormal basis of $\mathbb{R}^3$ and $\boldsymbol{u} \in \mathbb{R}^3$, the $3\times3$ antisymmetric matrix $u:=\widehat{\boldsymbol{u}} \in \mso(3)$ has entries
\begin{equation}
u_{jk}
=
(\widehat{\boldsymbol{u}})_{jk}
=
-\epsilon_{jkl}\boldsymbol{u}^l
\, .
\label{hatmap-components}
\end{equation}
Here $\epsilon_{jkl}$ with $j,k,l \in \{1,2,3\}$ is the totally antisymmetric tensor density with $\epsilon_{123}=+1$ that defines the cross product of vectors in $\mathbb{R}^3$.
In what follows, we shall employ this notation by writing $\Omega : = \widehat{\bOm}$ and $\omega : = \widehat{\bom}$.
\medskip

The physical interpretation of the variables \eqref{bundle.coords} is as follows:  The variable $\brho(s,t)$ represents the \textit{position of the filament in space as viewed by an observer} who rotates with the RCC at $(s,t)$. The variables \big($\bOm(s,t), \bGam(s,t)$\big) describe the \textit{deformation gradients as viewed by an observer} who rotates with the RCC.  The variables \big($\bom(s,t), \bgam(s,t)$\big) describe the \textit{body angular velocity} and the \textit{linear velocity as viewed by an observer} who rotates with the RCC.

\subsubsection{The canonical point of view}

The canonical viewpoint of continuum dynamics derives the equations of motion by applying a process of reduction by symmetry to a cotangent bundle $T^*Q$ endowed with a canonical symplectic form. This approach has been extensively studied for fluids, see for example \cite{MaRaWe1984} for the Hamiltonian description and \cite{HoMaRa1998} for the Lagrangian side. In hydrodynamics, $Q$ is the product of a Lie group $G$ and a representation space $V$ on which the group acts linearly as $G\,\times\,V\to V$. The dual space, 
$V^*$, is the space of linearly advected quantities such as the mass density or the magnetic field. The associated process of reduction by symmetry under the action of $G$ is called {\bfi Lie-Poisson reduction for semidirect products}. For such systems (in the left version), we have the relations
\begin{equation}\label{semidirect_relation}
\begin{array}{l}
\xi=g^{-1}\dot g\,,\\
a=g^{-1}a_0\,,
\end{array}
\end{equation}
where $g(t)\in G$ is the Lagrangian motion, $\xi(t)$ is the convection velocity, and $a(t)\in V^*$ is the evolution of the advected quantity for a given initial condition $a_0$. Note that $a(t)$ is also a convected quantity.
For the molecular strand we have $g=(\Lambda,\br)$ and $a=(\bOm,\bGam,\brho)$. However, the relations \eqref{bundle.coords} cannot be recovered from \eqref{semidirect_relation} because the variables $(\bOm,\bGam,\brho)$ are not linearly advected. Thus, a generalization of \eqref{semidirect_relation} is needed, in which $G$ acts on $a$ by an {\bfi affine action}. Such a generalization is given by the process of {\bfi affine Euler-Poincar\'e reduction} developed in the context of complex fluids in \cite{Gay-Bara2007}. This theory, which we recall in Section~\ref{sec:affineEP}, produces the relations
\begin{equation}\label{semidirect_relation_affine}
\begin{array}{l}
\xi=g^{-1}\dot g
\,,\\
a=g^{-1}a_0+c(g^{-1})
\,,
\end{array}
\end{equation}
where $c$ is a group one-cocycle.%
\footnote{ That is, $c$ satisfies the property $c(fg) = c(f) + fc(g)$, where $f$ acts on $c(g)$ by a left representation, as discussed in Section~\ref{sec:affineEP}.}
If we take $a_0=0$, then the advected quantity evolves as
\[
a=c(g^{-1})\,.
\]
Remarkably, the evolution of $(\bOm,\bGam,\brho)$ in \eqref{bundle.coords} is precisely of this form for a well chosen cocycle. The variables $(\Lambda(s,t),\br(s,t))$ are interpreted as \textit{time-dependent curves in the infinite dimensional Lie group}
\[
G=\mathcal{F}(I,SE(3))
\]
of all $SE(3)$-valued smooth functions on $I$.
\begin{remark}{\rm The variables
\[
(\bom,\bgam)=(\Lambda,\br)^{-1}(\dot\Lambda,\dot\br)\,,
\]
their associated momenta
\[
(\bmu,\bbeta):=\left(\frac{\delta l}{\delta \bom},\frac{\delta l}{\delta \bgam}\right),
\]
and the affine advected variables $(\bOm,\bGam,\brho)$ are all {\bfi convective quantities}, see \cite{MaRa2002}. In this context, convective quantities are also called {\bfi body quantities}, since they are defined in a frame following the motion of the molecular strand.

In contrast, the variables
\begin{equation} 
\left(\mathbf{\bpi}^{(S)},\mathbf{p}^{(S)}\right)
:=\operatorname{Ad}^*_{(\Lambda,\br)^{-1}}\left(\frac{\delta l}{\delta \bom},\frac{\delta l}{\delta \bgam}\right)
\label{BodytoKirchhoff} 
\end{equation} 
are {\bfi spatial quantities}, i.e., they are defined at fixed points in Euclidean space. 
 }
\end{remark}

\subsubsection{The covariant point of view}

The covariant point of view interprets the Lagrangian variables $(\Lambda(s,t),\br(s,t))$ as a \textit{section $\sigma$ of a trivial fiber bundle
over spacetime}. The bundle is given by
\[
\pi_{XP}:P=X\times SE(3)\rightarrow X,\quad X=[0,L]\times\mathbb{R}\ni (s,t)=x
\]
and the section $\sigma$ is naturally defined by
\[
\sigma(x)=(x,\Lambda(x),\br(x)).
\]
Carrying out a covariant reduction for the first jet extension $j^1\sigma$ recovers the relations \eqref{bundle.coords} in a natural way.

From the covariant point of view, the main results and relationships among the various sections of the paper may be understood by tracing through the following diagram:
\[
\begin{CD}
 \Lambda(s,t) \in SO(3) &&  SO(3)\\
@VVV  @VVV \\
\lp \Lambda, \br\rp(s,t)\in SE(3) @>>> P @>\pi_{XP}>> X  \\
 @V\pi_{SE(3)} VV  @V\pi_{\Sigma P} VV @VV\id V \\
 \brho(s,t) \in \mR^3 @>>> \Sigma @>>\pi_{X\Sigma} > X  \\
\end{CD}
\]
This diagram lays out the sequences of coordinates and manifolds used here to represent the molecular strand dynamics.
The definitions of the various projections are given in Section~\ref{sec:Coordinatechange} and Section~\ref{sec:Lagrange-Poincare}. For the moment we shall only introduce the variables and use the diagram to describe the relationships between the two main convective representations of filament dynamics used here.

Recall from \S\ref{definition_variables} that  $\Lambda(s,t) \in SO(3)$ denotes the rotation of the RCC at a point $s$ along the filament at time $t$.  The vector $\br(s,t)$ is the spatial position of the filament at $(s,t)$ and $\brho(s,t)$ is the position of the filament in space as viewed by an observer who rotates with the RCC at $(s,t)$.  The left hand vertical sequence describes the fibers over spacetime of the right hand vertical sequence.  Therefore we regard the variables introduced above as local descriptions of sections of the spaces in the right hand vertical sequence.  Since the bundles involved are trivial, these local sections are sufficient to describe the global sections.  For example, $\brho(s,t) \in \mR^3$ corresponds to the section
\[
\big(s,t, \brho(s,t) \big) \in \Sigma = X \times \mR^3
\,.
\]
This characterizes the space $\Sigma$. 
These sections naturally lead to the use of linear maps in place of tangent vectors. Thus, in place of tangent bundles, we introduce jet bundles, as described in Section~\ref{sec:Lagrange-Poincare}.  These jet bundles have local representations in terms of fiber linear maps from $TX$ into vector spaces covering the section $\brho(s,t)$ in the vector bundles $TP/SO(3)$, $T\Sigma$, and $\ad P$.  The space of all such linear maps is denoted, for example, $L(TX, T\Sigma)$.  With this in mind we can express the reduced configuration spaces in terms of the following diagram:

{\footnotesize
\[
\!\!\!\!
\begin{CD}
 \lp \brho, \bOm ds + \bom dt\rp_{SO(3)} \in L(TX,\ad SE(3)) && L(TX,\ad P)\\
 @VVV  @VVV \\
\lp \brho, \bGam ds + \bgam dt, \bOm ds + \bom dt\rp \in L(TX, TSE(3)/SO(3)) @>>> J^1P/SO(3) @>T\pi_{XP}/SO(3)>> L(TX, TX)  \\
 @V T\pi_{SE(3)}/SO(3) VV  @V T\pi_{\Sigma P}/SO(3) VV @VV\id V \\
 \lp \brho, \brho_s ds + \brho_t dt\rp \in L(TX,T\mR^3) @>>> J^1\Sigma @>> T\pi_{X\Sigma} > L(TX, TX) \\
\end{CD}
\]
}

where $T\pi_{XP}/SO(3)$ and $T\pi_{\Sigma P}/SO(3)$ denote the tangent maps $T\pi_{XP}$ and $T\pi_{\Sigma P}$ acting by composition on linear maps with values in the fibers of $TP/SO(3)$.

This diagram arises in Section~\ref{sec:Coordinatechange}, where we find that the equations of motion can be drastically simplified by a change of variables that passes from the upper to the lower horizontal sequence in the diagram.  The geometrical significance of this change of variables is made more precise in Section~\ref{sec:Lagrange-Poincare}.  In particular, the variables $(\brho, \brho_t, \brho_s)$ introduced in Section~\ref{sec:Coordinatechange} are coordinates in the $J^1\Sigma$ component, where $\Sigma:=P/SO(3)$. The variables $(\bOm, \bom)$ describe the $L(TX,\ad P)$ component, where $\ad P$ denotes the adjoint bundle associated with $P$. 
The coordinates on $L(TX,TP/SO(3))$ 
recover the definitions \eqref{bundle.coords} by covariant reduction. 
Within this framework, we develop a new formulation of molecular strand dynamics.  The modified Euler-Poincar\'e argument for a Lagrangian $l(\bom,\bgam,\bOm,\bGam,\brho)$ and related arguments that use the same set of variables $\lp\brho, \bGam, \bgam, \bOm, \bom\rp$ are formulated as sections of $TP/SO(3)= TX \times TSE(3)/SO(3)$, where $X := [0,L] \times \mR\ni (s,t)$ is the spacetime.  

\begin{remark}[Defining the convective representation]
{\rm Formulas \eqref{bundle.coords} for the variables in the upper horizontal sequence in the diagram above define the convective representation of the exact geometric rod theory \cite{SiMaKr1988}.  We shall see in a moment how these variables in the convective representation are related to the Kirchhoff variables for strand dynamics.}
\end{remark}

In Section~\ref{sec:Lagrange-Poincare} we introduce the structure described by the left hand vertical sequence.  This allows us to construct the lower horizontal sequence and the right hand vertical sequence.  When we formulate the problem in the new geometry in Section~\ref{sec:Lagrange-Poincare} using the {\bfi covariant Lagrange-Poincar\'e approach}, we recover the change of variables introduced in Section~\ref{sec:Coordinatechange}. In the new  geometry, the following coordinates are used, cf. the definitions in equation \eqref{bundle.coords}:
\begin{eqnarray}
\Om&=&\Lambda^{-1} \Lambda' \  \in \mso(3)
\,,\nonumber
\\
\om&=&\Lambda^{-1} \dot{\Lambda} \ \in \mso(3)
\,,\nonumber
\\
\brho_s&=&\Lambda^{-1} \br' 
- \Lambda^{-1}\Lambda'\brho \   \in \mathbb{R}^3
\,,\label{bundle.coords2}
\\
\brho_t&=&\Lambda^{-1} \dot{\br}
- \Lambda^{-1}\dot \Lambda \brho \  \in \mathbb{R}^3
\,,\nonumber
\\
\brho&=&\Lambda^{-1} \br \   \in \mathbb{R}^3
\,.\nonumber
\end{eqnarray}
This covariant Lagrange-Poincar\'e approach is generalized in Subsection~\ref{sec.strand.gen} and Section~\ref{sec:ApplicationLP} to consider higher dimensional problems such as the molecular sheet, as well as problems such as the spin chain that have different microstructures.

\subsection{Connection to previous studies}
\label{sec:PreviousStudies}
\subsubsection{Purely elastic motion and Kirchhoff equations for elastic rod}
\label{sec:Elastic-Rod}
The results of this paper may be compared to the classical Kirchhoff theory of the purely elastic rod, particularly in terms of the available conservation laws. This comparison was presented for the purely elastic case, \emph{i.e.}, the Lagrangian $l$ is an explicit (local) function of the variables $l=l(\bom, \bgam, \bOm, \bGam, \brho)$ 
 in \cite{SiMaKr1988}. This work is extended here to the case of nonlocal interaction. 

Of particular interest to us are the balance laws for angular and linear momenta.  For this comparison, we shall use the notation of \cite{DiLiMa1996}. For simplicity, we assume that the position $\br (s)$ along the filament is given by the arc length $s$. This assumption conveniently avoids extra factors of $\big| \bGam(s) \big|$ in the expressions. We shall also mention here that in order to connect to the Kirchhoff theory, we need to make an explicit choice of $\Lambda(s) \in SO(3)$ as a  transformation matrix from the fixed  orthonormal basis $\{\mathbf{E}_1, \mathbf{E}_2, \mathbf{E}_3\}$ of $\mathbb{R}^3$ to the orthonormal  basis of directors $\{ \mathbf{d}_1(s), \mathbf{d}_2(s), \mathbf{d}_3(s) \}$ describing the orientation of the filament (see Figure \ref{model-fig}), that is, 
\begin{equation}\label{def_basis_d}\mathbf{d}_i(s) = \Lambda_i ^k(s)\mathbf{E}_k, \quad i = 1,2,3.\end{equation}
There is some ambiguity in the choice of the basis $\{ \mathbf{d}_1(s), \mathbf{d}_2(s), \mathbf{d}_3(s) \}$ at every given point. The most popular selection of the basis  is governed by the so-called \emph{natural frame}. We shall not go into the details of this right now and refer the reader to \cite{DiLiMa1996} for a more complete discussion. In principle, we need not have taken this particular choice of $\Lambda$, since for  rigid charge conformations (RCC), the \emph{relative} configuration of charges is not changed under the dynamics, and the configuration of an RCC state at each point $s$ is completely described by a pair $(\Lambda, \br) \in SE(3)$.  Taking $\Lambda$ to be a different presentation of RCC would lead to a transformation 
$\Lambda(s,t) \rightarrow A \Lambda(s,t)$ where $A \in SO(3)$ is a fixed matrix. While our description is equivalent in this case, the explicit relation to Kirchhoff formulas is cumbersome.

We shall note that if the charge conformations were allowed to deform, then $\Lambda$ would no longer be an element of $SO(3)$. Instead, the charge conformation would be described by a general matrix $\Lambda$ and a vector $\br \in \mathbb{R}^3$. No explicit relation to Kirchhoff's formulas is possible in this case.

As mentioned in Section~\ref{subsec-phys}, Kirchhoff's approach does not allow for a simple computation of Euclidian distances between the charges unless the spatial length-scale of the rigid charge conformations (RCCs) holding the charges at given point $\boldeta_k(s)$ is negligible. It is interesting that in the more complex case considered here, the equations become formally equivalent to Kirchhoff's equations, provided the effects of non-locality are computed appropriately. In particular, one requires an appropriate mapping from the convective representation to the  Kirchhoff representation, as well as  some identities connecting nonlocal contributions to the \emph{total} derivatives of the Lagrangian. 
This mapping is amplified in more detail in Subsection~\ref{sec:Conservation} below.

The linear momentum density $\mathbf{p}$ is defined as $\mathbf{p}(s)=\rho_d(s) \dot{\br}(s)$, where $\rho_d(s)$ is the local mass density of the rod. In that case, the kinetic energy due to linear motion $K_{lin}$ is given by
\[ 
K_{lin}= \frac{1}{2} \int \rho_d(s) \|\dot{\br}(s) \|^2 \mbox{d} s 
= 
\frac{1}{2} \int \rho_d(s) \|  \Lambda^{-1} \dot{\br}(s)\|^2 \mbox{d} s
=  \frac{1}{2} \int \rho_d(s) \|\bgam(s) \|^2 \mbox{d} s 
\, . 
\] 
Consequently, the variable $\mathbf{p}$ and the linear momentum 
$\delta K_{lin}/\delta \bgam$ are related by 
\begin{equation} 
\mathbf{p}=\rho_d \dot{\br} = \Lambda \rho_d \bgam = \Lambda \frac{\delta K_{lin}}{\delta \bgam}
\, .
\label{pgamconnect}  
\end{equation}

After these preliminaries, we are ready for a detailed comparison with Kirchhoff's theory. A point on a rod in Kirchhoff's theory is parameterized by the distance $\br(s,t)$ measured from a \emph{fixed} point in space.
The $i$-th component of the local angular momentum in the \emph{body} frame $\{ \mathbf{d}_1(s), \mathbf{d}_2(s), \mathbf{d}_3(s) \}$ is defined by  $\bpi^i(s):=\mathbb{I}^i_{j}(s) \bom^j(s)$, where 
$\bom^j(s)$ is the $j$-th component of \emph{body} angular velocity given by $\widehat{\bom}(s) := \omega(s)=\Lambda(s)^{-1} \dot{\Lambda}(s)$, and 
$\mathbb{I}^i_{j}(s)$ is the local value of the inertia tensor. Note that the inertia tensor $\mathbb{I}(s)$ expressed in body coordinates is time-independent. Thus the local kinetic energy due to rotation is given by 
\[ 
K_{rot}= \frac{1}{2} \int \bom (s)\cdot  \mathbb{I}(s)  \bom(s) \mbox{d} s 
\] 
and hence
\[ 
\bpi = \mathbb{I} \bom = \frac{\delta K_{rot}}{\delta \bom }. 
\] 
To write the conservation laws, we need to express the angular momentum in the fixed spatial frame $\{ \mathbf{E}_1, \mathbf{E}_2, \mathbf{E}_3\}$. To distinguish it from $\bpi$ which was expressed in the  \emph{body} frame 
$\{ \mathbf{d}_1(s), \mathbf{d}_2(s), \mathbf{d}_3(s) \}$, we shall denote the 
same vector in the fixed spatial frame $\{ \mathbf{E}_1, \mathbf{E}_2, \mathbf{E}_3\}$ by $\bpi^{(\mathbf{E})}$. The same convention will be used for all other vectors. Thus, \eqref{def_basis_d} yields 
\[
\boldsymbol{\pi}(s)=\bpi^i (s)\mathbf{d}_i(s)= \mathbb{I}^i_{j}(s) \bom^j(s) \mathbf{d}_i(s) = 
\mathbb{I}^i_{j}(s) \bom^j(s) \Lambda^k_i(s)\mathbf{E}_k  = \bpi^{( \mathbf{E}),k}(s)\mathbf{E}_k, 
\]
so the $k$-th component of the spatial angular momentum is expressed in terms of the local body quantities $\mathbb{I}^i_{j}(s) $ and $\bom^k(s)$ as 
\begin{equation} 
\boldsymbol{\pi}^{(\mathbf{E}),k}=\Lambda^k_i\mathbb{I}^i_{j} \bom^j = \left[ \Lambda  \mathbb{I} \bom  \right]^k =  \left[ \Lambda  \frac{\delta K_{rot}}{\delta \bom}  \right]^k .
\label{piomconnect0}
\end{equation} 
Thus, the vector $\boldsymbol{\pi}^{( \mathbf{E})}(s)$  of \textit{body} angular momentum expressed in the \textit{spatial} frame  is connected to the local body quantities as 
\begin{equation} 
\boldsymbol{\pi}^{(\mathbf{E})}= \Lambda \mathbb{I} \bom   = \Lambda  \frac{\delta K_{rot}}{\delta \bom}.
\label{piomconnect}
\end{equation} 
\begin{remark} {\rm 
The vector $\boldsymbol{\pi}^{(\mathbf{E})}$ and all other vectors with the subscript $(\mathbf{E})$ do not have the physical meaning of the angular momentum in the fixed frame. The true angular and linear  momenta in the spatial frame will be denoted (see immediately below) with the superscript $(S)$. The quantities with the superscript $(\mathbf{E})$ are just the transformations of vectors with respect to rotation of the base frame. 
Hopefully, no confusion  should arise over this distinction. }
\end{remark} 

In general, it is assumed for physical reasons, that the Lagrangian in Kirchhoff's formulation has the form 
\begin{equation} 
l(\bom,\bgam,\bOm,\bGam)=K_{lin}(\bgam)+K_{rot}(\bom)-E(\bOm,\bGam) \, , 
\label{lkirchhoff}
\end{equation} 
where $E(\bOm, \bGam)$ is a certain explicit  function of $\bOm$ and $\bGam$ (not necessarily quadratic). In this case, the body  
forces $\mathbf{n}={\delta l}/{\delta \bGam}$  and torques $\mathbf{m}={\delta l}/{\delta \bOm}$ are connected to the transformed quantities $\mathbf{n}^{(\mathbf{E})}, \mathbf{m}^{(\mathbf{E})}$ in Kirchhoff's theory as 
\begin{equation} 
\mathbf{n}^{(\mathbf{E})}=\Lambda \frac{\delta l}{\delta \bGam}
\, , 
\quad 
\quad 
\mathbf{m}^{(\mathbf{E})}=\Lambda \frac{\delta l}{\delta \bOm}
\, . 
\label{forcetorque} 
\end{equation} 
Next, we use formula (\ref{BodytoKirchhoff}) to transfer to spatial frame. Identifying elements of $\mathfrak{se}(3)^\ast$  with pairs of vectors  $ \left(\boldsymbol{\mu}, \boldsymbol{\eta} \right) \in \mathbb{R}^3$, produces a useful formula for the coadjoint action 
\begin{equation} 
{\rm Ad}^*_{(\Lambda,\br)^{-1}} \left( \boldsymbol{\mu}, \boldsymbol{\eta} \right) 
= \left( \Lambda \boldsymbol{\mu}+ \br \times \Lambda \boldsymbol{\eta}, 
\Lambda \boldsymbol{\eta} \right). 
\label{AdstarR3}
\end{equation} 
Thus, the spatial momenta -- denoted by a superscript $(S)$ -- become 
\begin{align} 
\left(\mathbf{\bpi}^{(S)},\mathbf{p}^{(S)}\right):&=\operatorname{Ad}^*_{(\Lambda,\br)^{-1}}\left(\frac{\delta l}{\delta \bom},\frac{\delta l}{\delta \bgam}\right)=\left( 
\Lambda \frac{\delta l}{\delta \bom}+ \br \times \Lambda \frac{\delta l}{\delta \bgam}
\, , \, 
 \Lambda \frac{\delta l}{\delta \bgam}
\right)\nonumber\\
&=\left( \bpi^{(\mathbf{E})}+ \br \times \mathbf{p}^{(\mathbf{E})} \, , \, \mathbf{p}^{(\mathbf{E})} 
\right)
\, , 
\label{spacemomenta} 
\end{align} 
upon using (\ref{pgamconnect}) and (\ref{piomconnect}).  Analogously, using \eqref{forcetorque}, the spatial torques  $\mathbf{m}^{(S)}$ and forces $\mathbf{n}^{(S)}$  are  
\begin{align} 
\left(\mathbf{m}^{(S)},\mathbf{n}^{(S)}\right)
:&=\operatorname{Ad}^*_{(\Lambda,\br)^{-1}}
\left(\frac{\delta l}{\delta \bOm},\frac{\delta l}{\delta \bGam}\right)
=\left( 
\Lambda \frac{\delta l}{\delta \bOm}+ \br \times \Lambda \frac{\delta l}{\delta \bGam}
\, , \, 
 \Lambda \frac{\delta l}{\delta \bGam}
\right)\nonumber 
\\
&=\left( \mathbf{m}^{(\mathbf{E})}+ \br \times \mathbf{n}^{(\mathbf{E})} \, , \, \mathbf{n}^{(\mathbf{E})} 
\right). 
\label{spaceforces} 
\end{align} 
The conservation laws in Kirchhoff theory may now  be written as 
\begin{equation} 
\frac{\partial}{\partial t} (\mathbf{\bpi}^{(S)},\mathbf{p}^{(S)})
+
\frac{\partial}{\partial s}(\mathbf{m}^{(S)},\mathbf{n}^{(S)})
=(\mathbf{T}, \mathbf{f}), 
\label{Kirchcons0} 
\end{equation} 
where $\mathbf{T}$ and $\mathbf{f}$  are external torques and forces, respectively.  Equations (\ref{Kirchcons0}) give, componentwise,  the following linear and angular momentum conservation laws (cf. equations (2.5.5) and (2.5.7) of \cite{DiLiMa1996})
\begin{align}
&\frac{\partial}{\partial t} \mathbf{p}^{(\mathbf{E})} +\frac{\partial }{\partial s} \lp \mathbf{n}^{(\mathbf{E})} -\mathbf{F} \rp =0
\,,
\label{momcons}
 \\
&\frac{\partial }{\partial t} \lp  \boldsymbol{\pi}^{(\mathbf{E})}+\mathbf{r} \times \mathbf{p} ^{(\mathbf{E})}\rp 
+\frac{\partial }{\partial s} \lp \mathbf{m}^{(\mathbf{E})} +\br \times \mathbf{n} ^{(\mathbf{E})}- \mathbf{L} \rp=0
\, ,
\label{angcons}
\end{align}
where $\mathbf{F}$ and $\mathbf{L}$ are defined as the indefinite integrals,
\[
\mathbf{F}=\int^s \mathbf{f} (q)\, \mbox{d} q
\quad\hbox{and}\quad
\mathbf{L}=\int^s 
[\br(q) \times \mathbf{f}(q) + \mathbf{T}(q)]\, \mbox{d} q
\,.
\]
Opening the brackets in (\ref{momcons}) and (\ref{angcons}) gives 
the balances of linear and angular momenta in Kirchhoff's approach (cf. eqs. (2.3.5) and (2.3.6) of \cite{DiLiMa1996})
\begin{align}
&\frac{\partial \mathbf{p}^{(\mathbf{E})}}{\partial t}  + \frac{\partial \mathbf{n}^{(\mathbf{E})}}{\partial s} =\mathbf{f},
\label{linmom}
\\
&\frac{\partial \boldsymbol{\pi}^{(\mathbf{E})}}{\partial t} + \frac{\partial  \mathbf{m}^{(\mathbf{E})}}{\partial s}+
\frac{\partial \br}{\partial s} \times \mathbf{n}^{(\mathbf{E})}
=\mathbf{T}.
\label{angmom}
\end{align}
To see how these Kirchhoff balance laws look in our representation, 
one may substitute relations \eqref{spacemomenta} and \eqref{spaceforces}  into \eqref{Kirchcons0} to obtain: 
\begin{equation}\label{consgen-intro}\frac{\partial}{\partial t}
\left[{\rm Ad}^*_{(\Lambda,\br)^{-1}}
\left( \frac{\delta l}{\delta \bom} \, , \,  \frac{\delta l}{\delta \bgam}
\right) \right] + 
\frac{\partial}{\partial s}  \left[
 {\rm Ad}^*_{(\Lambda,\br)^{-1}}
\left( \frac{\delta l}{\delta \bOm} \, , \,
  \frac{\delta l}{\delta \bGam} 
\right) \right]
= ( \mathbf{T},  \mathbf{f}).
\end{equation}
Assume now that the Lagrangian $l$ depends explicitly on the  additional variable $\brho = \Lambda^{-1}\br$. This corresponds to \emph{potential} forces exerting forces and torques.  As shown in Section~\ref{sec:Conservation},  in our representation the external torques $\mathbf{T}$ and forces $\mathbf{f}$ are  given by
\begin{equation} 
(\mathbf{T},\mathbf{f})={\rm Ad}^*_{(\Lambda,\br)^{-1}}\left(\frac{\delta l}{\delta \brho} \times \brho, \frac{\delta l}{\delta \brho} \right) \, . 
\label{Tf0} 
\end{equation} 
By using formula \eqref{AdstarR3}, relationship \eqref{Tf0} simplifies to 
\begin{align} 
{\rm Ad}^*_{(\Lambda,\br)^{-1}} \left(\frac{\delta l}{\delta \brho}  \times \brho,\frac{\delta l}{\delta \brho}\right) &= 
\left( 
\Lambda \left( \frac{\delta l}{\delta \brho} \times \brho \right) + 
\br \times \Lambda \frac{\delta l}{\delta \brho}, 
\Lambda \frac{\delta l}{\delta \brho}
\right) = 
\nonumber 
\\ 
&\left( 
 \left( \Lambda \frac{\delta l}{\delta \brho}\right)  \times\left( \Lambda \brho \right) + 
\br \times \Lambda \frac{\delta l}{\delta \brho}, 
\Lambda \frac{\delta l}{\delta \brho}
\right) 
=\left(0,\Lambda \frac{\delta l}{\delta \brho} \right)
\, ,
\label{Tffin}
\end{align} 
upon remembering that $\Lambda \brho=\br$. 

\begin{remark}[Potential external forces produce no net torque]$\quad$\\
\rm
Equation \eqref{Tffin} implies that potential external forces produce no net torque on the strand. Hence, the nonzero torques $\mathbf{T}$ in \eqref{consgen-intro} must arise from non-potential forces. 
\end{remark}

The conservation law \eqref{consgen-intro} is formally equivalent to the classical expressions in \eqref{momcons} and \eqref{angcons}, even if  nonlocal interaction is present. This equivalence shows how the classical results \eqref{momcons} and \eqref{angcons} generalize for the
case of nonlocal orientation-dependent interactions. Clearly, the conservation laws are simpler in the Kirchhoff representation. However, if nonlocal interactions are present (called \emph{self-interaction forces} in \cite{DiLiMa1996}), the computation of the required time-dependent Euclidian distances in the interaction energy becomes problematic in the classical Kirchhoff approach.
As we shall see below in  Section~\ref{sec:Conservation}, these conservation laws may be obtained, even when nonlocal interactions are present.  Also in Section~\ref{sec:Conservation},  we show that the nonlocal forces are included in the conservation law (\ref{consgen-intro}) and are expressed in the same form as a purely elastic conservation law. 

The balance laws  \eqref{momcons} and \eqref{angcons} 
are much simpler in appearance than the expressions in (\ref{consgen-intro}), as they do not involve  computing $(\Lambda,\br)$ at each instant in time and point in space. Thus for elastic rods, in the absence of  nonlocal interactions, the Kirchhoff mixed (convective-spatial) representation appears simpler than either the convective or spatial representations. However, the presence of nonlocal terms summons the more general convective approach.

\begin{remark}
[Reduction of static equations of motion to the heavy top]$\quad$\\
{\rm A famous analogy exists between the stationary shapes of an elastic filament and the equations of motion of a heavy top \cite{KeMa1997, NiGo1999}. In our formulation, this analogy appears naturally. This shows the advantage of using the geometric approach, even in the study of classical problems of filament dynamics.
This paper focuses, however, on the derivations and geometric structures underlying the dynamical equations, rather than on the \emph{solutions} of the equations.}
\end{remark}

\subsubsection{Reductions for a fixed filament}
\label{fixedfilament}

We may briefly apply the ideas of the present paper to the particular case of a fixed filament, in order to compare the motion equations with those arising in \cite{Mezic2006}.

The analysis of filament dynamics induced by nonlocal interactions simplifies in the case when the position of the filament is fixed as $\br(s)$ and does not depend on time. For simplicity, we shall assume that the filament is straight and $s$ is the arc length, so that $\br(s) = (s,0,0)^T$. 
The following reduced Lagrangian is invariant under the left action of the Lie group $SO(3)$:
\begin{equation}
l
=
\underbrace{\
\frac{1}{2} \int
\bom (s) \cdot \mathbb{I}(s) \bom (s) 
\mbox{d} s\
}_{\hbox{Kinetic energy}}
-
\underbrace{\
 \frac{1}{2} \int
 f \lp \bOm(s) \rp  \mbox{d} s\
}_{\hbox{Elastic energy}}
-
\underbrace{\
 \frac{1}{2} \iint U\lp \brho(s),\xi(s,s') \rp
\mbox{d} s \mbox{d}s'\
}_{\hbox{Potential energy}}
\, .
\label{reducedL}
\end{equation}
A nonlocal interaction term appears in the potential energy of relative orientation in this Lagrangian. This term involves a variable 
\[
\xi(s,s')= \Lambda^{-1}(s) \Lambda(s') \in SO(3)
\,,
\]
which defines the relative orientation of rigid charge conformations at two \emph{different points} in space. The variable $\xi(s,s')\in SO(3)$ is invariant with respect to simultaneous rotations of the coordinate frames for $s$ and $s'$, but it is not an element of a Lie algebra. In particular, $\xi(s,s')$ is not a vector.
The presence of nonlocal interactions introduces dependence on relative orientation and thereby produces new types of nonlocal terms in the corresponding Euler-Poincar\'e dynamics obtained in applying reduction by $SO(3)$ symmetry to Hamilton's principle.

\begin{remark}[Aim of the paper]$\quad$\\
{\rm
In this paper, the influence of non-locality due to electrostatic forces on rod mechanics is studied by using various approaches, including the Euler-Poincar\'e variational method.  This variational approach leads to an equivalent Lie-Poisson Hamiltonian formulation of the new equations appearing below in \eqref{hpvarsig}, \eqref{hpvarpsi}.
Applying the Ad$\,^*_{(\Lambda,\br)^{-1}}$ transformation from convective to spatial variables in these equations streamlines the form and exposes the meaning of the interplay among their various local and nonlocal terms, relative to the Kirchhoff theory.}
\end{remark}

\paragraph{Euler-Poincar\'e dynamics} Euler-Poincar\'e dynamics for the angular dynamics on a \emph{fixed}  filament follows from stationarity of the left invariant total action 
\[
S=\int l(\bom,\brho,\xi, \bOm)\,dt
\,.
\]
Note that this case does not require computation of the evolution equation for $\bgam$ since the filament is assumed to be fixed in space, i.e., $\bgam=\Lambda^{-1} \dot{\br}=0$.
The variational derivative $\delta  S$ for such a Lagrangian is computed as,
\begin{equation}
\delta S
=
\int \delta l(\bom,\Lambda, \bOm) \,dt
 =
 \int
\left<
\frac{\delta l}{\delta \bom}, \delta \bom
\right>
+
\left<
\Lambda^{-1} \frac{\delta l}{\delta \Lambda}, \Sigma
\right>
+
\left<
\frac{\delta l}{\delta \bOm}, \delta \bOm
\right>
\,dt
\,,
\label{delta-lag}
\end{equation}
for the notation $\Sigma = \Lambda^{-1} \delta \Lambda$.
As we will see in Section~\ref{var-defs}, these variations are related by
\begin{align*}
\delta \omega&=\dot{\Sigma}+\lsb \omega, \Sigma \rsb =
\dot{\Sigma}+{\rm ad}_\omega \Sigma 
\,,\\
\delta \Omega &= \Sigma\,'+\lsb \,\Omega, \Sigma \rsb
=
\Sigma\,'+{\rm ad}_{\Omega} \Sigma 
\,,\\
\delta\brho&=-\mathbf{\Sigma}\times\brho
\,.
\end{align*}
Substituting these formulas into \eqref{delta-lag} then integrating by parts in the time $t$ and one-dimensional coordinate $s$ along the fiber yields
\begin{align}
\delta S
=
\int \delta l \,\mbox{d}t
=
\int
\Bigg<
&
- \frac{\partial}{\partial t} \frac{\delta l}{\delta \omega} +
{\rm ad}^*_\omega  \frac{\delta l}{\delta \omega}
-\frac{\partial}{\partial s} \frac{\delta l}{\delta \Om} +
{\rm ad}^*_ {\Om}  \frac{\delta l}{\delta \Om}
 \nonumber
 \\
 &
-
 \int
 \lp
-\,\frac{\partial  U}{\partial \xi}(s,s') \xi^T(s,s')
+\xi(s,s') \left( \frac{\partial  U}{\partial \xi}(s,s') \right)^T
\rp
\nonumber
\\
&-
 \left(
 \frac{\delta l}{\delta \brho} \times \brho
 \right)^{\! \widehat{\phantom{o}}}
\mbox{d}s'
\, ,
\Sigma
\Bigg>
\,\mbox{d}t
\,.
\label{deltalvar}
\end{align}
Thus, Hamilton's principle $\delta S=0$ implies the {\bfi Euler-Poincar\'e} equations,
\begin{align}\label{modified_EP_fixed_filament}
-\frac{\partial}{\partial t} \frac{\delta l}{\delta \omega} +
{\rm ad}^*_\omega  \frac{\delta l}{\delta \omega}
= & \,\frac{\partial}{\partial s} \frac{\delta l}{\delta\Om}
-
{\rm ad}^*_{\Om}  \frac{\delta l}{\delta\Om}
+
\int
\left(
 \frac{\delta l}{\delta \brho} \times \brho
 \right)^{\!\bf \widehat{\,\,\,}  }
\mbox{d} s'\nonumber\\
&+
\int
 \lp
-\frac{\partial  U}{\partial \xi}(s,s') \xi^T(s,s')
+\xi(s,s') \left( \frac{\partial  U}{\partial \xi}(s,s') \right)^T
\rp
\mbox{d}s'
\,. 
\end{align} 
Note that these Euler-Poincar\'e equations are \textit{nonlocal}. That is, they are integral-partial-differential equations.

Reformulating (\ref{modified_EP_fixed_filament}) in terms of vectors yields the following generalization of equations considered by \cite{Mezic2006}, written in a familiar vector form: 
\begin{align}\label{rigid_bodies}
\lp  -\frac{d}{dt}  \frac{ \delta l}{\delta  \bom} \right.
&
+\left.
 \frac{ \delta l}{\delta  \bom} \times \bom
 -
 \frac{\partial}{\partial s} \frac{\delta l}{\delta\bOm}
 -
\bOm \times   \frac{\delta l}{\delta \bOm}
+
\brho \times \frac{\delta l}{\delta \brho} 
\rp
^{\!\bf \widehat{\,\,\,}  }
\nonumber
\\
&=
 \int
 \lp
-\frac{\partial  U}{\partial \xi}(s,s') \xi^T(s,s')
+\xi(s,s') \left( \frac{\partial  U}{\partial \xi}(s,s') \right)^T
\rp
\mbox{d} s'
 \,.
\end{align}
In order to close the system, one computes the time derivative of
${\xi}(s,s')=\Lambda^{-1}(s') \Lambda (s) $:
\begin{align}
\dot{\xi}(s,s') &=- \Lambda^{-1}(s') \dot{\Lambda}(s')
\Lambda^{-1}(s')\Lambda (s)+ \Lambda^{-1}(s') \dot{\Lambda}({s})
\nonumber
\\
&= -\omega(s') \xi({s,s'}) + \xi(s,s') \omega({s})
\, .
\label{dotxi}
\end{align}
This expression is not quite a commutator because different positions $s$ and $s'$ appear in $\omega$. However, operating with $\xi^{-1}$ from the left in equation \eqref{dotxi} gives a proper Lie-algebraic expression for the reconstruction of the relative orientation,
\begin{equation}
\label{ }
\xi^{-1}\dot{\xi} (s,s')
=
 \omega({s})  - {\rm Ad}_{\xi^{-1}(s,s')}\omega({s}')
 \,.
\end{equation}
Formulas \eqref{modified_EP_fixed_filament} - \eqref{dotxi} generalize the results in \cite{Mezic2006} for a fixed filament from $SO(2)$ to $SO(3)$ rotations.

\section{Motion of exact self-interacting geometric rods}
\label{moving-filament}
\subsection{Problem set-up}
Suppose each rigid conformation of charges is identical and the $k$-th electrical charge is positioned near a given spatial point $\br$ through which the curve of base points of the RCCs passes. This curve is parametrized by a variable $s$ which need not be the arc length. Rather, we take $s \in [0,L]$ to be a parameter spanning a fixed interval.%
\footnote{Note: limiting its parametrization to a fixed interval does not mean that the filament is inextensible.}

The spatial reference (undisturbed) state for the $k$-th charge in a given RCC is the sum $\br(s)+\boldeta_k(s)$. That is, $\boldeta_k(s)$ is a vector of constant length that determines the position of the $k$-th electrical charge relative to the point $\br(s)$ along the curve in its reference configuration. The $\boldeta_k(s)$ specify the shape of the rigid conformation of charges. At time $t$ the position $\mathbf{c}_k$ of the $k$-th charge in the rigid conformation anchored at spatial position $\br(s)$ along the curve parametrized by $s$ may rotate to a new position corresponding to the orientation $\Lambda(s,t)$ in the expression
\begin{equation}
\mathbf{c}_k(s)=\br(s)+\Lambda(s,t) \boldeta_k(s)
\,,\quad\hbox{where}\quad
\Lambda(s,0) = {\rm Id}
\,.
\label{charge-position}
\end{equation}
This rigid conformational rotation is illustrated in Figure \ref{model-fig}. In Mezic's case \cite{Mezic2006}, the rotation is in the plane, so that $\Lambda \in SO(2)$, and there is only one charge, so $k=1$.

\subsection{Convected representation of nonlocal potential energy}\label{convected_representation}
One part of the potential energy of interaction between rigid conformations of charges  at spatial coordinates $\br(s)$ and $\br(s')$ along the filament depends only on the magnitude $|\mathbf{c}_m(s') -\mathbf{c}_k (s)|$ of the vector from charge $k$ at spatial position $\mathbf{c}_k(s)$ to charge $m$ at spatial position $\mathbf{c}_m(s' )$. This is the Euclidean spatial distance
\begin{equation}
d_{k,m} (s,s')=\big|\mathbf{c}_m(s') -\mathbf{c}_k (s) \big|
\label{distance0}
\end{equation}
between the $k$-th and $m$-th charges in the two conformations whose base points are at $\br(s)$ and $\br(s')$, respectively. In this notation, the  potential energy is given by
\begin{equation}
E=E_{loc}(\bOm,\bGam)-\sum_{k,m} \frac{1}{2}
\iint U\lp  d_{k,m} (s, s')  \rp
 \Big| \frac{d \br}{ds} (s) \Big|  \, \,  \Big| \frac{d \br}{ds} (s') \Big| \mbox{d}s \mbox{d}s'\,,
\label{Energy0}
\end{equation}
for an appropriate physical choice of the \textit{interparticle interaction potential} $U(d_{k,m})$, and 
the quantities $\bOm$, $\bGam$ (and $\bom$, $\bgam$ and $\brho$ below) are defined in  (\ref{bundle.coords}). The part $E_{loc}(\bOm,\bGam)$ represents the purely elastic part of the potential, and is usually taken to be a quadratic function of the deformations $(\bOm,\bGam)$, but more complex expressions are possible as well; we shall not restrict the functional form of that dependence.  The total Lagrangian $l$  is then written as a sum of local $l_{loc}$ and nonlocal $l_{np}$: 
\begin{equation} 
l_{loc}=K(\bom,\bgam)-E_{loc}(\bOm,\bGam,\brho)
\quad 
\mbox{and} 
\quad 
l_{np}=-E_{np} 
\, , 
\label{locL} 
\end{equation}
where  $K$ is the kinetic energy that depends only on the local velocities $\bom, \bgam$. For the sake of generality, here and everywhere else below we shall simply consider the total Lagrangian to be a sum of the local part $l_{loc}(\bom,\bgam,\bOm,\bGam,\brho)$, and the nonlocal part given by (\ref{Energy0}): 
\begin{equation} \label{l_loc}
l=l_{loc}(\bom,\bgam,\bOm,\bGam,\brho) + l_{np} 
\, . 
\end{equation} 
The scalar distance $d_{k,m}$ in \eqref{distance0} and \eqref{Energy0} may also be expressed in terms of vectors seen from the frame of orientation of the rigid body at a spatial point $\br(s)$ along the filament, as
\begin{align}
d_{k,m}(s,s')
&=\left| \mathbf{c}_m(s') -\mathbf{c}_k (s) \right|\nonumber\\
&=\left|
 \Lambda^{-1}(s) \left( \mathbf{c}_m(s') -\mathbf{c}_k (s) \right)
 \right|
\nonumber
\\
&=
\left|
 \Lambda^{-1}(s) \left(\br(s')-\br(s))  + \Lambda^{-1}(s) \Lambda(s')\boldeta_m(s') - \boldeta_k(s) \right)
 \right|
 \nonumber
 \\
 & =:
 \left|
\bkappa(s,s')+ \xi(s,s') \boldeta_m (s') - \boldeta_k(s)
 \right|
 \, ,
\label{LiePoissondist}
\end{align}
where we have defined the quantities
\begin{equation}
\bkappa(s,s'):= \Lambda^{-1}(s) \lp \br(s')-\br(s) \rp  \in \mathbb{R}^3
\quad\hbox{and}\quad
\xi(s,s'):=  \Lambda^{-1}(s) \Lambda(s') \in SO(3)
\,.
\label{rho-xi-defs}
\end{equation}
The first of these quantities is the spatial vector from $\br(s)$ to $\br(s')$, as seen from the orientation $\Lambda(s)$ of the rigid charge conformation located at coordinate label $s$ along the filament. 
The second is the relative orientation of the rigid charge conformations located at coordinate labels $s$ and $s'$. For later use, we record the {\bfi transposition identities},
\begin{equation}
\xi(s,s')^T=\xi(s',s)=\xi(s,s')^{-1}
\,,
\label{xi-trans}
\end{equation}
which follow from the definition of $\xi(s,s')$ in \eqref{rho-xi-defs}.

\begin{remark}[Left $SO(3)$ invariance]
{\rm Both the body separation vector $\bkappa(s,s')$ and the relative
orientation $\xi(s,s')$ defined in \eqref{rho-xi-defs} are invariant under rotations of the spatial coordinate system obtained by the left action 
\[
(\br(s')-\br(s))\to O \big(\br(s') - \br(s) \big)
\quad\hbox{and}\quad
\Lambda\to O\Lambda
\,,
\] 
by any element $O$ of the rotation group $SO(3)$.}
\end{remark}

\begin{proposition}[Left $SE(3)$ invariance]
The quantities $(\xi,\bkappa)\in SO(3)\times \mathbb{R}^3$ defined in \eqref{rho-xi-defs} are invariant under all transformations of the special Euclidean group $SE(3)$ acting on the left.
\end{proposition}
\begin{proof}
As a set, the special Euclidean group $SE(3)$
is the Cartesian product $SE(3)=SO(3)\times \mathbb{R}^3$ whose elements are denoted as $(\Lambda,\br)$. Its group multiplication is given, e.g., in \cite{Ho2008} by the {\bfi semidirect-product action},
\begin{equation}
(\Lambda_1,\br_1)(\Lambda_2,\br_2)
=(\Lambda_1\Lambda_2,\br_1
+
\Lambda_1\br_2)
\,,
\label{SE3-multiplyrule}
\end{equation}
where the action of $\Lambda\in SO(3)$ on $\br\in \mathbb{R}^3$ is denoted as the concatenation $\Lambda\br$ and the other notation is standard. For the choice
\[
(\Lambda_1,\br_1)=(\Lambda,\br)^{-1}(s)
\quad\hbox{and}\quad
(\Lambda_2,\br_2)=(\Lambda, \br)(s^\prime)
\,,
\]
the $SE(3)$ multiplication rule \eqref{SE3-multiplyrule} yields the quantities $(\xi,\bkappa)\in SO(3)\times \mathbb{R}^3$ as
\begin{equation}
(\Lambda,\br)^{-1}(s)(\Lambda, \br)(s^\prime) = (\xi(s,s^\prime),\bkappa(s,s^\prime))
\,.
\label{SE3-act-xikappa}
\end{equation}
This expression is invariant under the left action $(\Lambda,\br)\to(O,\bv)(\Lambda,\br)$ of any element $(O,\bv)$ of the special Euclidean group $SE(3)$.
\end{proof}

\begin{remark}{\rm 
The $SE(3)$ setting will be especially important to the development of the Lagrange-Poincar\'e formulation of the dynamical filament equations in Section~\ref{sec:Lagrange-Poincare}.}
\end{remark}

Next, let us define the following $SE(3)$-invariant quantities, where prime denotes the derivative with respect to $s$ and dot is the derivative with respect to $t$:
\begin{eqnarray}
\Omega &\!\!\!\!\!\!\!\!\!\!\!\!\!\!\!\!\!\! =\, \Lambda^{-1} \Lambda' &  \in \mso(3)
\,,\nonumber
\\
\omega& \!\!\!\!\!\!\!\!\!\!\!\!\!\!\!\!\!\! = \, \Lambda^{-1} \dot{\Lambda}& \in \mso(3)
\,,\nonumber
\\
\bGam&\!\!\!\!\!\!\!\!\!\!\!\!\!\!\!\!\!\! = \, \Lambda^{-1} \br'& \in \mathbb{R}^3
\,,\label{rhodef}
\\
\bgam& \!\!\!\!\!\!\!\!\!\!\!\!\!\!\!\!\!\! = \, \Lambda^{-1} \dot{\br}& \in \mathbb{R}^3
\,,\nonumber
\\
\brho&= \, \Lambda^{-1} (\br - \br_0)&  \in \mathbb{R}^3
\,.\nonumber
\end{eqnarray}
Hereafter, we shall choose $\br_0=\mathbf{0}$ to recover the bundle coordinates (\ref{bundle.coords}).

\begin{remark}{\rm  
Note that here $\Lambda, \br, \Om, \om, \bGam, \bgam, \brho$ are interpreted as functions of the two variables $s$ and $t$. It will be important to see these variables as time-dependent curves with values in functions spaces. For example, we can interpret $\Lambda(s,t)$ as a function of space and time
\[
(s,t)\in [0,L]\times\mathbb{R}\mapsto \Lambda(s,t)\in SO(3),
\]
or we can see $\Lambda$ as a curve in an infinite dimensional Lie group
\[
t\in \mathbb{R}\mapsto \Lambda(\cdot ,t)\in \mathcal{F}([0,L],SO(3)),
\]
where $\mathcal{F}([0,L],SO(3))$ denotes the group of smooth functions defined on $[0,L]$ with values in $SO(3)$.

This observation is fundamental and leads to two different geometric approaches to the same equations: the {\bfi affine Euler-Poincar\'e} and the {\bfi covariant Lagrange-Poincar\'e} approaches.}
\end{remark}

\begin{remark}{\rm 
Since $\Lambda \in SO(3)$, one finds that
\begin{equation}
\bigg| \frac{d \br}{ds} (s) \bigg| 
= \bigg| \Lambda^{-1} \frac{d \br}{ds}  (s) \bigg|
=\big| \bGam(s) \big|
\, ,
\end{equation}
and the nonlocal part of the potential energy in \eqref{Energy0} reduces to
\begin{equation}
E_{np}=-\sum_{k,m} \frac{1}{2}
\iint U\lp  d_{k,m}(s,s')  \rp
 \big| \bGam(s) \big|   \big| \bGam(s') \big| \mbox{d}s \mbox{d}s'
\,.
\label{Energy1}
\end{equation}}
\end{remark}

\begin{remark} 
{\rm 
Everywhere in this paper, we shall assume that the nonlocal part of the Lagrangian $l_{np}$ is a function or functional of $\bGam$, $\xi$ and $\bkappa$. It could, for example, be expressed in the integral form 
\begin{equation} 
l_{np}( \xi, \bkappa,\bGam)= \iint U\big( \xi(s,s'), \bkappa(s,s'), \bGam(s), \bGam(s')\big) \mbox{d}s \mbox{d} s'
\label{lnpgen}
\end{equation} 
or be a more general functional. In this work, we shall consistently use formula \eqref{lnpgen} to make our computations more explicit, although of course the methods would apply to more general functionals. 
Clearly, expression (\ref{Energy0}) is a reduction of (\ref{lnpgen})  obtained when the energy of the system of charges is a (half)-sum of interactions between all charges.  This happens, for example, when investigating electrostatic or screened electrostatic charges in a linear media. 

Even though the expression $l_{np}=l_{np}(\xi,\bkappa,\bGam)$ is rather general, it is interesting to note that physical systems exist whose nonlocal interactions do not satisfy that law.
For example, the electrostatic potential around a DNA molecule immersed in a fluid satisfies the nonlinear Poisson-Boltzmann equation and finding the potential in that case is a well-known problem for supercomputers \cite{Ba-etal-2001}. If we could somehow explicitly solve this equation -- which is impossible --  we would be able to write a more general Lagrangian. In general, to apply our theories to this problem we would have to couple our methods to a numerical solution of the Poisson-Boltzmann equation at each time step. We shall also note that for the case of linearized Poisson-Boltzmann equation, we can solve the equation exactly in terms of the screened electrostatic potential $U(r)= e^{-k r}/r$ and the expression (\ref{lnpgen}) holds. 
} 
\end{remark}

\subsection{Kinematics}\label{kinematics}

We first define auxiliary kinematic equations that hold without any reference to dynamics. We call these {\bfi advection relations}, in order to distinguish them from the {\bfi dynamical  equations}  
(derived later) that balance the forces determined from the physics of the problem. In contrast, the advection relations hold for all strands, irrespective of their dynamic properties. 

In order to derive the first set of advection relations, 
we compute the time and space derivatives of $\brho(s,t)= \Lambda^{-1}\br(s,t)$. First, the $s$-derivative along the filament is given  by:
\begin{equation*}
\brho' = - \Lambda^{-1} \Lambda ' \Lambda^{-1} \br + \Lambda^{-1} \br'
\,,
\end{equation*}
and hence equations \eqref{rhodef} imply
\begin{equation}
\brho' = - \Omega \brho +\bGam=-\bOm \times  \brho +\bGam
\,.
\label{rhospacederiv}
\end{equation}
 Next, the time derivative is written as,
\begin{equation}
\dot{\brho} = - \Lambda^{-1}\dot{ \Lambda } \Lambda^{-1} \br + \Lambda^{-1} \dot{\br}
\,,
\label{rhokin}
\end{equation}
and equations \eqref{rhodef} yield the formula,
\begin{equation}
\dot{\brho} = - \omega \brho +\bgam=- \bom \times  \brho +\bgam
\,.
\label{rhotimederiv}
\end{equation}

The next set of advection relations is derived by the equality of cross-derivatives with respect to $t$ and $s$ for any sufficiently smooth quantity. 
First, we use the fact that  $\partial_s \partial_t \br=\partial_t \partial_s\br$.
Equality of these cross-derivatives implies the relations,
\[
\bgam\, '=- \bOm \times \bgam + \Lambda^{-1} \dot{\br} ' 
\,,
\]
and
\[
\dot{\bGam} =- \bom \times \bGam + \Lambda^{-1} \dot{\br} ' 
\,.
\]
The difference of the last two equations yields the following relation
\begin{equation}
\dot{\bGam}+ \bom \times \bGam =
\bgam\,'+ \bOm \times \bgam \, .
\label{kincond}
\end{equation}
As we shall see later, the latter is a type of {\bfi zero-curvature relation}.
Similarly, equality of cross-derivatives $\partial_s \partial_t \Lambda=\partial_t \partial_s \Lambda$
yields  the other advection relation,
\begin{equation}
\dot{\bOm}-\bom'=\bom \times \bOm \, .
\label{kincondom}
\end{equation}

\subsection{Remark on the $n$-dimensional generalization and the use of other groups}\label{n_dimensional_generalization}
The previous setting may be generalized to $n$ dimensions and to arbitrary Lie groups. This is not only useful for the generalization of charged strands to membranes and, more generally, to deformable media; it also gives a more transparent vision of the underlying geometric structure underlying the phenomena.

\medskip
Consider the semidirect product $\mathcal{O}\,\circledS\,E$ of a Lie group $\mathcal{O}$ with a  \textit{left} representation space $E$. The variables $\br$ and $\Lambda$ defined above are now functions defined on a spacetime $\mathcal{D}\times \mathbb{R}$, where $\mathcal{D}$ is a $n$-dimensional manifold:
\[
\Lambda : (s,t)\in \mathcal{D}\times\mathbb{R}\rightarrow \Lambda(s,t)\in\mathcal{O},\quad\text{and}\quad r:(s,t)\in \mathcal{D}\times\mathbb{R}\rightarrow r(s,t)\in E.
\]
We will avoid using boldface notation as the functions we consider may be more general geometric quantities, not only vectors. 
As before, ``dot'' $(\ \dot{}\  )$ over a quantity denotes its time derivative. The derivative with respect to a variable in $\mathcal{D}$ is denoted by $\mathbf{d}$; for $\mathcal{D} = [0,L] $ this was previously denoted by ``prime'' $(\,{}'\,)$. 
The definitions \eqref{rhodef} become
\begin{align}\label{rhodef_n_dimensional}
\Omega&=\Lambda^{-1} \mathbf{d}\Lambda : T\mathcal{D}\rightarrow \mathfrak{o}
\,,\nonumber\\
\omega&=\Lambda^{-1} \dot{\Lambda} :\mathcal{D}\rightarrow\mathfrak{o}
\,,\nonumber\\
\Gamma&=\Lambda^{-1} \mathbf{d} r :T\mathcal{D}\rightarrow E
\,,\\
\gamma&=\Lambda^{-1} \dot r :\mathcal{D}\rightarrow E
\,,\nonumber\\
\rho&=\Lambda^{-1} r :\mathcal{D}\rightarrow\nonumber E
\,.
\end{align}
Thus, if we interpret $(\Lambda,r)$ as a curve in the group $\mathcal{F}(\mathcal{D},\mathcal{O}\,\circledS\,E)$, the previous definition can be rewritten as
\begin{align*}
(\omega,\gamma)&=(\Lambda, r)^{-1}(\dot\Lambda,\dot r)
\,,\\
(\Omega,\Gamma,\rho)&=c((\Lambda, r)^{-1}),
\end{align*}
where $c$ is defined by
\begin{equation}\label{cocycle}
c(\Lambda, r)=\big((\Lambda, r)\,\mathbf{d}(\Lambda, r)^{-1},- r \big)
\,.
\end{equation}
Remarkably, $c$ is a {\bfi group one-cocycle}.
Thus, the previous definition simply says that $(\Omega,\Gamma,\rho)$ are affine advected quantities with zero initial values. This observation strongly  suggests a relation with the affine Euler-Poincar\'e theory developed in the context of complex fluids in \cite{Gay-Bara2007}.

On the other hand, if we interpret $(\Lambda,r)$ as a section of the trivial principal bundle
\[
(\mathcal{D}\times \mathbb{R})\times \mathcal{O}\,\circledS\,E\rightarrow \mathcal{D}\times \mathbb{R}
\]
over spacetime, definition \eqref{rhodef_n_dimensional} simply says that the variables $(\Omega,\omega,\rho)$ are obtained by reduction by the subgroup $\mathcal{O}$ of the first jet extension of $(\Lambda,\br)$. This, in turn, leads to a relation with the covariant Lagrange-Poincar\'e reduction for field theories developed in \cite{CaRa03}.
Note that by choosing the one-dimensional interval $\mathcal{D}=[0,L]$, the Lie group $\mathcal{O}=SO(3)$ and left representation space $E=\mathbb{R}^3$, one recovers the advection of charged strands discussed earlier. 

\begin{remark}\rm
Generalizing to higher dimensions reveals certain distinct aspects of the underlying geometry of the problem that are not distinguished in considering the particular case of the charged strands. For example, in the case of \textit{charged sheets} or \textit{charged elastic deformed media}, $\mathcal{D}$ is a domain in $\mathbb{R}^n$, with $n=2$ or $3$, respectively, so the coordinate $s$ has several dimensions. Then, $\bGam$ should be considered as a set of vectors $\bGam_1, \ldots, \bGam_n$. 
Likewise, for the problem of flexible strands of rigid charge conformations the distinct objects $E$ and $\mathfrak{o}$ both coincide with $\mathbb{R}^3$. This coincidence is removed in higher dimensions and thereby clarifies the underlying geometric structure of the theory. 
\end{remark}

\section{Derivation of the equations of motion}

In this section we shall derive the convective equations of motion for a charged strand from two different, but equivalent, viewpoints. The first derivation is based on the classical Hamilton-Pontryagin (HP) approach in control theory (see, for example, \cite{Bloch2003}). The second derivation is based on the Euler-Poincar\'e (EP) approach, modified to include additional terms describing nonlocal contributions. We shall present both methods in this section. 

The Hamilton-Pontryagin Theorem \ref{HamPont-thm} elegantly delivers the key formulas for the Euler-Poincar\'e equations and leads efficiently to its Lie-Poisson Hamiltonian formulation. Perhaps surprisingly, the HP theorem produces these results without invoking any properties of how the invariance group of the Lagrangian acts on the configuration space (a manifold) and leads directly to the equations of motion \eqref{hpvarsig} and \eqref{hpvarpsi}. The equivalent alternative EP derivation of these formulas does explicitly involve the action of the Lie group on the configuration space and is, therefore, slightly more elaborate than the HP theorem. This elaboration invokes the Lie group action on the configuration space and thereby provides additional information. In particular, the EP approach reveals how the Lie group action on the configuration space induces the affine structure of the EP equations \eqref{EPsigma} and \eqref{EPpsi}. The alternative EP approach also yields information that explains precisely how the canonical phase space (the cotangent bundle of the configuration manifold) maps to the Lie-Poisson space associated to the action, which is the dual of the Lie algebra of symmetries via the momentum map defined by the infinitesimal affine Lie algebra action. We explore in detail the EP route in this paper because it explicitly reveals the role of the Lie group action in symmetry reduction. In Section~\ref{sec:affineEP} it will be shown that the derivation of the EP equations and of the associated variational principle are corollaries of general theorems for systems whose configuration space is a Lie group. The complementary, but less transparent, HP route reveals other perspectives and results whose abstract general formulation will be explored in future work.

\subsection{A modified Hamilton-Pontryagin approach}
\label{sec:Hamilton-Pontryagin}
\subsubsection{Filament dynamics}

We begin with the Hamilton-Pontryagin approach applied to the case when the Lagrangian includes only the local part, so $l=l_{loc}(\bom, \bgam, \bOm,  \bGam, \brho)$. 
In  order to simplify the formulas and avoid extra factors in the integrals, we shall implicitly incorporate the dependence of the nonlocal potential on $\bGam=\Lambda^{-1}\mathbf{r'} $. See \eqref{EPsigma} and \eqref{EPpsi} below for the explicit formulas for the pairwise potential, that gives multiplication by  $| \bGam(s)| |\bGam(s')|$ in the integrals.

Inspired by the classical Hamilton-Pontryagin approach, we introduce Lagrange multipliers for the holonomic constraints that impose the defining relations \eqref{rhodef} for the five quantities $(\bom, \bgam, \bOm,  \bGam, \brho)$.

\begin{thm}[Hamilton-Pontryagin theorem for filament dynamics]
\label{HamPont-thm}$\quad$\\
The equations for filament dynamics arise from the variational principle $\delta S=0$ with action $S$ given by
\begin{eqnarray*}
S &=& \int l (\bom, \bgam, \bOm,  \bGam, \brho)\,  dt
+
\iint 
\bigg(
\boldsymbol{\pi}\cdot
\lp  \Lambda^{-1} \dot{\Lambda} - \omega  \rp
+
\boldsymbol{\Pi} \cdot
\lp  \Lambda^{-1} \Lambda' - \Omega  \rp
\\
&& \qquad
+\
 \mathbf{R} \cdot \lp  \Lambda^{-1}\br - \brho   \rp
+ \boldsymbol{ \mu} \cdot \lp  \Lambda^{-1}\dot{\br} - \bgam \rp
+  \mathbf{M} \cdot \lp  \Lambda^{-1}\br' - \bGam \rp
 \bigg)   \mbox{d}s\,\mbox{d}t.
\end{eqnarray*}
These equations are
\begin{equation*}
\frac{\delta l }{\delta \brho} - \mathbf{R}
 = 0
\,, \quad
\frac{\delta l }{\delta \bom} - \boldsymbol{\pi}
 = 0
\,, \quad
\frac{\delta l }{\delta \bOm} - \boldsymbol{\Pi}
 = 0
\,, \quad
\frac{\delta l }{\delta \bgam} - \boldsymbol{\mu}
 = 0
\,, \quad
\frac{\delta l }{\delta \bGam} - \mathbf{ M}
 = 0,
 \,
\end{equation*}
\begin{equation*}
\dot{\boldsymbol{\pi}} + \bom\times\boldsymbol{\pi} +\boldsymbol{\Pi}' + \bOm\times \boldsymbol{\Pi}
+ \bgam\times\boldsymbol{\mu}
+ \bGam\times \mathbf{M}
+ \brho\times\mathbf{R}
= 0
 \,,
\end{equation*}
and
\begin{equation*}
\dot{\boldsymbol{\mu}} + \bom\times \boldsymbol{\mu} + \mathbf{M}' + \bOm\times \mathbf{M} - \mathbf{R}
= 0
 \,,
\end{equation*}
together with the constraints,
\begin{displaymath}
\Lambda^{-1} \dot{\Lambda} = \omega
\,,\quad  
\Lambda^{-1} {\Lambda}' = \Om 
\,,\quad 
\Lambda^{-1}\br = \brho 
\,, \quad \Lambda^{-1}\dot{\br} = \bgam 
\,,\quad 
\Lambda^{-1}\br' = \bGam.
\end{displaymath}
\end{thm}

We begin by computing the variations of the quantities appearing in the action $S $.

\begin{lem} \label{lamvar} The variations of the quantities in $\Lambda$ and $\br$ of the formulas in \eqref{rhodef} are
\begin{align*}
\delta \lp  \Lambda^{-1} \dot{\Lambda} \rp
& =
\frac{\partial \widehat \bsigma}{\partial t}
+ \lsb \Lambda^{-1}\dot \Lambda , \widehat \bsigma \rsb
\,,\\
\delta \lp  \Lambda^{-1} \Lambda' \rp
& =
\widehat \bsigma' + \lsb \Lambda^{-1} \Lambda' , \widehat \bsigma\rsb
\,,\\
\delta \lp  \Lambda^{-1}  \mathbf{r} \rp
& =
\bPsi - \widehat\bsigma\lp \Lambda^{-1}\br\rp
\,,\\
\delta \lp  \Lambda^{-1}  \mathbf{\dot{r}} \rp
& =
\dot\bPsi - \widehat\bsigma\lp \Lambda^{-1} \dot \br \rp + \lp\Lambda^{-1}\dot \Lambda\rp \bPsi
\,,\\
\delta \lp  \Lambda^{-1}  \mathbf{r'} \rp
& =
\bPsi' - \widehat\bsigma\lp \Lambda^{-1} \br' \rp + \lp\Lambda^{-1} \Lambda'\rp \bPsi
\,.
\end{align*}
\end{lem}

\begin{proof}  We calculate the variations directly, one by one.  First we have,
\begin{eqnarray*}
\delta \lp  \Lambda^{-1} \dot{\Lambda} \rp
& = &  - \Lambda^{-1}\de\Lambda \lp\Lambda^{-1}\dot{\Lambda}\rp + \Lambda^{-1}\de\dot\Lambda \\
& = &  - \Lambda^{-1}\de\Lambda \lp\Lambda^{-1}\dot{\Lambda}\rp + {\lp\Lambda^{-1}\de\Lambda\rp}^{.} + \lp\Lambda^{-1}\dot \Lambda\rp\lp\Lambda^{-1}\de \Lambda\rp\\
&=& \frac{\partial \widehat \bsigma}{\partial t}
+ \lsb \Lambda^{-1}\dot \Lambda , \widehat \bsigma\rsb.\\
\end{eqnarray*}
Similarly, for the variation of $\Lambda^{-1}\Lambda'$ we have,
\begin{displaymath}
\de\lp\Lambda^{-1}\Lambda'\rp = \widehat \bsigma' + \lsb \Lambda^{-1} \Lambda' , \widehat \bsigma\rsb
\,.
\end{displaymath}

Now we consider the variation of $\Lambda^{-1}\dot\br$, which is given by
\begin{eqnarray*}
\de\lp \Lambda^{-1}\dot\br\rp &=& -\Lambda^{-1} \de \Lambda \Lambda^{-1} \br + \Lambda^{-1}\de \dot \br\\
&=&  -\lp\Lambda^{-1} \de \Lambda\rp\lp \Lambda^{-1} \br\rp + \lp\Lambda^{-1}\de \br\rp^. + \lp\Lambda^{-1}\dot \Lambda\rp \lp\Lambda^{-1} \br\rp\\
&=& \dot\bPsi - \widehat\bsigma\lp \Lambda^{-1} \dot\br \rp + \lp\Lambda^{-1}\dot \Lambda\rp \bPsi
\,.\\
\end{eqnarray*}
A similar argument yields the variation of $\Lambda^{-1}\br'$,
\[
\de\lp \Lambda^{-1}\br'\rp = \bPsi' - \widehat\bsigma\lp \Lambda^{-1} \br' \rp + \lp\Lambda^{-1} \Lambda'\rp \bPsi
\,.
\]

Finally, the variation of $\Lambda^{-1}\br$ is given by,
\[
\de\lp\Lambda^{-1}\br\rp = -\Lambda^{-1}\de\Lambda \Lambda^{-1}\br + \Lambda^{-1}\de\br\\
= \bPsi - \widehat\bsigma\lp \Lambda^{-1}\br\rp
\]
and all the formulas in the statement are proved.
\end{proof}
\medskip

We may now use these identities to prove the Hamilton-Pontryagin Theorem \ref{HamPont-thm} for the equations of filament dynamics.
\medskip

\begin{proof}
The main results from this Hamilton's principle arise from the following identities, written in terms of the skew-symmetric $3\times3$ matrix $\Sigma= \Lambda^{-1}\delta \Lambda\in \mathfrak{so}(3)\simeq\mathbb{R}^3$ and the vector $\bpsi = \Lambda^{-1} \delta \mathbf{r}\in \mathbb{R}^3$.
Variations with respect to the Lagrange multipliers impose the expected defining relations for the five quantities $\{\brho,\bom,\bOm,\bgam,\bGam\}$. The conjugate variations give
\begin{equation*}
\frac{\delta l }{\delta \brho} - \mathbf{R}
 = 0
\,, \quad
\frac{\delta l }{\delta \bom} - \boldsymbol{\pi}
 = 0
\,, \quad
\frac{\delta l }{\delta \bOm} - \boldsymbol{\Pi}
 = 0
\,, \quad
\frac{\delta l }{\delta \bgam} - \boldsymbol{\mu}
 = 0
\,, \quad
\frac{\delta l }{\delta \bGam} - \mathbf{ M}
 = 0
 \,.
\end{equation*}
Finally, the variations proportional to $\bsigma$ and $\bpsi$ yield the filament equations
\begin{equation*}
\dot{\boldsymbol{\pi}} + \bom\times\boldsymbol{\pi} +\boldsymbol{\Pi}' + \bOm\times \boldsymbol{\Pi}
+ \bgam\times\boldsymbol{\mu}
+ \bGam\times \mathbf{M}
+ \brho\times\mathbf{R}
= 0
 \,,
\end{equation*}
and
\begin{equation*}
\dot{\boldsymbol{\mu}} + \bom\times \boldsymbol{\mu} + \mathbf{M}' + \bOm\times \mathbf{M} - \mathbf{R}
= 0
 \,,
\end{equation*}
respectively. 
\end{proof}

\begin{remark}{\rm 
The Hamilton-Pontryagin approach used here also allows nonholonomic constraints to be imposed on the motion of the strands, if one desires. See \cite{Ho2008} for a discussion of nonholonomic constraints using the Hamilton-Pontryagin approach.}
\end{remark}

\subsubsection{Nonlocal potential}

For the nonlocal potential \eqref{Energy1} we may form a Hamilton-Pontryagin variational principal in a similar fashion.  In this case, the action $S_{np}$ is given by
\begin{eqnarray*}
S_{np} &=& \int l_{np}(\xi, \bkappa, \bGam) dt + \iint\bm\cdot \lp \Lambda^{-1}(s)\br'(s) - \bGam\rp  \mbox{d}s \mbox{d}t\\
&&
+\iiint \Big( X \cdot \lp \Lambda^{-1}(s)\Lambda(s') - \xi\rp + \bK\cdot \lp \Lambda^{-1}(s)\lp \br(s') - \br(s)\rp - \bkappa\rp \Big)  \mbox{d}s  \mbox{d}s' \mbox{d}t
\end{eqnarray*}

\begin{lem} \label{lemvarnonlocal}
The additional variational formulas needed for calculating the equations of motion are given by
\begin{eqnarray*}
\Lambda^{-1}(s')\Lambda(s) \lp\delta \lp\Lambda^{-1}(s)\Lambda(s')
\rp\rp  
&=&
-{\rm Ad}_{\Lambda^{-1}(s')\Lambda(s)} \widehat{\bsigma} (s) + \widehat{\bsigma}(s')
\,,\\
\de \lp \Lambda^{-1}(s)\lp \br(s') - \br(s)\rp\rp 
&=& 
-\widehat{\bsigma}(s)\Lambda^{-1}(s)\lp \br(s') - \br(s)\rp 
\\&&\hspace{5mm}
+\
 \Lambda^{-1}(s)\Lambda(s')\bpsi(s') + \bpsi(s)
 \,,
\end{eqnarray*}
where the independent variations are defined by
\begin{equation}
\bpsi(s)=\Lambda^{-1}(s) \delta \br(s)
\quad\hbox{and} \quad
\widehat{\bsigma}(s)=\Lambda^{-1}(s) \delta \Lambda(s) \, .
\end{equation}
\end{lem}

\begin{proof}  The first variational formula is calculated directly, as
\begin{eqnarray*}
\Lambda^{-1}(s')\Lambda(s) \lp\delta \lp\Lambda^{-1}(s)\Lambda(s')\rp\rp &=& \Lambda^{-1}(s')\Lambda(s)\lp \Lambda^{-1}(s)\de\Lambda(s')\rp \\
&& \qquad - \Lambda^{-1}(s')\Lambda(s)\lp \Lambda^{-1}(s)\de\Lambda(s)\Lambda^{-1}(s)\Lambda(s')\rp\\
&=& -{\rm Ad}_{\Lambda^{-1}(s')\Lambda(s)} \widehat{\bsigma} (s) + \widehat{\bsigma}(s')
\,.
\end{eqnarray*}
The second variational formula follows similarly from a direct calculation,
\begin{eqnarray*}
\de \lp \Lambda^{-1}(s)\lp \br(s') - \br(s)\rp\rp
&=& - \Lambda^{-1}(s)\de \Lambda(s) \Lambda^{-1}(s)\lp \br(s') - \br(s)\rp
\\ && \qquad + \Lambda^{-1}(s)\lp \de\br(s') - \de\br(s)\rp\\
&=& -\widehat{\bsigma}(s)\Lambda^{-1}(s)\lp \br(s') - \br(s)\rp + \Lambda^{-1}(s)\Lambda(s')\bpsi(s') + \bpsi(s)
\end{eqnarray*}
which proves the lemma.
\end{proof}

\begin{thm} The equations that arise from the variational principle with the nonlocal action
\begin{eqnarray*}
S_{np} &=& \iiint U(\xi, \bkappa, \bGam)  \mbox{d}s  \mbox{d}s' \mbox{d}t + \iint\bm\cdot \lp \Lambda^{-1}(s)\br'(s) - \bGam\rp  \mbox{d}s \mbox{d}t\\
&& \quad
+\iiint \Big( X \cdot \lp \Lambda^{-1}(s)\Lambda(s') - \xi\rp + \bK\cdot \lp \Lambda^{-1}(s)\lp \br(s') - \br(s)\rp - \bkappa\rp \Big)  \mbox{d}s  \mbox{d}s' \mbox{d}t
\end{eqnarray*}
are given by:
\[
X = \frac{\partial U}{\partial \xi} 
\,, \qquad
\bK =  \frac{\partial U}{\partial \bkappa} 
\,, \qquad
\bm =  \frac{\partial U}{\partial \bGam} 
\,,
\]
and
\[
 \bGam \times \bm = \int \lp \xi(s,s') X(s',s) - X(s,s')\xi^{-1}(s,s')  +  \bK(s,s') \times \bkappa(s,s')\rp ds',
\]

\[
\bm' + \bOm \times \bm = \int \lp \xi(s,s')\bK(s',s) - \bK(s,s')\rp ds',
\]
together with the constraints,
\[
\xi = \Lambda^{-1}(s)\Lambda(s')
\,, \qquad
\bkappa = \Lambda^{-1}(s)\lp \br(s') - \br(s)\rp
\,, \qquad
\bGam = \Lambda^{-1}(s)\br'(s)
\,.
\]
\end{thm}

\begin{proof}  The proof is obtained by substituting the variations given in Lemma \ref{lemvarnonlocal} into the Hamilton's principle for the action in the statement of the theorem.  Variations in $X$, $\bK$ and $\bm$ yield the constraints,
\[
\xi = \Lambda^{-1}(s)\Lambda(s')
\,, \qquad
\bkappa = \Lambda^{-1}(s)\lp \br(s') - \br(s)\rp
\,, \qquad
\bGam = \Lambda^{-1}(s)\br'(s)
\,.
\]

Variations in $\xi$, $\bkappa$ and $\bGam$ yield the relationships
\[
X = \frac{\partial U}{\partial \xi} 
\,, \qquad
\bK =  \frac{\partial U}{\partial \bkappa} 
\,, \qquad
\bm =  \frac{\partial U}{\partial \bGam} 
\,.
\]

Finally, the variations proportional to $\widehat{\bsigma}(s)$ and $\bpsi(s)$ yield
\[
 \bGam \times \bm = \int \lp \xi(s,s') X(s',s) - X(s,s')\xi^{-1}(s,s')  +  \bK(s,s') \times \bkappa(s,s')\rp ds'
\]

and
\[
\bm' + \bOm \times \bm = \int \lp \xi(s,s')\bK(s',s) - \bK(s,s')\rp ds'
\,,
\]
respectively.
\end{proof}

\medskip

We may combine these nonlocal terms and the local part of the equations to produce the full set of equations. These are given by
\begin{align*}\label{hampont1}
\dot{\boldsymbol{\pi}} + \bom\times\boldsymbol{\pi} +\boldsymbol{\Pi}' + \bOm\times \boldsymbol{\Pi}
+ \bgam\times\boldsymbol{\mu}
+& \bGam\times \lp \mathbf{M} + \bm\rp
+ \brho\times\mathbf{R}\\
&= \int \lp \bK(s,s') \times \bkappa(s,s')
+ \bZ(s,s')\rp ds'
\end{align*}
and
\begin{equation*}\label{hampont2}
\dot{\boldsymbol{\mu}} + \bom\times \boldsymbol{\mu} + \lp\mathbf{M}+\bm\rp' + \bOm\times \lp\mathbf{M}+\bm\rp - \mathbf{R} = \int \lp \xi(s,s')\bK(s',s)- \bK(s,s')\rp ds'
 \,,
\end{equation*}
where one defines
\begin{equation}
\widehat\bZ(s,s') :=  \xi(s,s') X(s',s) -X(s,s')\xi^{-1}(s,s')
\,,
\label{Zdef1}
\end{equation}
denoted as $\widehat\bZ$ since the right hand side of this equation is in $\mathfrak{so}(3)$.

We may now use these functional-derivative relations to express the equations of motion in terms of the reduced Lagrangian, $l=l_{loc}+l_{np}$.  The functional-derivative relations obtained in the Hamilton-Pontryagin approach are
\begin{eqnarray*}
\mathbf{R} = \frac{\delta l_{loc} }{\delta \brho} \,, &\qquad&
\boldsymbol{\pi} = \frac{\delta l_{loc} }{\delta \bom} \,,\\
\boldsymbol{\Pi} = \frac{\delta l_{loc} }{\delta \bOm} \,, &\qquad&
\boldsymbol{\mu} = \frac{\delta l_{loc} }{\delta \bgam}\,,\\
\mathbf{ M} = \frac{\delta l_{loc} }{\delta \bGam} \,, &\qquad&
X = \dede{l_{np}}{\xi} \,,\\
\bK = \dede{l_{np}}{\bkappa} \,, &\qquad&
\bm = \dede{\big( l_{loc}+l_{np}\big)}{\bGam}.
\end{eqnarray*}

Substituting these relations into the equations of motion above gives the following equations of motion for the charged strand.

\begin{align}\label{hpvarsig}
\lp \prt_t + \bom\times\rp\dede{l_{loc}}{\bom} + \lp\prt_s + \bOm\times\rp&\dede{l_{loc}}{\bOm}= \dede{l_{loc}}{\bgam}\times\bgam + \dede{\lp l_{loc}+l_{np}\rp }{\bGam}\times\bGam + \dede{l_{loc}}{\brho}\times\brho \nonumber\\
&\quad +\int \lp \frac{\partial U}{\partial \bkappa}(s,s')\times\bkappa(s,s') +  \bZ(s,s')\rp ds'\,,
\end{align}
\begin{align}\label{hpvarpsi}
\lp \prt_t + \bom\times\rp\dede{l_{loc} }{\bgam} &+ \lp\prt_s + \bOm\times\rp\dede{\lp l_{loc}+l_{np} \rp }{\bGam}\nonumber \\
&\quad= \dede{l_{loc}}{\brho} + \int\lp \xi(s,s')\frac{\partial U}{\partial \bkappa}(s',s) -  \frac{\partial U}{\partial \bkappa}(s,s')\rp ds'.
\end{align}

The term $\widehat\bZ(s,s')$ is the contribution from the nonlocal part of the Lagrangian that we have sought. 
\begin{remark}$\quad${\rm 

\begin{itemize}
\item
The dynamical equations \eqref{hpvarsig} and \eqref{hpvarpsi} must be augmented by the advection conditions  \eqref{rhokin}, \eqref{kincond} and \eqref{kincondom} in order to close the system.
\item
The resulting system of equations  describes an elastic filament with  \emph{two nonlocal} additional components, or degrees of freedom, compared to the ordinary Kirchhoff filament (to which the system reduces, when $\brho$ and $\xi$  are absent).
\item
The two additional (nonlocal) degrees of freedom in $\brho$ and $\bXi$ (with $\boldsymbol{\widehat{\Xi}}=\xi^{-1}{\partial \xi/\partial s}$) will produce an important effect that will distinguish the behavior of this system from that of the ordinary Kirchhoff filament. Namely, the presence of the two additional equations for $\brho$ and $\bXi$ \emph{raises the order}  of the equation set. In turn, the increase in differential order of the system will produce additional modes of excitation for the waves that will propagate along the filament when the system is linearized around the static solutions.
\end{itemize}
}
\end{remark}

\paragraph{Summary.} Equations \eqref{hpvarsig}, \eqref{hpvarpsi}, \eqref{rhokin}, \eqref{kincond} and \eqref{kincondom} represent the generalization of the Kirchhoff model that we have sought. As we shall see in the next section, under a certain transformation of variables this model reduces to a conservation law formulated in terms of coadjoint action on $\mse(3)$ Lie algebras.

\subsection{A modified Euler-Poincar\'e approach}
\label{sec:Euler-Poincare}

The Euler-Poincar\'e approach is based on applying Hamilton's variational principle to the symmetry-reduced Lagrangian and constraining the variations properly. While this will be yet another way of deriving equations \eqref{hpvarsig}, \eqref{hpvarpsi}, we believe that such a ``bare hands'' derivation will benefit understanding, as it represents a direct and explicit derivation of those equations of motion. See \cite{MaRa2002} and \cite{Ho2008} for an introduction to the
classical Euler-Poincar\'e approach. Some calculations in this section overlap with those in Section~\ref{sec:Hamilton-Pontryagin}. Nonetheless, we have chosen to present them here for completeness of exposition.

\subsubsection{Variations: Definitions}\label{var-defs}
Let us compute variations of $\brho$, $\bom$, $\bgam$, $\bOm$ and $\bGam$. We proceed by first computing,
\begin{equation}
\delta \brho =
- \Lambda^{-1} \delta \Lambda \Lambda^{-1} \br
+ \Lambda^{-1} \delta \br
=
- \Sigma \brho +\bpsi
=
- \bsigma \times \brho + \bpsi
=  \brho \times \bsigma  + \bpsi
\,,
\label{rhovar}
\end{equation}
where we have defined the variational quantities
\begin{eqnarray}
\Sigma= \Lambda^{-1} \delta \Lambda
\,,
\label{sigmadef}
\\
\bpsi=\Lambda^{-1} \delta \br
\,.
\label{psidef}
\end{eqnarray}
Next, we compute the space and time derivatives of $\Sigma$ and $\bpsi$ along the curve.
We have the space derivative,
\begin{equation}
\frac{\partial \bpsi}{ \partial s} =
 - \Lambda^{-1} \Lambda'  \Lambda^{-1} \delta \br + \Lambda^{-1} \delta  \br '
 = - \Omega \bpsi +  \Lambda^{-1} \delta  \br ' =
 - \bOm \times \bpsi +  \Lambda^{-1} \delta  \br ' \,,
 \label{psisderiv}
\end{equation}
and the time derivative,
\begin{equation}
\frac{\partial \bpsi}{ \partial t} =
 - \Lambda^{-1} \dot{\Lambda}  \Lambda^{-1} \delta \br + \Lambda^{-1} \delta  \dot{\br}
 = - \omega \bpsi +  \Lambda^{-1} \delta  \dot{\br} =
 - \bom \times \bpsi +  \Lambda^{-1} \delta  \dot{\br} \,  .
 \label{psitderiv}
\end{equation}
Analogously, for the space derivative of $\Sigma$,
\begin{equation}
\frac{\partial \Sigma}{ \partial s} =
 - \Lambda^{-1} \Lambda'  \Lambda^{-1} \delta \Lambda + \Lambda^{-1} \delta  \Lambda'
 =
  - \Omega \Sigma  + \Lambda^{-1} \delta  \Lambda ' \,,
 \label{sigmasderiv}
\end{equation}
while the time derivative of $\Sigma$ is computed as follows:
\begin{equation}
\frac{\partial \Sigma}{ \partial t} =
 - \Lambda^{-1} \dot{\Lambda}  \Lambda^{-1} \delta \Lambda + \Lambda^{-1} \delta  \Lambda'
 =
  - \omega \Sigma  + \Lambda^{-1} \delta \dot{ \Lambda }
  \,.
 \label{sigmatderiv}
\end{equation}

Now we are ready to compute the variations $\delta \bgam$, $\delta \bGam$, $\delta \omega$ and  $\delta \Omega$. The first of these is
\[
\delta \bgam =
  - \Lambda^{-1} \delta \Lambda \Lambda^{-1} \dot{\brho} +
\underbrace{\Lambda^{-1} \delta \dot{\brho} }
_{\mbox{use  (\protect{\ref{psitderiv}})}}
=
 -\Sigma \bgam + \omega \bpsi +\frac{\partial \bpsi}{\partial t}  \, ,
\]
so in vector form,
\begin{equation}
\delta \bgam =
 \bgam \times \bsigma + \bom \times \bpsi +\frac{\partial \bpsi}{\partial t}
 \, .
 \label{gammavar}
\end{equation}
Likewise,
\[
\delta \bGam =
  - \Lambda^{-1} \delta \Lambda \Lambda^{-1} \brho' +
\underbrace{\Lambda^{-1} \delta \rho' }
_{\mbox{use  (\protect{\ref{psisderiv}})}}
=
 -\Sigma \bGam + \Omega \bpsi +\frac{\partial \bpsi}{\partial s}  \, ,
\]
which has the vector form,
\begin{equation}
\delta \bGam =
 \bGam \times \bsigma + \bOm \times \bpsi +\frac{\partial \bpsi}{\partial s}
 \, .
 \label{Gammavar}
\end{equation}
Next,
\[
\delta \omega =
  - \Lambda^{-1} \delta \Lambda \Lambda^{-1} \dot{\Lambda} +
\underbrace{\
\Lambda^{-1} \delta \dot{\Lambda}\
 }_{\mbox{use  (\protect{\ref{sigmatderiv}})}} 
= 
 -\Sigma \omega + \omega \Sigma +\frac{\partial \Sigma}{\partial t} 
 =
 [\omega, \Sigma]+ \frac{\partial \Sigma}{\partial t}
 \, ,
\]
so expressing these formulas in terms of vectors yields
\begin{equation}
\delta \bom =
 \bom  \times \bsigma
 +
 \frac{\partial \bsigma}{\partial t}
 \, .
 \label{omegavar}
\end{equation}
Finally,
\[
\delta \Omega =
  - \Lambda^{-1} \delta \Lambda \Lambda^{-1} \Lambda ' +
\underbrace{\
\Lambda^{-1} \delta \Lambda' \
}_{\mbox{use  (\protect{\ref{sigmasderiv}})}}
=
 -\Sigma \Omega + \Omega \Sigma +\frac{\partial \Sigma}{\partial s} =
 [\Omega, \Sigma]+ \frac{\partial \Sigma}{\partial s} \, ,
\]
so, again, expressing in terms of vectors leads to
\begin{equation}
\delta \bOm =
 \bOm  \times \Sigma +\frac{\partial \Sigma}{\partial s}
 \, .
 \label{Omegavar}
\end{equation}
Finally, the variation of $\xi(s,s')$ is given by
\begin{equation}
\lp \xi^{-1} \delta \xi(s,s') \rp  =-{\rm Ad}_{\xi^{-1}(s,s')} \Sigma (s) + \Sigma(s') \, .
\label{xivar3}
\end{equation}

\subsubsection{Derivation of the equations of motion}
Suppose now we want to compute variations of the reduced energy Lagrangian $l$ which is a functional of
$(\brho, \bgam, \bGam, \omega, \Omega)$.
From \eqref{LiePoissondist} we see that
\begin{align}
d_{k,m}(s,s')
=&
\big|
\Lambda^{-1} (s) \br(s,t) - \Lambda^{-1}(s,t) \br(s',t)+   \boldeta_k(s) -\xi(s,s') \boldeta_m (s')
\big|
\nonumber
\\
=&
\big|
\bkappa(s, s')+   \boldeta_k(s) -\xi(s,s') \boldeta_m(s')
\big|
\, ,
\label{dist2}
\end{align}
where we have defined
\begin{equation}
\bkappa(s,s') = \Lambda^{-1} (s)\lp  \br(s') - \br(s) \rp
=  \xi(s,s') \brho(s')- \brho(s) 
\label{kappadef}
\, .
\end{equation}
The variation of $\bkappa$ is then given by
\begin{align}
\delta \bkappa(s,s') =
& \Sigma(s)  \bkappa(s,s')-\bpsi(s) 
+ \xi(s,s') \bpsi(s')
\nonumber
\\
= & \, \bsigma(s)  \times \bkappa(s,s')
-\bpsi(s)+\xi(s,s') \bpsi(s') \, .
\label{deltakappa}
\end{align}

Let us first define the Lagrangian $l$ as the sum of a `local' part $l_{loc}$ and a nonlocal part $l_{np}$, according to
\begin{align}
l(\bom, &\bgam, \bOm,  \bGam, \brho, \xi, \bkappa):=
l_{loc}+l_{np}
\nonumber
\\
& =l_{loc}(\bom, \bgam, \bOm,  \bGam, \brho)
+ \iint
U \lp
\bkappa(s,s'), \xi(s,s'), \bGam(s), \bGam(s') 
\rp
\mbox{d} s
\mbox{d} s'
\, .
\label{totalL}
\end{align}

\emph{Note.} From now on, we assume that the nonlocal part of the potential energy $U$  is a function of the two variables $\bkappa(s,s')$ and $\xi(s,s')$, as well as $\bGam$, since $s$ is not necessarily the arc length. In particular, for a potential energy depending on the distance $d_{k,m}$, the variables  $\bkappa$ and $\xi$ enter in the linear combination defined by \eqref{dist2}. In principle, the potential energy could have chosen to be an arbitrary functional of $\Lambda^{-1}(s)  \br(s)$, $\Lambda^{-1}(s) \br(s')$ and $\xi(s,s')$. Euler-Poincar\'e methods would be directly applicable to these functionals as well.
\smallskip

The equations of motion are computed from the stationary action principle $\delta S=0$, with $S=\int l\,dt$ and $l=l_{loc}+l_{np}$ in equation \eqref{totalL}. We have
\begin{align}
\delta S= & \int
\left<
\frac{\delta l_{loc}}{\delta \brho}
\, , \,
\delta \brho
\right>
+
\left<
\frac{\delta l_{loc}}{\delta \bgam}
\, , \,
\delta \bgam
\right>
+
\left<
\frac{\delta ( l_{loc}+l_{np}) }{\delta \bGam}
\, , \,
\delta \bGam
\right>
+
\left<
\frac{\delta l_{loc}}{\delta \omega}
\, , \,
\delta \omega
\right>
\\
&
+
\left<
\frac{\delta l_{loc}}{\delta \Omega}
\, , \,
\delta \Omega
\right>
\nonumber
+
\left<
\frac{\delta l_{np}}{\delta \bkappa}
\, ,
\delta \bkappa
\right>
+
\left<
\xi^{-1} \frac{\delta l_{np}}{\delta \xi}
\, ,
\xi^{-1} \delta \xi
\right> \mbox{d} t =0 
\, ,
\label{deltal0}
\end{align}
where $\langle\,\cdot\,,\,\cdot\, \rangle =  \int (\,\cdot\,,\,\cdot)_{\mR^3} ds $ represents $L^2$ pairing in the filament variable $s$.
We may now substitute $\delta \brho$ from \eqref{rhovar}, $\delta \bgam$ from \eqref{gammavar} and $\delta \bOm$ from \eqref{Omegavar}.
We have
\begin{equation}
\left<
\frac{\delta l_{loc}}{\delta \brho}
\, , \,
\delta \brho
\right>
=
\left<
\frac{\delta l_{loc}}{\delta \brho}
\, , \,
\brho \times \bsigma  +\bpsi
\right>
=
\left<
\frac{\delta l_{loc}}{\delta \brho} \times \brho
\, , \,
\bsigma
\right>
+
\left<
\frac{\delta l_{loc}}{\delta \brho}
\, , \,
\bpsi
\right>
\, .
\label{dedrho}
\end{equation}

For $\delta \bkappa$  we obtain
\begin{align}
\left<
\frac{\delta l_{np}}{\delta \bkappa}
\, , \,
\delta \bkappa
\right> = &
\int \Big<
 \int \frac{\partial U}{\partial \bkappa} (s,s') \times \bkappa(s,s')  \mbox{d} s'
\, , \,
\bsigma(s)
\Big>
\nonumber
\\
& +
\Big<
 \int
 \Big(
\xi(s,s') \frac{\partial U}{\partial \bkappa} (s',s)-\frac{\partial U}{\partial \bkappa} (s,s') 
\Big)
\mbox{d} s'
\, , \,
\bpsi(s)
\Big>
\, .
\label{dlnpdrhovar}
\end{align}
Next,
\begin{align}
&\Bigg<
\frac{\delta \big( l_{loc}+l_{np} \big) }{\delta \bGam}
\, , \,
\delta \bGam
\Bigg>
=
\left<
\frac{\delta\big( l_{loc}+l_{np} \big)}{\delta \bGam}
\, , \,
\bGam \times \bsigma + \bOm \times \bpsi +\frac{\partial \bpsi}{\partial s}
\right>
\nonumber
\\
&
=
\left<
\frac{\delta \big( l_{loc}+l_{np} \big)}{\delta \bGam} \times \bGam
\, , \,
\bsigma
\right>
+
\left<
\frac{\delta \big( l_{loc}+l_{np} \big)}{\delta \bGam} \times \bOm
-\frac{\partial}{\partial s} \frac{\delta \big(l_{loc}+l_{np} \big) }{\delta \bGam}
\, , \,
\bpsi
\right>
\, ,
\label{dedGamma}
\end{align}
and
\begin{align}
\left<
\frac{\delta l_{loc}}{\delta \bgam}
\, , \,
\delta \bgam
\right>
=
&
\left<
\frac{\delta l_{loc}}{\delta \bgam}
\, , \,
 \bgam \times \bsigma + \bom \times \bpsi +\frac{\partial \bpsi}{\partial t}
\right>
\nonumber
\\
 = &
\left<
\frac{\delta l_{loc}}{\delta \bgam} \times \bgam
\, , \,
\bsigma
\right>
+
\left<
\frac{\delta l_{loc}}{\delta \bgam} \times \bom
-\frac{\partial}{\partial t} \frac{\delta l_{loc}}{\delta \bgam}
\, , \,
\bpsi
\right>
\, .
\label{dedgamma}
\end{align}

Variations in $\bom$ and $\bOm$ give, respectively, after integrating by parts,
\begin{align}
\int
\left<
\frac{\delta l_{loc}}{\delta \bom}
\, , \,
\delta \bom
\right>
dt
&=
\int
\left<
\frac{\delta l_{loc}}{\delta \bom}
\, , \,
\bom \times \bsigma  +\frac{\partial \bsigma}{\partial t}
\right> 
dt  \nonumber \\
&=
\int
\left<
\frac{\delta l_{loc}}{\delta \bom} \times \bom
-\frac{\partial}{\partial t} \frac{\delta l_{loc}}{\delta \bom}
\, , \,
\bsigma
\right>
dt
\, ,
\label{dedom}
\end{align}
and
\begin{equation}
\left<
\frac{\delta l_{loc}}{\delta \bOm}
\, , \,
\delta \bOm
\right>
=
\left<
\frac{\delta l_{loc}}{\delta \bOm}
\, , \,
\bOm \times \bsigma  +\frac{\partial \bsigma}{\partial s}
\right>
=
\left<
\frac{\delta l_{loc}}{\delta \bOm} \times \bOm
-\frac{\partial}{\partial s} \frac{\delta l_{loc}}{\delta \bOm}
\, , \,
\bsigma
\right>
\, .
\label{dedOm}
\end{equation}
Finally, one computes the variations in $\xi$ as follows:
\begin{align}
&\int
\left<
\xi^{-1} \frac{\delta l_{np}}{\delta \xi}
\, , \, 
\xi^{-1} \delta \xi
\right>
\mbox{d} s' \nonumber \\
& \qquad \qquad =
\int
\left< \xi^{-1}(s,s') \frac{\partial U}{\partial \xi} (s,s')
\, ,
-{\rm Ad}_{\xi^{-1}(s,s')} \Sigma (s) + \Sigma(s')
\right>_{\mso(3)}
\mbox{d} s'
\,,
\label{dldxi}
\end{align}
where $\left\langle\cdot , \cdot \right\rangle_{\mso(3)}: 
\mso(3)^\ast \times \mso(3)\rightarrow \mathbb{R}$ is the real-valued pairing between the Lie algebra $\mso(3)$ and its dual $\mso(3)^\ast$.

Substitution of \eqref{dedrho},\eqref{dedgamma}, and \eqref{dedom}
gives an expression for $\delta S$ that is linear in $\bsigma$ and $\bpsi$. Collecting those terms when imposing $\delta S =0$ implies from the term proportional to $\bsigma$ that:
\begin{align}
\lp \frac{\partial}{\partial t} \frac{\delta l_{loc}}{\delta \bom} \right.
&+
\left.
\bom \times \frac{\delta l_{loc}}{\delta \bom}
\rp
+
\lp
\frac{\partial}{\partial s} \frac{\delta l_{loc}}{\delta \bOm}
+
\bOm \times \frac{\delta l_{loc}}{\delta \bOm}
\rp
= \frac{\delta l_{loc}}{\delta \bgam} \times \bgam
+
\frac{\delta \lp l_{loc}+l_{np} \rp }{\delta \bGam} \times \bGam
\nonumber \\
&
+
\frac{\delta l_{loc}}{\delta \brho} \times \brho
+\int \left( \frac{\partial U}{\partial \bkappa} (s,s') \times \bkappa (s,s')
+
 \mathbf{Z}(s,s') \right)   
\mbox{d} s'
\,,
\label{EPsigma}
\end{align}
where the term $\bZ(s,s')$ is the vector given by
\begin{equation}
\widehat\bZ(s,s') = \xi(s,s')
\lp
\frac{\partial U}{\partial \xi} (s,s')
\rp ^T
-\frac{\partial U}{\partial \xi} (s,s')
\xi^T(s,s')
\,,
\label{Zdef}
\end{equation}
which is the same quantity that we found using the Hamilton-Pontryagin approach. \smallskip

Formula \eqref{Zdef} is computed from the variation
in \eqref{dldxi} as follows

\begin{align}
&\iint
\left< \xi^{-1}(s,s') \frac{\partial U}{\partial \xi} (s,s')
\, ,
-{\rm Ad}_{\xi^{-1}(s,s')} \Sigma (s) + \Sigma(s')
\right>_{\mso(3)}
\mbox{d} s \mbox{d} s'
\nonumber
\\
&=
\iint
 \left< -
{\rm Ad}^*_{\xi^{-1}(s,s')}  \xi^T(s,s') \frac{\partial U}{\partial \xi} (s,s')
+
 \xi^T(s',s) \frac{\partial U}{\partial \xi} (s',s)
\, ,
\Sigma (s)
\right>_{\mso(3)}
\mbox{d} s \mbox{d} s'
\nonumber
\\
&=
\iint
 \left< -
\xi(s,s')  \xi^T(s,s') \frac{\partial U}{\partial \xi} (s,s') \xi^T(s,s')
+
 \xi(s,s') \left( \frac{\partial U}{\partial \xi} (s,s') \right)^T
, \,
\Sigma (s)
\right>_{\mso(3)}
\mbox{d} s \mbox{d} s'
\nonumber
\\
&= \iint \left< -
\frac{\partial U}{\partial \xi} (s,s') \xi^T(s,s')
+
 \xi(s,s') \left( \frac{\partial U}{\partial \xi} (s,s') \right)^T, \,
\Sigma (s) \right>_{\mso(3)}
\mbox{d} s \mbox{d} s' \, .
\end{align}
Here, we have used the fact that $\xi^T(s,s')=\xi^{-1}(s,s')$, and
$\xi(s',s)=\xi^{-1}(s',s)$.

\smallskip
Next, we collect the terms proportional to $\bpsi$ in order to close the system.   We find
\begin{align}
\lp \frac{\partial}{\partial t} \frac{\delta l_{loc}}{\delta \bgam}
\right.
+&
\left.
\bom \times \frac{\delta l_{loc}}{\delta \bgam}
\rp
+
\left(
\frac{\partial}{\partial s} \frac{\delta \lp  l_{loc}+l_{np} \rp }{\delta \bGam}
+
\bOm \times \frac{\delta \lp  l_{loc}+l_{np} \rp }{\delta \bGam}
\right)
\nonumber
\\
&
=\frac{\delta l_{loc}}{\delta \brho}
+
\int \left(
\xi(s,s') \frac{\partial U}{\partial \bkappa} (s',s)
-\frac{\partial U}{\partial \bkappa} (s,s')
\right) \mbox{d} s' \,.
\label{EPpsi}
\end{align}
\begin{remark}{\rm 
Equations \eqref{EPsigma} and \eqref{EPpsi} obtained by the Euler-Poincar\'e approach recover equations \eqref{hpvarsig} and \eqref{hpvarpsi}, respectively, from the Hamilton-Pontryagin approach.}
\end{remark}

\section{Conservation laws}
\label{sec:Conservation}

In order to elucidate the physical meaning of the somewhat complex-looking equations \eqref{hpvarsig} and \eqref{hpvarpsi}, we shall write them explicitly as conservation laws.  For this purpose, we invoke the following identities valid for any Lie group $G$. Given a smooth curve $g(t)\in G$, $\eta\in\mathfrak{g}$, and $\mu\in\mathfrak{g}^*$, we have
\begin{align}
&{\rm Ad}_{g^{-1}(t)} \frac{\partial }{\partial t} {\rm Ad}_{g(t)} \eta={\rm ad}_{\sigma(t)}\eta
\,,
\label{diffAd}\\
&{\rm Ad}^*_{g(t)} \frac{\partial}{\partial t}
{\rm Ad}^*_{g^{-1}(t)}\mu=- {\rm ad}^*_{\sigma(t)} \mu
\,,
\label{diffAdstar}
\end{align}
where $\sigma(t)=g^{-1}\dot g(t)\in\mathfrak{g}$ and $\operatorname{Ad}^*$ denotes the coadjoint action of $G$ on $\mathfrak{g}^*$ defined by
$\langle {\rm Ad}^*_g\mu\,,\,\eta\rangle:=\langle \mu\,,\,{\rm Ad}_g\eta\rangle$. Formula \eqref{diffAdstar} generalizes to a curve $\mu(t)$ as
\begin{equation}\label{diffAd2}
{\rm Ad}^*_{g(t)} \frac{\partial}{\partial t}
{\rm Ad}^*_{g^{-1}(t)}\mu(t)=\dot\mu(t)- {\rm ad}^*_{\sigma(t)} \mu(t).
\end{equation}

To derive the conservation form of equations  \eqref{hpvarsig} and \eqref{hpvarpsi}
we need to consider the group $G=SE(3)$ whose elements are denoted by $g=(\Lambda , \br )$. Consider the function $(\Lambda(s,t), \br(s,t) )$ defined on spacetime. Then we have
\begin{equation}
\sigma = (\Lambda,\br)^{-1}(\dot{\Lambda}, \dot {\br})=( \Lambda^{-1} \dot{\Lambda}, \Lambda^{-1} \dot {\br}) = (\bom, \bgam) \, .
\label{mut}
\end{equation}
Recall that the infinitesimal coadjoint action on $\mathfrak{se}(3)^*$ is
\begin{equation}\label{coad_se3}
{\rm ad}^*_{(\bom, \bgam)}(\bmu,\bbeta)=-(\bom\times\bmu+\bgam\times\bbeta,\bom\times\bbeta)
\,.
\end{equation}
Then, using equations \eqref{diffAd2} and \eqref{coad_se3} for the temporal dual Lie algebra elements $(\bmu,\bbeta)=\lp  \d l/ \d \bom\, , \, \d l/\d \bgam \rp $ yields
\begin{align}
&{\rm Ad}^*_{(\Lambda,\br)}
\frac{\partial}{\partial t} \lsb
 {\rm Ad}^*_{(\Lambda,\br)^{-1}}
 \lp
 \frac{\delta l_{loc}}{\delta \bom} \, , \,  \frac{\delta l_{loc}}{\delta \bgam}
 \rp
  \rsb \nonumber\\
&\qquad=\frac{\partial}{\partial t}
\lp
\frac{\delta l_{loc}}{\delta \bom} \, , \,  \frac{\delta l_{loc}}{\delta \bgam}
\rp
+
\lp
\bom \times \frac{\delta l_{loc}}{\delta \bom} +
\bgam \times  \frac{\delta l_{loc}}{\delta \bgam}
\, , \, \bom \times  \frac{\delta l_{loc}}{\delta \bgam}
\rp
\, .\label{timeAd}
\end{align}

For the derivative with respect to curve parametrization $s$, we need to remember that the nonlocal part of the potential depends on $\bGam$ as well.
Thus, we have
\begin{align}
&{\rm Ad}^*_{(\Lambda,\br)}\frac{\partial}{\partial s} \lsb
 {\rm Ad}^*_{(\Lambda,\br)^{-1}}
 \lp
 \frac{\delta l_{loc}}{\delta \bOm} \, , \,  \frac{\delta ( l_{loc}+l_{np})}{\delta \bGam}
 \rp
  \rsb\nonumber\\
&\qquad=\frac{\partial}{\partial s}
\lp
\frac{\delta l_{loc}}{\delta \bOm} \, , \,  \frac{\delta l_{loc}}{\delta \bGam}
\rp
+ 
\lp
\bOm \times \frac{\delta l_{loc}}{\delta \bOm} +
\bGam \times  \frac{\delta (l_{loc}+l_{np}) }{\delta \bGam}
\, , \,
 \bOm \times  \frac{\delta l_{loc}}{\delta \bGam}
\rp
\, .
\label{spaceAd}
\end{align}
Some additional identities derived below will be needed in treating the nonlocal part of the potential. 

First we deal with the nonlocal term by referring to equation \eqref{Zdef1}. This can be expressed as a formal derivative of the nonlocal part of the potential with respect to Lie algebra elements $\bOm$ and $\bGam$ as follows.
Note that there are only \emph{two} free variations $
\boldsymbol{\widehat{\Sigma}} = \Lambda^{-1} \delta \Lambda$ and $\bpsi=\Lambda^{-1} \delta \br$.
On the other hand, the nonlocal part of the Lagrangian depends on three variables
$\brho, \xi$, and $\bGam$. Thus, there must be a relation between the partial derivatives of the nonlocal part of the Lagrangian and the total derivatives with respect to $\bGam$ and $\bOm$. This relation is computed as follows.

Upon identifying coefficients of the free variations  $\bsigma\times=\Lambda^{-1} \delta \Lambda$ and $\bpsi= \Lambda^{-1} \delta\br$, the following identity relates different variational derivatives of the nonlocal potential $l_{np}$:
\begin{align}
\delta l_{np} =
\left\langle
\xi^{-1} \frac{\delta l_{np} }{\delta \xi}
\, , \,
\xi^{-1} \delta \xi
\right\rangle
\
+\
&
\left\langle
\frac{\delta l_{np} }{\delta \bkappa}
\, , \,
\delta \bkappa
\right\rangle
\
+\
\left\langle
\frac{\delta l_{np} }{\delta \bGam}
\, , \,
\delta \bGam
\right\rangle
\nonumber
\\
&
=
\left\langle
\left.
\frac{\delta l_{np} }{\delta \bGam}
\right|_{Tot}
\, , \,
\delta \bGam
\right\rangle
\
+\
\left\langle
\frac{\delta l_{np} }{\delta \bOm}
\, , \,
\delta \bOm
\right\rangle
\, .
\label{changevar}
\end{align}
We will discuss this point in detail in \S\ref{sec:recovering_EP}.
Here, the subscript on $ (\,\cdot\,) |_{Tot}$ denotes the \emph{total}  derivative with respect to $\bGam$.
Using expressions \eqref{xivar3} for $\xi^{-1} \delta \xi$,
\eqref{deltakappa} for $\delta \bkappa$, \eqref{Omegavar} for $\delta \bOm$ and
\eqref{gammavar} for $\delta \bGam$,  then collecting terms proportional to the free variation $\bsigma$
yields the following identity, which implicitly defines $\delta l_{np} / \delta \bOm$  in terms of known quantities,
\begin{align}
-\frac{\partial}{\partial s} \frac{\delta  l_{np}}{\delta \bOm}
 &- \bOm \times \frac{\delta  l_{np}}{\delta \bOm}
 =
\label{dlnpdOm}
 \\
& \int \frac{\partial U}{\partial \bkappa} (s,s') \times \bkappa (s,s')
 \,\mbox{d} s'
+
\int \mathbf{Z}(s,s')  \,\mbox{d} s'
\,,
\nonumber
\end{align}
where we have defined $\mathbf{Z}(s,s')$ according to \eqref{Zdef1}. Likewise, identifying terms multiplying $\bpsi$ gives
\begin{align}
-\frac{\partial}{\partial s}
\left. \frac{\delta  l_{np}}{\delta \bGam}
\right|_{Tot}
&- \bOm \times
\left. \frac{\delta  l_{np}}{\delta \bGam}
\right|_{Tot} =
-\frac{\partial}{\partial s}
 \frac{\delta  l_{np}}{\delta \bGam}
- \bOm \times
\frac{\delta  l_{np}}{\delta \bGam}
\label{dlnpdGam}
\\
& + \int \frac{\partial U}{\partial \bkappa} (s,s')
-
\xi(s,s') \frac{\partial U}{\partial \bkappa} (s',s)
 \,\mbox{d} s'
\,.
 \nonumber
\end{align}

Therefore, we conclude that equations \eqref{hpvarsig}, \eqref{hpvarpsi} are equivalent to the following equations expressed on $\mse^*(3)$ in conservative form using variations of the total Lagrangian, $l:=l_{loc}+l_{np}$:
\begin{eqnarray}
\frac{\partial}{\partial t}
\lsb
 {\rm Ad}^*_{(\Lambda,\br)^{-1}}
 \lp
 \frac{\delta l}{\delta \bom} \, , \,  \frac{\delta l}{\delta \bgam}
 \rp
  \rsb
 &+&
\frac{\partial}{\partial s} \lsb
 {\rm Ad}^*_{(\Lambda,\br)^{-1}}
 \lp
 \frac{\delta l}{\delta \bOm} \, , \,
 \left.  \frac{\delta l}{\delta \bGam} \right|_{Tot}
\, \rp
  \rsb
  \nonumber\\
  &=&
  {\rm Ad}^*_{(\Lambda,\br)^{-1}}
 \lp
  \frac{\delta l}{\delta \brho} \times \brho
  \, , \,
  \frac{\delta l}{\delta \brho}
  \rp
  \,.
\label{consgen}
\end{eqnarray}
Here, the components of 
\[
 {\rm Ad}^*_{(\Lambda,\br)^{-1}}
 \lp
 \frac{\delta l}{\delta \bom} \, , \,  \frac{\delta l}{\delta \bgam}
 \rp
\]
represent, respectively, the spatial angular momentum density and the spatial linear momentum density of the strand, whose center of mass lies along its centerline. The components of
\[
 {\rm Ad}^*_{(\Lambda,\br)^{-1}}
 \lp
  \frac{\delta l}{\delta \brho} \times \brho
  \, , \,
  \frac{\delta l}{\delta \brho}
  \rp = \lp 0, \Lambda   \frac{\delta l}{\delta \brho} \rp 
\]
are the external torques and forces. (See \eqref{Tffin} for the last simplification.) As mentioned above, only external forces arising from potentials are considered in this paper.  In principle, more general non-conservative forces and torques can be considered as well, but we shall leave this question for further studies.

\begin{remark} {\rm
For future reference, it is advantageous to write out the conservation law (\ref{consgen}) in convective form as 
\begin{equation}\label{Final_Euler_Poincare_equations}
\left\lbrace\begin{array}{l}
\displaystyle\lp \prt_t + \bom\times\rp\dede{l}{\bom} + \lp\prt_s + \bOm\times\rp\dede{l}{\bOm} +\brho\times\dede{l}{\brho}+\bGam\times\dede{l}{\bGam}+\bgam\times\dede{l}{\bgam}=0,\\
\displaystyle\lp \prt_t + \bom\times\rp\dede{l}{\bgam} + \lp\prt_s + \bOm\times\rp\dede{l}{\bGam}
 -\dede{l}{\brho}=0
 \,.
\end{array}\right.
\end{equation}
Here we have defined the total Lagrangian $l:=l_{loc}+l_{np}$, and all the variational derivatives are assumed to be the \emph{total} derivatives. 
Note that in these equations coincide precisely with the equations for the  purely elastic filaments derived in \cite{SiMaKr1988}.

We note that the variations with respect to $\bOm$ and $\bGam$ are computed implicitly in (\ref{dlnpdOm}, \ref{dlnpdGam}). 
To actually use these equations to explicitly describe nonlocal interactions, we must expand the derivatives with respect to $\xi$ and $\bkappa$ in  \eqref{Final_Euler_Poincare_equations}. However, we emphasize again that it is interesting that nonlocal interactions can be expressed so as to formally coincide with the equations for the purely elastic motion. See \S\ref{sec:recovering_EP} for a detailed discussion of this point. 
} 
\end{remark}

\section{Hamiltonian structure of the strand equations}
\label{sec:Hamiltonian}
It is useful to transform the Lagrangian dynamical equations into the Hamiltonian description, both to relate these equations to previous work on elastic rods and to elucidate further their mathematical structure. We start by Legendre transforming the total Lagrangian $l$ to the Hamiltonian,
\begin{equation}
h(\bmu,\bbeta,\bOm,\bGam,\brho)
=
\int (\bmu\cdot \bom + \bbeta\cdot\bgam)\,ds
-
l(\bom,\bgam,\bOm,\bGam,\brho)
\,,
\label{Leg-transf}
\end{equation}
where $\bom, \bgam$ are determined from the relations $\bmu = \delta l/ \delta\bom$ and $\bbeta= \delta l/\delta\bgam $ upon assuming that $l $ is hyperregular.
Then,
equations \eqref{rhotimederiv}, \eqref{kincond}, \eqref{kincondom}, and  \eqref{Final_Euler_Poincare_equations} may be expressed in \textit{Lie-Poisson form with three cocycles} as
\begin{equation}
\frac{\partial}{\partial t}
\left[
\begin{array}{c}
    \bmu
    \\
    \bbeta
    \\
    \bOm
    \\
    \bGam
    \\
    \brho
    \end{array}
\right]
\!=\!
\left[
\begin{array}{ccccc}
   \bmu\times
   &
   \bbeta\times
   &
   (\partial_s + \bOm\times)
   &
   \bGam\times
   &\
   \brho\times\
   \\
   \bbeta\times
   &
   0
      &
   0
   &
   (\partial_s + \bOm\times)
   &
   -\rm Id\
   \\
   (\partial_s + \bOm\times)
   &
   0
   &
   0
   &
   0
   &
   0
   \\
   \bGam\times
   &
   (\partial_s + \bOm\times)
   &
   0
   &
   0
   &
   0
   \\
   \brho\times
   &
   \rm Id
   &
   0
   &
   0
   &
   0
   \end{array}
\right]
\left[
\begin{array}{c}
   \delta h/\delta\bmu\\
   \delta h/\delta\bbeta\\
   \delta h/\delta \bOm\\
   \delta h/\delta\bGam\\
   \delta h/\delta\brho
   \end{array}
\right] .
\label{LP-Ham-struct-vec-se3}
\end{equation}
Note that $\bom = \delta h / \delta \bmu$ and $\bgam = \delta h / \delta \bbeta$.
The affine terms $\partial_s$ and $\rm Id$ arise from the cocycle appearing in the definition of the variables $\bOm,\bGam, \brho$ in \eqref{bundle.coords}; see also \eqref{rhodef_n_dimensional}. These equations produce the affine terms located in the matrix elements $\{\bmu,\,\bOm\}$, $\{\bbeta,\bGam\}$, and $\{\bbeta,\,\brho\}$.

This Hamiltonian matrix defines an \textit{affine Lie-Poisson bracket} on the dual of the semidirect product Lie algebra
\[
\mathcal{F}(I,\mse(3))\,\circledS\,\mathcal{F}(I,\mse(3)\times\mathbb{R}^3),
\]
where $\mse(3)=\mso(3)\,\circledS\,\mathbb{R}^3$, $I=[0,L]$, and
\[
(\bmu,\bbeta)\in \mathcal{F}(I,\mse(3))^*\quad\text{and}\quad (\bOm,\bGam,\brho)\in \mathcal{F}(I,\mse(3)\times\mathbb{R}^3)^*.
\]
The associated affine Lie-Poisson bracket reads
\begin{align}
\{f,g\}_{\_}(\bmu,\bbeta,\bOm,\bGam,\brho)
=&
-\int\bmu\cdot\left(
\dede{f}{\bmu}\times\dede{g}{\bmu}\right)
-\int
\bbeta\cdot\left(
 \dede{f}{\bbeta}\times\dede{g}{\bmu}
- 
\dede{g}{\bbeta}\times \dede{f}{\bmu}
\right)
\nonumber\\
&-\int\bOm\cdot\left(
\dede{f}{\bOm}\times\dede{g}{\bmu}
-
\dede{g}{\bOm}\times\dede{f}{\bmu}
\right)
\nonumber\\
&-\int\bOm\cdot\left(
\dede{f}{\bGam}\times\dede{g}{\bbeta}
- \dede{g}{\bGam}\times\dede{f}{\bbeta}
\right)
\nonumber\\
&
-\int
\bGam\cdot\left(
\dede{f}{\bGam}\times\dede{g}{\bmu}
-
\dede{g}{\bGam}\times\dede{f}{\bmu}
\right)
\label{convect-brkt} \\
&
-\int
\brho\cdot\left(
\dede{f}{\brho}\times\dede{g}{\bmu}
-
\dede{g}{\brho}\times\dede{f}{\bmu}
\right)
\nonumber \\
&+\int
\dede{f}{\bOm}\cdot\partial_s\dede{g}{\bmu}
+\dede{f}{\bGam}\cdot\partial_s\dede{g}{\bbeta}
+\dede{f}{\brho}\cdot\dede{g}{\bbeta}
\nonumber\\
&
- \int
\dede{g}{\bOm}\cdot\partial_s\dede{f}{\bmu}
+\dede{g}{\bGam}\cdot\partial_s\dede{f}{\bbeta}
+\dede{g}{\brho}\cdot\dede{f}{\bbeta}
\,.
\nonumber
\end{align}

The first line represents the Lie-Poisson bracket on the Lie algebra $\mathcal{F}(I,\mse(3))$. The first five lines represent the Lie-Poisson bracket on the semidirect product Lie algebra
\[
\mathcal{F}(I,\mse(3))\,\circledS\,\mathcal{F}(I,\mse(3)\times\mathbb{R}^3).
\]
The last two lines represent the affine terms due to the presence of a cocycle, as well as the canonical Poisson bracket in $(\brho,\, \bbeta)$. The Poisson bracket (\ref{convect-brkt}) is an extension to include $\brho$ of the Poisson bracket for the exact geometric rod theory of \cite{SiMaKr1988} in the convective representation. Remarkably, from a geometric point of view, this Hamiltonian structure is \emph{identical} to that of complex fluids \cite{Gay-Bara2007, Ho2001}. The reason for this will be explained in detail in Section~\ref{sec:affineEP}.

\section{The affine Euler-Poincar\'e and Lie-Poisson approaches}
\label{sec:affineEP}

This section explains how the equations of the charged strand may be obtained by affine Euler-Poincar\'e and affine Lie-Poisson reduction. This proves that the charged strand admits the same geometrical description as the complex fluids and spin systems.\smallskip

We begin by recalling from \cite{Gay-Bara2007} the theory of {\bfi affine Euler-Poincar\'e} and {\bfi Lie-Poisson reduction}. In contrast to \cite{Gay-Bara2007}, however, we consider here Lagrangians and Hamiltonians that are  \textit{left\/}-invariant, rather than right-invariant.

\subsection{Notations for semidirect products}\label{Notations}

Let $V$ be a vector space and assume that the Lie group $G$ acts on the \textit{left\/} by linear maps (and hence $G$ also acts on the left on the dual space $V^*$). As a set, the semidirect product $S=G\,\circledS\,V$
is the Cartesian product $S=G\times V$ whose group multiplication is given by
\[
(g_1,v_1)(g_2,v_2)=(g_1g_2,v_1+g_1v_2),
\]
where the action of $g\in G$ on $v\in V$ is denoted simply as $gv$.
The Lie algebra of $S$ is the semidirect product Lie algebra,
$\mathfrak{s}=\mathfrak{g}\,\circledS\,V$, whose bracket has the expression
\[
\operatorname{ad}_{(\xi_1,v_1)}(\xi_2,v_2)=[(\xi_1,v_1),(\xi_2,v_2)]=([\xi_1,\xi_2],\xi_1v_2-\xi_2v_1),
\]
where $\xi v$ denotes the induced action of $\mathfrak{g}$ on $V$, that is,
\[
\xi v:=\left.\frac{d}{dt}\right|_{t=0}\operatorname{exp}(t\xi)v\in
V.
\]
From the expression for the Lie bracket, it follows that for
$(\xi,v)\in\mathfrak{s}$ and $(\mu,a)\in\mathfrak{s}^*$ we have
\[
\operatorname{ad}_{(\xi,v)}^*(\mu,a)=(\operatorname{ad}^*_\xi\mu-v\diamond a,-\xi a)
\]
where $\xi a\in V^*$ and $v\diamond a\in\mathfrak{g}^*$ are given by
\[
\xi a:=\left.\frac{d}{dt}\right|_{t=0}\operatorname{exp}(t\xi)a\quad\text{and}\quad
\langle v\diamond a,\xi\rangle_\mathfrak{g}:=-\langle \xi a,v\rangle_V,
\]
and where $\left\langle\cdot , \cdot \right\rangle_ \mathfrak{g}: \mathfrak{g}
^\ast \times \mathfrak{g}\rightarrow \mathbb{R}$ and $\left\langle \cdot ,
\cdot
\right\rangle_V: V ^\ast \times V \rightarrow \mathbb{R}$ are the duality pairings. The coadjoint action of $S $ on $\mathfrak{s} ^\ast$ has the expression
\begin{equation}
\label{coadjoint_action_group}
\operatorname{Ad}^\ast_{( g, v ) ^{-1}} ( \mu, a ) = \left( \operatorname{Ad}^\ast_{g ^{-1}} \mu + v \diamond ga , ga \right).
\end{equation}

Suppose we are given a \textit{left\/} representation of $G$ on the vector space $V^*$. We can form an \textit{affine left\/}
representation $\theta_g(a):= ga+c(g)$, where
$c\in\mathcal{F}(G,V^*)$ is a {\bfi left group
one-cocycle\/}, that is, it verifies the property
\begin{equation}\label{cocycle_identity}
c(gh)=c(g)+gc(h)
\,,
\end{equation}
for all $g, h \in G$. Note that
\[
\left.\frac{d}{dt}\right|_{t=0}\theta_{\operatorname{exp}(t\xi)}(a)=\xi a+\mathbf{d}c(\xi)
\]
and
\[
\langle \xi a+\mathbf{d}c(\xi),v\rangle_V=\langle \mathbf{d}c^T(v)-v\diamond
a,\,\xi\rangle_\mathfrak{g}
\,,
\]
where $\mathbf{d}c :\mathfrak{g}\rightarrow V^*$ is defined by
$\mathbf{d}c(\xi):=T_ec(\xi)$, and $\mathbf{d}c^T:V\rightarrow\mathfrak{g}^*$
is
defined by
\[
\langle \mathbf{d}c^T(v),\xi\rangle_\mathfrak{g}:=\langle
\mathbf{d}c(\xi),v\rangle_V.
\]
\subsection{Affine Lagrangian and Hamiltonian semidirect product theory}

Concerning the Lagrangian side, the general setup is the following.

\begin{itemize}
\item Assume that we have a function $L:TG\times V^*\rightarrow\mathbb{R}$
which
is \textit{left\/} $G$-invariant under the affine action $(v_h,a)\mapsto
(gv_h,\theta_g(a))=(gv_h,ga+c(g))$.
\item In particular, if $a_0\in V^*$, define the Lagrangian
$L_{a_0}:TG\rightarrow\mathbb{R}$ by $L_{a_0}(v_g):=L(v_g,a_0)$. Then $L_{a_0}$
is left invariant under the lift to $TG$ of the left action of $G_{a_0}^c$ on
$G$, where $G_{a_0}^c$ is the isotropy group of $a_0$ with respect to the
affine action $\theta$.
\item Define $l:\mathfrak{g}\times V^*\rightarrow\mathbb{R}$ by $l: =  L|_{ \mathfrak{g}\times V ^\ast}$. Left $G$-invariance of $L$ yields
\[
l(g^{-1}v_g,\theta_{g^{-1}}(a))=L(v_g,a)
\]
for all $g \in G $, $v _g\in T _gG $, $a \in V ^\ast$.
\item For a curve $g(t)\in G$, let $\xi(t):=g(t)^{-1}\dot{g}(t)$ and
define the curve $a(t)$ as the unique solution of the following affine
differential equation with time dependent coefficients
\[
\dot{a}=-\xi a-\mathbf{d}c(\xi),
\]
with initial condition $a(0)=a_0$. The solution can be written as
$a(t)=\theta_{g(t)^{-1}}(a_0)$.
\end{itemize}

\begin{theorem}\label{AEPSD} In the preceding notation, the following are equivalent:
\begin{itemize}
\item[\bf{i}] With $a_0$ held fixed, Hamilton's variational principle
\begin{equation}\label{Hamilton_principle}
\delta\int_{t_0}^{t_1}L_{a_0}(g,\dot{g})dt=0,
\end{equation}
holds, for variations $\delta g(t)$ of $g(t)$ vanishing at the endpoints.
\item[\bf{ii}] $g(t)$ satisfies the Euler-Lagrange equations for $L_{a_0}$ on
$G$.
\item[\bf{iii}] The constrained variational principle
\begin{equation}\label{Euler-Poincare_principle}
\delta\int_{t_0}^{t_1}l(\xi,a)dt=0,
\end{equation}
holds on $\mathfrak{g}\times V^*$, upon using variations of the form
\[
\delta\xi=\frac{\partial\eta}{\partial t}+[\xi,\eta],\quad \delta
a=-\eta a-\mathbf{d}c(\eta),
\]
where $\eta(t)\in\mathfrak{g}$ vanishes at the endpoints.
\item[\bf{iv}] The affine Euler-Poincar\'e equations hold on
$\mathfrak{g}\times
V^*$:
\begin{equation}\label{AEP}
\frac{\partial}{\partial t}\frac{\delta
l}{\delta\xi}=\operatorname{ad}^*_\xi\frac{\delta l}{\delta\xi}+\frac{\delta
l}{\delta a}\diamond a-\mathbf{d}c^T\left(\frac{\delta l}{\delta a}\right).
\end{equation}
\end{itemize}
\end{theorem}

See \cite{Gay-Bara2007} for the proof and applications to \textit{spin systems} and \textit{complex fluids}. Concerning the Hamiltonian side, the setup is the following.

\begin{itemize}
\item Assume that we have a function $H: T ^\ast G \times V ^\ast \rightarrow
\mathbb{R}$ which is \textit{left\/} invariant under the affine action $(\alpha_h,a)\mapsto(g\alpha_h ,\theta_g(a))$.
\item In particular, if $a_0\in V^*$, define the Hamiltonian $H_{a_0}:T^*G\rightarrow\mathbb{R}$ by $H_{a_0}(\alpha_g):=H(\alpha_g,a_0)$. Then $H_{a_0}$ is left invariant under the lift to $T^*G$ of the left action of $G^c_{a_0}$ on $G$.
\item Define $h:\mathfrak{g}^*\times V^*\rightarrow\mathbb{R}$ by $h: = H|_{ \mathfrak{g}^\ast \times V ^\ast}$. Left $G$-invariance of $H$ yields
\[
h(g^{-1}\alpha_g,\theta_{g^{-1}}(a))=H(\alpha_g,a).
\]
for all $g \in G $, $\alpha_g\in T_g ^\ast G$, $a \in V ^\ast$.
\end{itemize}

Note that the $G$-action on $T ^\ast G \times V ^\ast$ is induced by the $S$-action on $T ^* S $ given by
\begin{align}\label{affine_action_on_T*S}
\Psi_{(g,v)}(\alpha_h,(u,a)):=\left(g\alpha_h ,v+gu,ga+c(g)\right).
\end{align}
The affine action $\Psi$ appears as a modification of the cotangent lift of left translation on $S$ by an affine term. Thus, we can think of the Hamiltonian $H:T^*G\times V^*\to\mathbb{R}$
as being the Poisson reduction of a $S$-invariant Hamiltonian
$\overline{H}:T^*S\to\mathbb{R}$ by the normal subgroup $\{e\} \times V $ since
$(T ^\ast S)/(\{e\} \times V)  \cong T ^\ast G \times V^\ast$. Note also that every  Hamiltonian $\overline{H}=\overline{H}(\alpha_h,(u,a))$, defined on $T^*S$ and left invariant under the affine action $\Psi$, does not depend on the variable $u\in V$.

\begin{theorem}\label{ALPSD} Let $\alpha(t)\in T^*_{g(t)}G$ be a solution of Hamilton's equations associated to $H_{a_0}$ with initial condition $\mu_0 \in T_e ^\ast G$. Then $(\mu(t), a (t)):=(g(t)^{-1}\alpha(t), \theta_{ g (t)^{-1}} (a_0)) \in\mathfrak{g}^* \times  V ^\ast$  is a solution of the affine Lie-Poisson equations on $\mathfrak{s}^*$:
\[
\frac{\partial}{\partial t}(\mu,a)=\left(\operatorname{ad}^*_{\frac{\delta
h}{\delta \mu}}\mu-\frac{\delta h}{\delta a}\diamond
a+\mathbf{d}c^T\left(\frac{\delta h}{\delta a}\right),-\frac{\delta h}{\delta
\mu}a-\mathbf{d}c\left(\frac{\delta h}{\delta \mu}\right)\right)
\]
with initial conditions $( \mu(0), a (0)) = ( \mu_0, a_0)$.
The associated Poisson bracket is
the affine Lie-Poisson bracket on the dual $\mathfrak{s}^*$
\begin{align}
\label{affine_LP}
\{f,g\}(\mu,a)&=-\left\langle\mu,\left[\frac{\delta
f}{\delta\mu},\frac{\delta g}{\delta\mu}\right]\right\rangle-\left\langle
a,\frac{\delta f}{\delta \mu}\frac{\delta g}{\delta a}-\frac{\delta g}{\delta
\mu}\frac{\delta f}{\delta a}\right\rangle \nonumber \\
&\qquad+\left\langle\mathbf{d}c\left(\frac{\delta
f}{\delta\mu}\right),\frac{\delta g}{\delta
a}\right\rangle-\left\langle\mathbf{d}c\left(\frac{\delta
g}{\delta\mu}\right),\frac{\delta f}{\delta a}\right\rangle.
\end{align}
Conversely, given $\mu_0 \in T_e ^\ast G$, the solution $\alpha(t) $ of the Hamiltonian system associated to $H_{a_0}$ is reconstructed from the solution $( \mu(t), a (t)) $ of the affine Lie-Poisson equations with initial conditions $( \mu(0), a (0)) = ( \mu_0, a_0)$ by setting $\alpha(t) = g (t) \mu(t) $, where $g(t) $ is the unique solution of the differential equation $\dot{g}(t) = g(t) \frac{\delta h}{ \delta\mu(t)}$ with initial condition $g(0) = e$.
\end{theorem}

\begin{proof}
See \cite{Gay-Bara2007} for the proof and some applications. 
\end{proof}

\paragraph{Momentum maps} 
We now comment on the momentum maps at each stage of the reduction process. In \cite{Gay-Bara2007} it is shown that the momentum map associated to the affine action \eqref{affine_action_on_T*S} is given by
\begin{equation}\label{total_momentum}
\mathbf{J}:T^*S\rightarrow\mathfrak{s}^*,\quad\mathbf{J}(\alpha_g,(u,b))=(\alpha_gg^{-1}+u\diamond b-\mathbf{d}c^T(u),b).
\end{equation}
The proof of this formula uses the general formula for the momentum map on a magnetic cotangent bundle with respect to the cotangent-lifted action. 
In order to apply this formula, an adequate fiber translation on $T^*S$ used. This fiber translation turns out to be equivariant with respect to the affine action and the action \eqref{affine_action_on_T*S} on $T^*S$ as well as symplectic with respect to the canonical symplectic form and a magnetic symplectic form on $T^*S$. The above formula for $\mathbf{J}$ is then obtained by pulling back the magnetic momentum map via the fiber translation.

One observes that the conservation of $\mathbf{J}$ implies that the motion takes place on affine coadjoint orbits. 

Note that the Poisson action of $G $ on $T ^\ast G \times V ^\ast$ does not admit a momentum map because the leaves $T ^\ast G \times \{ b\} $, $b \in V^\ast$, are not invariant under this action. Given $a_0\in V^*$, the momentum map on $T^*G$ corresponding to the cotangent lifted left action of the isotropy group $G_{a_0}^c$ is given by
\[
\mathbf{J}_{a_0}: T^*G\rightarrow(\mathfrak{g}_{a_0}^c)^*,\quad\mathbf{J}_{a_0}(\alpha_g)=\alpha_gg^{-1}|_{\mathfrak{g}_{a_0}^c},
\]
where $\mathfrak{g}^c_a=\{\xi\in\mathfrak{g}\mid \xi a+\mathbf{d}c(\xi)=0\}$ is the Lie algebra of $G^c_{a_0}$.

\subsection{Affine reduction at fixed parameter}

As we will see, the affine reduction theorems recalled above do  
not apply directly to the molecular strand. This is because the Lagrangian of the molecular strand is only given  
for the particular value $a_0=0$ of the parameter and we do not have a concrete expression for $L_{a_0}$ when $a_0\neq 0$ is an arbitrary element of $V ^\ast$. Extending $L_0$ by $G$-invariance only yields a Lagrangian on $TG \times \mathcal{O}^c_{0}$, where $\mathcal{O}^c_{0} \subset V ^\ast$ is the orbit of the affine $G$-action on $V^\ast$. Fortunately, the  
Lagrangian $L_0$ for the molecular strand is invariant under the  
isotropy group $G^c_0=\{g\in G\mid c(g)=0\}$ and this turns out to be enough for the extension of the affine semidirect product reduction theorem.

\subsubsection{Lagrangian approach}

We consider here the case of a $G^c_{a_0}$-invariant Lagrangian  
$L_{a_0}:TG\rightarrow\mathbb{R}$ for a fixed $a_0\in V^*$, but we do   not suppose that this Lagrangian comes from a $G$-invariant function  $L:TG\times V^*\rightarrow\mathbb{R}$. In particular, we do not know the expression of $L_a$ when $a\neq 0$ is an arbitrary element of $V ^\ast$. To $L_{a_0}$ we associate the reduced Lagrangian $l$ defined on the submanifold
\[
\mathfrak{g}\times\mathcal{O}^c_{a_0}\subset \mathfrak{g}\times V ^\ast,\quad \mathcal{O}^c_{a_0}:=\{\theta_g(a_0)\mid g\in G\}
\]
given by $l(\xi,\theta_g(a_0))=L_{a_0}(g^{-1}\xi)$.  The tangent space at  $a $ to $\mathcal{O}^c_{a_0}$ is given by 
\begin{equation}\label{tangent_space_orbit_0}T_a \mathcal{O}^c_{a_0} = \{ \mathbf{d}c( \eta) + \eta  a \mid \eta  \in \mathfrak{g}\}.
\end{equation}

The analogue of Theorem \ref{AEPSD} in this case is given below.

\begin{theorem}\label{coadjoint_lagrangian_reduction} 
Let $a_0$ be a fixed element in $V^*$ and $g(t)$ be a curve in $G$ with $g(0) = e$. Define the curves $\xi(t)=g(t)^{-1}\dot g(t) \in \mathfrak{g}$ and $a (t) : = \theta_{g (t)^{-1}}a_0  \in V ^\ast$. Then the following are equivalent.
\begin{itemize}
\item[\bf{i}] With $a_0$ held fixed, Hamilton's variational principle
\begin{equation}\label{Hamilton_principle_new}
\delta\int_{t_0}^{t_1}L_{a_0}(g,\dot{g})dt=0,
\end{equation}
holds, for variations $\delta g(t)$ of $g(t)$ vanishing at the endpoints.
\item[\bf{ii}] $g(t)$ satisfies the Euler-Lagrange equations for $L_{a_0}$ on
$G$.
\item[\bf{iii}] The constrained variational principle
\begin{equation}\label{Euler_Poincare_principle_new}
\delta\int_{t_0}^{t_1}l(\xi,a)dt=0,
\end{equation}
holds on $\mathfrak{g}\times\mathcal{O}^c_{a_0} \subset  
\mathfrak{g}\times V ^\ast$, upon using variations of the form
\[
\delta\xi=\frac{\partial\eta}{\partial t}+[\xi,\eta],\quad\delta a=-\eta a-\mathbf{d}c(\eta),
\]
where $\eta(t)\in\mathfrak{g}$ vanishes at the endpoints.
\item[\bf{iv}] Extending $l$ arbitrarily to $\mathfrak{g}\times V^*$, the affine Euler-Poincar\'e equations hold on the submanifold
$\mathfrak{g}\times\mathcal{O}^c_{a_0} \subset  
\mathfrak{g}\times V ^\ast$:
\begin{equation}\label{AEP1}
\frac{\partial}{\partial t}\frac{\delta
l}{\delta\xi}=\operatorname{ad}^*_\xi\frac{\delta l}{\delta\xi}+\frac{\delta
l}{\delta a}\diamond a-\mathbf{d}c^T\left(\frac{\delta l}{\delta a}\right).
\end{equation}
\end{itemize}
\end{theorem}

\begin{proof} The equivalence of \textbf{i} and \textbf{ii} is true in general. The equivalence of \textbf{i} and \textbf{iii} and the equivalence of \textbf{iii} and \textbf{iv} can be shown exactly as in the standard case, that is, the case when $l$ is defined on the whole space $\mathfrak{g}\times V^*$. The only minor difference occurs when $l$ is differentiated with respect to the second variable. In this case the functional derivative $\delta l/\delta a\in V$ is replaced by the tangent map $\mathbf{d}_2l(\xi,a)\in T^*_a\mathcal{O}^c_{a_0}$ and one observes that
\[
\mathbf{d}_2l(\xi,a)\!\cdot\!\delta a=\left\langle\frac{\delta \tilde l}{\delta a},\delta a\right\rangle,\;\;\text{for all}\;\; \delta a\in T_a\mathcal{O}^c_{a_0}
\]
for any extension $\tilde l$ of $l$ to $\mathfrak{g}\times V^*$. Note that $\delta a=-\eta a-\mathbf{d}c(\eta)\in T_a\mathcal{O}^c_{a_0}$ for $\eta\in\mathfrak{\eta}$ and that any vector in $T_a\mathcal{O}^c_{a_0}$ is of this form. From now on we denote also by $l$, instead of $\tilde l$, an arbitrary extension of $l$. \end{proof}

\begin{remark}[The case $a_0=0$ and the charged strand]\label{remark_Lagrangian_side}{\rm For the charged molecular strand we will need to choose $a_0=0$. In this case the isotropy group is $G^c_0=\{g\in G\mid c(g)=0\}$. Given a $G^c_0$-invariant Lagrangian $L_0:TG\rightarrow\mathbb{R}$, the reduced Lagrangian $l$ is defined on $\mathfrak{g}\times \mathcal{O}^c_0$ by
\[
l(\xi,c(g^{-1}))=L_0(g\xi).
\]
It will be sufficient to restrict to Lagrangians for simple mechanical systems with symmetry, that is, of the form $L_0(v_g)=K(v_g)-P(g)$, where $K$ is the kinetic energy associated to a $G^c_0$-invariant Riemannian metric on $G$ and the potential $P$ is $G^c_0$-invariant. In this case, the reduced Lagrangian is
\[
l(\xi,c(g^{-1}))=K(g\xi)-P(g).
\]
Note that the right hand side of this expression is well defined on $\mathfrak{g} \times\mathcal{O}^c_0$, that is, it depends on $g $ only through $c(g ^{-1})$. Indeed, $c(g^{-1})=c(h^{-1})$ if and only if $\theta_{g^{-1}}(0)=\theta_{h^{-1}}(0)$, which means that $hg ^{-1}\in G^c_0$. Therefore, $P(h) = P(( h g ^{-1})g) = P(g) $ by left $G^c_0 $-invariance of $P$. For the kinetic energy the same argument works since the metric is $G^c_0 $-invariant.

Thus we can write $L_0(v_g)=K(v_g)-E(c(g^{-1}))$ for the function $E:V^*\rightarrow\mathbb{R}$ uniquely determined by the relation $P(g) = E(c(g^{-1}))$. In this case, we have
\[
l(\xi,c(g^{-1}))=K(g\xi)-E(c(g^{-1})).
\]
For the Lagrangian of the charged molecular strand the potential energy is the sum of two terms, one of which, denoted by $E_{loc}$, explicitly depends only on $c(g ^{-1}) $ and the other, denoted by $E_{np} $, does not have a concrete expression only in terms of $c(g ^{-1}) $ but it is $G^c_0$-invariant. In addition, for the charged molecular strand the kinetic energy metric is not just $G^c_0$-invariant but $G $-invariant which then implies that it is only a function of $\xi \in \mathfrak{g}$. For the molecular strand the Lagrangian is of the form
\[
L_0(v_g)=K(v_g)-E_{loc}(c(g^{-1}))-E_{np}(\zeta(g),c(g^{-1})),
\]
where $\zeta$ is a $G^c_0$-invariant function defined on $G$
and  the reduced Lagrangian is 
\begin{align*}
l(\xi,c(g^{-1}))&=\underbrace{K(\xi)-E_{loc}(c(g^{-1}))}_{=l_{loc}}-E_{np}(\zeta(g),c(g^{-1}))\\
&=l_{loc}(\xi,c(g^{-1}))+l_{np}(\zeta(g),c(g^{-1})).
\end{align*}
Note that $l$ can be expressed in terms of $(\xi,a)\in\mathfrak{g}\times\mathcal{O}^c_0$ as
\begin{equation}\label{abstract_strand}
l(\xi,a)=K(\xi)-E_{loc}(a)-E_{np}(\zeta(g_a),a)=l_{loc}(\xi,a)+l_{np}(\zeta(g_a),a),
\end{equation}
where $g_a\in G$ is such that $c(g^{-1})=a$. This $g_a$ is determined only up to left multiplication by $G^c_0$. Since $E_{np}$ is $G^c_0$-invariant, the function $a\mapsto E_{np}(g_a)$ is well-defined. 
Note that the Lagrangian of the strand (see \eqref{locL}, \eqref{Energy1}, and \eqref{lnpgen}) is exactly of the form \eqref{abstract_strand}, with $\zeta=\big( \xi(s,s'), \bkappa(s,s') \big) \in SE(3)$. 
 Since $a\mapsto l_{np}(\zeta(g_a),a)$ is a well-defined function of $a\in\mathcal{O}^\sigma_{0}$ one can ask why we insist in denoting $l_{np}=l_{np}(\zeta(g_a),a)$ instead of simply $l_{np}=l_{np}(a)$ which is mathematically correct. The reason is that for the molecular strand we do not have an explicit expression for $l_{np}:\mathcal{O}^c_0\rightarrow\mathbb{R}$; see \eqref{lnpgen}. Note that \eqref{lnpgen} is exactly of the form $l_{np}=l_{np}(\zeta(g_a),a)$. This will be explained in detail  in \S \ref{Application_strand}.
}
\end{remark}

\subsubsection{Recovering the modified Euler-Poincar\'e approach}
\label{sec:recovering_EP}

By Theorem \ref{coadjoint_lagrangian_reduction}, we have seen that the Euler-Lagrange equations of a $G^c_0$-invariant Lagrangian $L_0:TG\rightarrow\mathbb{R}$ are  equivalent to the affine Euler-Poincar\'e equations for $l:\mathfrak{g}\times\mathcal{O}^c_0\rightarrow\mathbb{R}$, that is, 
\begin{equation}\label{AEP_recall}
\frac{\partial}{\partial t}\frac{\delta l}{\delta\xi}=\operatorname{ad}^*_\xi \frac{\delta l}{\delta\xi}+\frac{\delta l}{\delta a}\diamond a-\mathbf{d}c^T\left(\frac{\delta l}{\delta a}\right).
\end{equation}
Recall that to write these equations, we need to extend $l$ to $\mathfrak{g}\times V^*$. Nevertheless, as we have shown, this extension does not affect the solution of these equations. For the molecular strand, there is an additional complication coming from the fact that the Lagrangian
\begin{equation}\label{Lagrangian_g_a}
l(\xi,a)=l_{loc}(\xi,a)+l_{np}(\zeta(g_a),a)
\end{equation}
being a well defined function of $(\xi,a)\in\mathfrak{g}\times\mathcal{O}^c_0$, is not explicitly written in terms of $a$. Therefore, when computing the affine Euler-Poincar\'e equations in concrete examples, there is still a dependence on $g_a$ in the final equation, although we know that this dependence can be replaced by a dependence in $a$ uniquely, by the results above. 

Let us apply the variational principle \eqref{Euler_Poincare_principle_new} to  Lagrangian in \eqref{Lagrangian_g_a}. Let $g (t)$ be a given curve in $G $. Take a family of curves $g_ \varepsilon(t)$ satisfying $g_0(t) = g(t)$ and denote $\eta(t): = g ^{-1}(t)\delta g (t)$. Then $\delta \int_{t_0}^{t_1} l( \xi(t), c( g (t)^{-1}) )dt = 0$ implies 
\begin{equation}\label{modified_AEP}
\frac{\partial}{\partial t}\frac{\delta l_{loc}}{\delta\xi}=\operatorname{ad}^*_\xi \frac{\delta l_{loc}}{\delta\xi}+\frac{\delta ( l_{loc}+l_{np})}{\delta a}\diamond a-\mathbf{d}c^T\left(\frac{\delta (l_{loc}+l_{np})}{\delta a}\right)+g^{-1}\frac{\delta l_{np}}{\delta\zeta}T_g\zeta.
\end{equation}
Note that this equation is the abstract generalization of equations \eqref{EPsigma} and \eqref{EPpsi}. 

Recall from the abstract theory that $l_{np}$ depends only on $a \in \mathcal{O}^c_0$. However, $l_{np}$ is given as a function of $(\zeta(g), 
c(g^{-1}))$. Let 
\[
\left.\frac{\delta l_{np}}{\delta a}\right|_{Tot}
\]
denote  the functional derivative of $l_{np}$ viewed as a function of $a \in \mathcal{O}^c_0$ only. Since every curve in $\mathcal{O}^c_0$ through $a = c(g^{-1}) \in \mathcal{O}^c_0$ is of the form $c( g_ \varepsilon^{-1})$, where $g_0=g$, we have

\begin{align} \label{tot_first}
\left.\frac{d}{d\varepsilon}\right|_{\varepsilon=0} l_{np}( \zeta( g_ \varepsilon), c(g _ \varepsilon^{-1})) 
&= \left\langle \left.\frac{\delta l_{np}}{ \delta a }\right|_{Tot}, 
\left.\frac{d}{d\varepsilon}\right|_{\varepsilon=0} c(g_ \varepsilon^{-1}) \right\rangle
= -\left\langle \left.\frac{\delta l_{np}}{ \delta a }\right|_{Tot}, \eta a + \mathbf{d}c( \eta) \right\rangle \nonumber \\
& = \left\langle \left.\frac{\delta l_{np}}{ \delta a }\right|_{Tot} \diamond a - \mathbf{d}c^T \left(\left.\frac{\delta l_{np}}{ \delta a }\right|_{Tot} \right), \eta \right\rangle,
\end{align}
where $\eta : = g ^{-1}\delta g$. On the other hand, 
\begin{align} \label{tot_second}
\left.\frac{d}{d\varepsilon}\right|_{\varepsilon=0} l_{np}( \zeta( g_ \varepsilon), c(g _ \varepsilon^{-1})) 
&= \left\langle \frac{ \delta l_{np}}{ \delta \zeta}, T_g \zeta( g\eta ) \right\rangle  - \left\langle\frac{\delta l_{np}}{ \delta a}, \eta a + \mathbf{d}c ( \eta) \right\rangle  \nonumber\\
& = \left\langle g^{-1} \frac{\delta l_{np}}{ \delta \zeta} T_g \zeta 
+ \frac{\delta l_{np}}{\delta a}\diamond a - \mathbf{d}c^T\left(\frac{\delta l_{np}}{\delta a}\right),
\eta \right\rangle.
\end{align}
Equations \eqref{tot_first} and \eqref{tot_first} prove the following identity
\[
\left.\frac{\delta l_{np}}{\delta a}\right|_{Tot}\diamond a-\mathbf{d}c^T\left(\left.\frac{\delta l_{np}}{\delta a}\right|_{Tot}\right)=\frac{\delta l_{np}}{\delta a}\diamond a-\mathbf{d}c^T\left(\frac{\delta l_{np}}{\delta a}\right)+g ^{-1} \frac{\delta l_{np}}{\delta\zeta}T_g\zeta,
\]
where $a = c(g^{-1}) $. Using this identity in \eqref{modified_AEP} we obtain the 
affine Euler-Poincar\'e equations \eqref{AEP_recall} since
\[
\frac{\delta l}{\delta a}=\left.\frac{\delta l_{np}}{\delta a}\right|_{Tot}+\frac{\delta l_{loc}}{\delta a}.
\]
Thus, the affine Euler-Poincar\'e process recovers the results of the modified Euler-Poincar\'e approach described in \S\ref{sec:Euler-Poincare}.

\subsubsection{Hamiltonian approach}

We now explore the Hamiltonian counterpart of the theory, that is, the case of a $G^c_{a_0}$-invariant Hamiltonian $H_{a_0}:T^*G\rightarrow\mathbb{R}$, defined only for a fixed value $a_0\in V^*$. As before, we do not suppose that $H_{a_0}$ is induced from a $G$-invariant Hamiltonian on $T^*G\times V^*$. In particular, we do not know the expression of $H_a$ for other choices of.. In particular, we do not know the expression of $H_a$ for other choices of $a$. Such an $H_{a_0}$ is usually induced by a hyperregular $G^c_{a_0}$-invariant Lagrangian $L_{a_0}$.
\medskip

As on the Lagrangian side, the reduced Hamiltonian is only defined on the submanifold
\[
\mathfrak{g}^*\times\mathcal{O}^c_{a_0}\subset\mathfrak{s}^*
\]
and so Theorem \ref{ALPSD} cannot be applied. However, as is shown in the next theorem, the fact that the reduced motion is Hamiltonian on an affine coadjoint orbit remains true for this more general case.
\medskip

We need to introduce the affine coadjoint orbit  
$\mathcal{O}^\sigma_{( \mu, a )}$. The left $V ^\ast$-valued group  
one-cocycle $c:G \rightarrow V ^\ast$ induces a left group one-cocycle  
$\sigma:S \rightarrow (\mathfrak{g}\,\circledS\,V) ^\ast$ by
\[
\sigma(g,u)=(u\diamond c(g)-\mathbf{d}c^T(u),c(g)).
\]
The affine coadjoint action of $S $ on $\mathfrak{s} ^\ast$ is hence given by
\[
(g, u)( \mu, a): =\operatorname{Ad}^*_{(g,u)^{-1}}(\mu,a)+\sigma((g,u)^{-1}).
\]
The connected components of the coadjoint orbits  
$\left(\mathcal{O}^\sigma_{(\mu,a_0)},\omega^-\right)$ are the  
symplectic leaves of $\mathfrak{s} ^\ast$ endowed with the affine  
Lie-Poisson bracket \eqref{affine_LP}. Denote by $S^\sigma_{(\mu,a)}$  
the isotropy group of the affine coadjoint action.

\begin{theorem}\label{coadjoint_reduction} Let $H_{a_0}:T^*G\rightarrow \mathbb{R}$ be a  
$G_{a_0}^c$-invariant Hamiltonian, where $a_0$ is a fixed element in  
$V^*$.  By $G^c_{a_0}$-invariance, we obtain the reduced Hamiltonian  
$h$ on
\[
\mathfrak{g}^*\times\mathcal{O}^c_{a_0} \subset \mathfrak{s} ^\ast
\]
defined by $h(\mu,\theta_g(a_0))=H_{a_0}(g^{-1}\mu)$.
\begin{itemize}
\item[{\bf(i)}]
Let $\alpha(t) \in T_{g(t)} ^\ast G$ be a solution of Hamilton's equations  
associated to $H_{a_0}$ with initial condition $\mu_0 \in T_e ^\ast G  
= \mathfrak{g}^\ast$. Then $(\mu(t), a (t)): = (g(t)^{-1}\alpha(t),  
\theta_{g(t)^{-1}}(a_0)) \in \mathfrak{s} ^\ast$ is the integral curve  
of the Hamiltonian vector field $X_h$ on the affine coadjoint orbit  
$\left(\mathcal{O}^\sigma_{(\mu_0,a_0)}, \omega^- \right)$ with  
initial condition $( \mu_0, a _0)$. Conversely, given $\mu_0 \in  
T_e ^\ast G$, the solution $\alpha(t) $ of the Hamiltonian system  
associated to $H_{a_0}$ is reconstructed from the solution $( \mu(t),  
a (t)) $ of $X_h \in  
\mathfrak{X}\left(\mathcal{O}^\sigma_{(\mu_0,a_0)} \right)$ with  
initial condition $( \mu_0, a _0) $ by setting $\alpha(t) = g (t)  
\mu(t) $, where $g(t) $ is the unique solution of the differential  
equation $\dot{g}(t) = g(t) \frac{\delta h}{ \delta\mu(t)}$ with  
initial condition $g(0) = e$.
\item[{\bf(ii)}] Extending $h $ arbitrarily to $\mathfrak{s} ^\ast$,  
Hamilton's equations on $\left(\mathcal{O}^\sigma_{(\mu_0,a_0)},  
\omega^- \right)$ can be written as
\[
\frac{\partial}{\partial t}(\mu,a)=\left(\operatorname{ad}^*_{\frac{\delta
h}{\delta \mu}}\mu-\frac{\delta h}{\delta a}\diamond
a+\mathbf{d}c^T\left(\frac{\delta h}{\delta a}\right),-\frac{\delta h}{\delta
\mu}a-\mathbf{d}c\left(\frac{\delta h}{\delta \mu}\right)\right)
\]
where $\mu(0) = \mu_0 $, $a(0)=a_0$.
\end{itemize}
\end{theorem}
\begin{remark}{\rm It important to observe that the given Hamiltonian  
$h$ is not defined on the whole dual Lie algebra $\mathfrak{s}^*$.  
Part \textbf{ii} of the theorem states that the equations of motion  
can be nevertheless computed from the usual formula of an affine  
Lie-Poisson vector field by arbitrarily extending $h $ to  
$\mathfrak{s} ^\ast$. Note that $\delta h / \delta \mu$ and $\delta h  
/ \delta a$ are only defined when one thinks of $h $ as being defined  
on $ \mathfrak{s} ^\ast$. It will be shown in the theorem that the  
extension of $h $ does not matter. This difficulty will appear  
concretely when dealing with the molecular strand.}
\end{remark}

\begin{proof}$\quad$\\
 \textbf{(i)} The action $\Psi$ of $S$ on $T^*S$ induces  
an action of $V$ given by
\[
(\alpha_h,(u,a))\mapsto(\alpha_h,v+u,a).
\]
Since $V$ is a closed subgroup of $S$, this action admits a momentum map
given by
\[
\mathbf{J}_V(\alpha_g,(u,a))=a.
\]
Since $V$ is an Abelian group, the coadjoint isotropy group of $a_0\in V^*$ is
$V_{a_0}=V$ and the first reduced space  
$(T^*S)_{a_0}=\mathbf{J}_V^{-1}(a_0)/V$ is
symplectically diffeomorphic to the canonical symplectic manifold
$(T^*G,\Omega_{\rm can})$. The action $\Psi$ of $S$ on $T^*S$ restricts to an
action $\Psi^{a_0}$ of $G_{a_0}^c\,\circledS\,V$ on  
$\mathbf{J}_V^{-1}(a_0)$. Passing to quotient spaces, this action  
induces an action of $G^c_{a_0}$ on $(T^*S)_{a_0}$, which is readily  
seen to be the cotangent lifted action of $G^c_{a_0}$ on $T^*G$. We  
denote by
$\mathbf{J}_{a_0} : (T^*S)_{a_0}\to(\mathfrak{g}^c_{a_0})^*$ the  
associated equivariant
momentum map, where $\mathfrak{g}_{a_0}^c$ is the Lie algebra of  
$G_{a_0}^c$. Reducing
$(T^*S)_{a_0}$ at the point $\mu_{a_0}:=\mu|\mathfrak{g}^c_{a_0}$, we  
get the second
reduced space  
$\left((T^*S)_{a_0}\right)_{\mu_{a_0}}=\mathbf{J}_{a_0}^{-1}(\mu_{a_0})/(G_{a_0}^c)_{\mu_{a_0}}$, with symplectic form denoted by  
$(\Omega_{a_0})_{\mu_{a_0}}$.

By the Reduction by Stages Theorem for nonequivariant momentum maps  
\cite{MaMiOrPeRa2007}, the second
reduced space is symplectically diffeomorphic to the reduced space
\[
\left(\mathbf{J}^{-1}(\mu,a_0)/S^\sigma_{(\mu,a_0)},\Omega_{(\mu,a_0)}\right)
\]
obtained by reducing $T^*S$ by the whole group $S$ at the point
$(\mu,a_0)\in\mathfrak{s}^*$. By affine Lie-Poisson reduction, this  
space is symplectically diffeomorphic to the affine coadjoint orbit
\[
\left(\mathcal{O}^\sigma_{(\mu,a_0)},\omega^-\right)
\]
endowed with the affine orbit symplectic symplectic form.

Note finally that by the symplectic reduction theorem, any solution of  
Hamilton's equations associated to $H_{a_0}$ on $T^*G$ reduces to and  
is reconstructed from a solution of  Hamilton's equations for the  
reduced Hamiltonian $h_{\mu_{a_0}}:  
\mathbf{J}_{a_0}^{-1}(\mu_{a_0})/(G_{a_0}^c)_{\mu_{a_0}} \rightarrow  
\mathbb{R}$, for a given momentum value  
$\mu_{a_0}\in(\mathfrak{g}^c_{a_0})^*$. As we have seen, this reduced  
space is symplectically diffeomorphic to the affine coadjoint orbit  
$\mathcal{O}^\sigma_{(\mu,a_0)}\subset\mathfrak{s}^*$, where  
$\mu\in\mathfrak{g}^*$ is such that  
$\mu|_{\mathfrak{g}^c_{a_0}}=\mu_{a_0}$. Thus, we can think of $h_{  
\mu_{a_0}}$ as being defined on $\mathcal{O}^\sigma_{(\mu,a_0)} $.
Viewed this way, $h_{ \mu_{a_0}}$ is simply the restriction of the  
function $h$ constructed from $H_{a_0}$ by
\[
h(\mu,\theta_g(a_0))=H(g^{-1}\mu,a_0).
\]
Note that $h$ is defined on any affine coadjoint orbit  
$\mathcal{O}_{(\mu,a_0)}^ \sigma$ with fixed $a _0\in V ^\ast$ since
\[
\mathfrak{g}^*\times \mathcal{O}^c_{a_0}=\bigcup_{\mu\in\mathfrak{g}^*}\mathcal{O}^\sigma_{(\mu,a_0)}\subset\mathfrak{s}^*.
\]

\textbf{(ii)} We begin by recalling a general fact from the theory of  
Poisson manifolds. Let $\varphi \in C ^{\infty}(P) $, where $P $ is a  
Poisson manifold and $X_ \varphi$ its Hamiltonian vector field. If $L  
$ is a symplectic leaf of $P $, then $X_ \varphi|_L = X_{ \varphi|L}$,  
where the right hand side is the Hamiltonian vector field on the  
symplectic manifold $L $. In our case $P = \mathfrak{s} ^\ast$ and $L  
= \mathcal{O}_{(\mu,a_0)}^ \sigma$.
\end{proof}
\medskip

\begin{remark}[The case $a_0=0$ and the charged strand]\label{remark_Hamiltonian_side}{\rm The Lagrangian
\[
L_0(v_g)=K(v_g)-E_{loc}(c(g^{-1}))-E_{np}(\zeta(g),c(g^{-1}))
\]
discussed in Remark \ref{remark_Lagrangian_side} is hyperregular, thus it induces the $G^c_0$-invariant Hamiltonian
\[
H_0(\alpha_g)=K(\alpha_g)+E_{loc}(c(g^{-1}))+E_{np}(\zeta(g),c(g^{-1}))
\]
whose reduced expression on $\mathfrak{g}^*\times\mathcal{O}^c_0$ reads
\[
h(\mu,c(g^{-1}))=\frac{1}{2}\|\mu\|^2+E_{loc}(c(g^{-1}))+E_{np}(\zeta (g), c(g^{-1})).
\]
As on the Lagrangian side, for $(\mu,a)\in \mathfrak{g}^*\times\mathcal{O}^c_0$ (or $(\mu,a)\in \mathcal{O}^{\sigma}_{(\mu_0,0)}$), we can write
\[
h(\mu,a)=\frac{1}{2}\|\mu\|^2+E_{loc}(a)+E_{np}(\zeta(g_a),a),
\]
where $g_a \in G$ is any group element satisfying $c(g_a^{-1})=a$.
}
\end{remark}

\begin{remark}[Affine coadjoint orbits and Noether's  theorem]{\rm As we have seen, the solution $(\mu,a)$ evolves on an affine coadjoint orbit, for any $G^c_{a_0}$-invariant Hamiltonian $H_{a_0}$. If $L_{a_0}$ is the Lagrangian of a simple mechanical system with symmetry then, by Noether's theorem, the solution $(\xi,a)$ is constrained to evolve on the submanifolds
\[
\left(\mathcal{O}^\sigma_{(\mu_0,a_0)}\right)^\sharp=\left \{(\xi,a)\in \mathfrak{g}\times V^*\mid (\xi^\flat,\mu)\in  \mathcal{O}^\sigma_{(\mu_0,a_0)}\right\}.
\]
}
\end{remark}

\subsection{Application to the charged strand}\label{Application_strand}

In this subsection we apply the affine Euler-Poincar\'e and Lie-Poisson reduction theorems to the molecular strand. In order to give a more transparent vision of the underlying geometric structures, we consider the $n$-dimensional generalization described in Subsection \ref{n_dimensional_generalization}, that is, we replace the interval $[0,L]$ be an arbitrary manifold $\mathcal{D}$ and we replace $SE(3)$ by the semidirect product $S=\mathcal{O}\,\circledS\,E$ of a Lie group $\mathcal{O}$ with a  \textit{left\/} representation space $E$. Given a manifold $\mathcal{D}$, we define the group $G:=\mathcal{F}(\mathcal{D},S)$ and the dual vector space $V^*:=\Omega^1(\mathcal{D},\mathfrak{s})\oplus\mathcal{F}(\mathcal{D},E)$. The elements of the group $G$ are denoted by $(\Lambda,r)$, where $\Lambda:\mathcal{D}\rightarrow \mathcal{O}$ and $r:\mathcal{D}\rightarrow E$. The elements of $V^*$ are denoted by $(\Omega,\Gamma,\rho)$, where $\Omega\in\Omega^1(\mathcal{D},\mathfrak{o})$, $\Gamma\in\Omega^1(\mathcal{D},E)$, and $\rho:\mathcal{D}\rightarrow E$. The space $V^*$ can be seen as the dual of $V=\mathfrak{X}(\mathcal{D},\mathfrak{s}^*)\oplus\mathcal{F}(\mathcal{D},E^*)$, where $\mathfrak{X}(\mathcal{D},\mathfrak{s})$ is the space of $\mathfrak{s}$-valued vector fields on $\mathcal{D}$.

Consider the representation of $G$ on $V^*$ defined by
\begin{equation}\label{left_representation}
(\Lambda,r)(\Omega,\Gamma,\rho)=(\operatorname{Ad}_{(\Lambda,r)}(\Omega,\Gamma),\Lambda\rho)
\end{equation}
where the adjoint action is that of $S$, acting here on functions defined on $\mathcal{D}$, and $\Lambda\rho$ denotes the left representation of $\mathcal{O}$ on $E$, acting on functions. The main object for this approach is the {\bfi group one-cocycle} $c$ appearing already implicitly in the definition of the variables $\bOm, \bGam, \brho$ in \eqref{rhodef}, \eqref{rhodef_n_dimensional}, and explicitly in \eqref{cocycle}. Recall that it is given by
\[
c(\Lambda,r):=\left((\Lambda,r)\mathbf{d}(\Lambda,r)^{-1},-r\right).
\]
Let's verify the cocycle identity  for the first component $(\Lambda,r)\mathbf{d}(\Lambda,r)^{-1} $. To simplify notation, denote $\chi_i: = ( \Lambda_i, r_i ) \in \mathcal{F}( \mathcal{D}, S), i\in \{1;2\}$. We have
\begin{align*}
 \chi_1 \chi_2 \mathbf{d}(\chi_1 \chi_2)^{-1} 
 &=   \chi_1 \chi_2 \mathbf{d}(\chi_2^{-1}\chi_1  ^{-1}) 
 =  \chi_1 \chi_2  \mathbf{d}(\chi_2^{-1}) \chi_1  ^{-1} + 
 \chi_1 \chi_2 \chi_2^{-1}\mathbf{d}(\chi_1  ^{-1}) \\
 & = \operatorname{Ad}_ {\chi_1} \left(\chi_2  \mathbf{d}\chi_2^{-1} \right)
 +  \chi_1 \mathbf{d}(\chi_1  ^{-1}).
\end{align*}
Since the second coordinate of $\left(( \Lambda_1,r _1) ( \Lambda_2,r _2) \right)$ is equal to   $r= r _1 + \Lambda _1 r _2$,  we find
\begin{align*}
&c \left(( \Lambda_1,r _1) ( \Lambda_2,r _2) \right)= \\
&\qquad\quad = \left(\operatorname{Ad}_ {( \Lambda_1,r _1)} \left(( \Lambda_2,r _2)  \mathbf{d}( \Lambda_2,r _2)^{-1} \right)
 +  ( \Lambda_1,r _1) \mathbf{d}(( \Lambda_1,r _1)^{-1}), - r _1 - 
 \Lambda _1 r _2  \right) \\
 &\qquad\quad = \left( \operatorname{Ad}_ {( \Lambda_1,r _1)} \left(( \Lambda_2,r _2)  \mathbf{d}( \Lambda_2,r _2)^{-1} \right), -  \Lambda _1 r _2 \right) + 
 \left(  (\Lambda_1,r _1) \mathbf{d}(( \Lambda_1,r _1)^{-1}), - r _1 \right)\\
 &\qquad\quad = ( \Lambda_1,r _1) c ( \Lambda_2, r_2) + 
 c ( \Lambda_1, r_1).
\end{align*} 
This shows that $c$ verifies the cocycle property \eqref{cocycle_identity} relative to the representation \eqref{left_representation}.$\qquad\blacklozenge$

\medskip
Note that the first component of $c$ is the left version of the cocycle appearing in the theory of complex fluids; see \cite{Gay-Bara2007}. Using the expressions
\[
(u,w,f)\diamond (\Omega,\Gamma,\rho)=(\operatorname{ad}^*_{\Omega_i}u^i+w^i\diamond\Gamma_i+f\diamond \rho,-\Omega_iw^i)
\,,
\]
\[
\mathbf{d}c(\omega,\gamma)=(-\mathbf{d}\omega,-\mathbf{d}\gamma,-\gamma),\quad\text{and}\quad \mathbf{d}c^T(u,w,f)=(\operatorname{div}(u),\operatorname{div}(w)-f)
\,,
\]
the affine Euler-Poincar\'e equations \eqref{AEP} become
\begin{equation}\label{generalized_molec_strand}
\left\lbrace\begin{array}{l}
\displaystyle\vspace{0.2cm}\left(\prt_t-\operatorname{ad}^*_{\omega}\right)\frac{\delta l}{\delta\omega}+\lp \operatorname{div} - \ad^*_\Om\rp\frac{\delta l}{\delta\Omega}=\frac{\delta l}{\delta\gamma}\diamond\gamma+\frac{\delta l}{\delta\Gamma}\diamond\Gamma+\frac{\delta l}{\delta\rho}\diamond\rho
\,,\\
\displaystyle\left(\prt_t+\omega\right)\frac{\delta l}{\delta\gamma}+\left(\operatorname{div}+\Omega\right)\frac{\delta l}{\delta\Gamma}=\frac{\delta l}{\delta\rho}
\,.
\end{array}\right.\end{equation}
and the {\bfi advection equations} are
\begin{equation}\label{kincoc}
\left\lbrace\begin{array}{l}
\displaystyle\vspace{0.2cm}
\prt_t\Omega+\operatorname{ad}_\omega\Omega=\mathbf{d}\omega
\,,\\
\displaystyle\vspace{0.2cm}\lp \prt_t+\omega\rp\Gamma=\lp\mathbf{d}+\Omega\rp\gamma
\,,\\
\displaystyle\vspace{0.2cm}
\prt_t\rho+\omega\rho=\gamma
\,.
\end{array}\right.
\end{equation}
\begin{remark}{\rm To write these equations, we have supposed that the dynamics is described by a Lagrangian $l$ given explicitly in terms of the variables $(\omega,\gamma,\Omega,\Gamma,\rho)$. Equivalently, we have assumed that $l$ is induced by an affine left-invariant Lagrangian $L$ defined on $TG\times V^*$. As we have seen in \S\ref{convected_representation}, such a hypothesis is not verified when nonlocal terms are taken into account. In this case, the affine Euler-Poincar\'e and affine Lie-Poisson reductions are not applicable and one needs to restrict to a particular value of the parameter $a_0$, by using Theorems \ref{coadjoint_lagrangian_reduction} and \ref{coadjoint_reduction}. For convenience, we first present the simpler case where the nonlocal terms are ignored.  We shall call this case elastic filament dynamics for simplicity.}
\end{remark}

\subsubsection{Elastic filament dynamics and Kirchhoff's theory}

Suppose that the dynamics of the strand is described by a Lagrangian $l=l(\omega,\gamma,\Omega,\Gamma,\rho)$ defined on $\mathfrak{g}\times V^*$, where $\mathfrak{g}=\mathcal{F}(\mathcal{D},\mathfrak{s})$ and $V^*=\Omega^1(\mathcal{D},\mathfrak{s})\oplus\mathcal{F}(\mathcal{D},E)$.  The Lagrangian $l$ is induced by a left invariant Lagrangian $L$ defined on $TG\times V^*$, where $G=\mathcal{F}(\mathcal{D},S)$.

Note that there is no restriction in the way $l$ depends on the variables. In particular the dependence can be nonlocal. However, it is supposed here that $l$ depends explicitly on the variables $(\omega,\gamma,\Omega,\Gamma,\rho)$. Recall that such an hypothesis is verified for the Lagrangian of Kirchhoff's theory \eqref{lkirchhoff} but is not verified for the Lagrangian of the molecular strand \eqref{locL}. 

The affine Euler-Poincar\'e reduction applies as follows. Fix the initial values $(\Omega_0,\Gamma_0,\rho_0)$ and define the Lagrangian
\[
L_{(\Omega_0,\Gamma_0,\rho_0)}(\Lambda,r):=L(\Lambda,r,\Omega_0,\Gamma_0,\rho_0).
\]
Consider a curve $(\Lambda,r)\in G$ and define the quantities
\begin{align*}
(\Omega,\Gamma,\rho)&=(\Lambda,r)^{-1}(\Omega_0,\gamma_0,\rho_0)+c((\Lambda,r)^{-1})\\
&=(\operatorname{Ad}_{\Lambda^{-1}}\Omega_0,\Lambda^{-1}(\Gamma_0+\Omega_0r),\Lambda^{-1}\rho_0) + (\Lambda^{-1}\mathbf{d}\Lambda,\Lambda^{-1}\mathbf{d}r,\Lambda^{-1}r).
\end{align*}
and
\[
\omega=\Lambda^{-1}\dot\Lambda,\quad\gamma=\Lambda^{-1}\dot r.
\]
Note that when the initial values $\Omega_0, \Gamma_0, \rho_0$ are zero, the definition of the variables $\omega, \gamma, \Omega, \Gamma, \rho$ coincide with those given in \eqref{rhodef} and \eqref{rhodef_n_dimensional}.

Then the curve $(\Lambda,r)$ is a solution of the Euler-Lagrange equations associated to $L_{(\Omega_0,\Gamma_0,r_0)}$ on $TG$ if and only if $(\omega,\gamma,\Omega,\Gamma,\rho)$
is a solution of the Euler-Poincar\'e equations \eqref{generalized_molec_strand}.

Of course, when $\mathcal{D}$ is the interval $[0,L]$ and $S$ is the semidirect product of $\mathcal{O}=SO(3)$ with $E=\mathbb{R}^3$, then we recover from  \eqref{generalized_molec_strand} the dynamical equation of the charged strand \eqref{Final_Euler_Poincare_equations}, since 
\[
\ad^* \to -\times
\quad\hbox{and}\quad
\di \to \times
.
..
\]
These equations are the convective representation of Kirchhoff's equations.
From \eqref{kincoc} we recover the advection relations derived in Subsection \ref{kinematics}.

\subsubsection{The charged strand: general case}

Recall from \S\ref{convected_representation} that the Lagrangian of the molecular strand has the expression
\[
l=l_{loc}(\bom,\bgam,\bOm,\bGam,\brho)+l_{np}(\xi,\bkappa,\bGam),
\]
where $l_{loc}$ is a local function of the form
\begin{align}\label{l_loc_expression}
l_{loc}(\bom,\bgam,\bOm,\bGam,\brho)=K(\bom,\bgam)-E_{loc}(\bOm,\bGam,\brho)
\end{align}
and $l_{np}$ is of the form
\[
l_{np}(\xi,\bkappa,\bGam)=\iint U\left(\xi(s,s'),\bkappa(s,s'),\bGam(s),\bGam(s')\right)dsds',
\]
where
\[
U:SE(3)\times\mathbb{R}^3\times\mathbb{R}^3\rightarrow\mathbb{R}\quad\text{and}\quad \left(\xi(s,s'),\bkappa(s,s')\right):=(\Lambda,\br)^{-1}(s)(\Lambda,\br)(s').
\]

\begin{remark}[Two crucial observations]$\quad$\\
\begin{enumerate}\rm
\item The nonlocal Lagrangian $l_{np}$ is induced by a $SO(3)$-invariant potential $E_{np}=E_{np}(\Lambda,\br)$. Thus the total Lagrangian $l$ can be seen as being induced by the $SO(3)$-invariant Lagrangian $L_0=L_0(\Lambda,\dot\Lambda,\br,\dot \br)$ given by
\[
L_0(\Lambda,\dot\Lambda,\br,\dot \br)=K(\Lambda,\dot\Lambda,\br,\dot \br)-E_{loc}\left(c\left((\Lambda,\br)^{-1}\right)\right)-E_{np}(\Lambda,\br),
\]
where $K$ is the $\mathcal{F}( \mathcal{D}, SE(3))$-left invariant extension of the kinetic energy $K$ in \eqref{l_loc_expression}. Note that we have replaced the dependence of $E_{loc}$ on $(\bOm,\bGam,\brho)$ by a dependence on $(\Lambda,\br)$ through the cocycle $c$. The affine Euler-Poincar\'e dynamics will yield the relation $(\bOm,\bGam,\brho)=c\left((\Lambda,\br)^{-1}\right)$ which allows us to recover the dependence of the potential on $(\bOm,\bGam,\brho)$.

\item The group $SO(3)$ is precisely the isotropy group 
\[
G^c_0=\mathcal{F}( \mathcal{D}, SE(3))^c_0 = \{(\Lambda,\br)\in G\mid c(\Lambda,\br)=0\}
\]
of the affine action at zero.
\end{enumerate}
\end{remark}

These two remarks allow us to obtain the dynamics of the molecular strand by the affine reduction processes described in Theorems \ref{coadjoint_lagrangian_reduction} and \ref{coadjoint_reduction}. As before, we choose to work with the general framework involving $\mathcal{D}$ and $\mathcal{O}\,\circledS\,E$. The present approach is applicable to any $\mathcal{O}$-invariant Lagrangian
\[
L_0=L_0(\Lambda,\dot\Lambda,r,\dot r):T[\mathcal{F}(\mathcal{D},\mathcal{O}\,\circledS\,E)]\rightarrow\mathbb{R}.
\]
Note there are no conditions on the dependence of $L_0$ on the variables $( \Lambda, r)$. In particular, $L_0$ can be nonlocal, and may depend on the derivatives of $\Lambda$ and $r$. An important class of such Lagrangians is given by
\[
L_0(\Lambda,\dot\Lambda,r,\dot r)=K(\Lambda,\dot\Lambda,r,\dot r)-P(\Lambda,r),
\]
where $K$ is the kinetic energy associated to an $\mathcal{O}$-invariant metric on $\mathcal{F}(\mathcal{D},\mathcal{O}\,\circledS\,E)$ and the potential $P $ is an $\mathcal{O}$-invariant function on $\mathcal{F}(\mathcal{D},\mathcal{O}\,\circledS\,E)$. In particular, $P$ can be nonlocal, or depend on derivatives of $\Lambda$ and $r$; see \eqref{lkirchhoff} for an example. In the case of the molecular strand, $K$ is assumed to be left-invariant and $P$ is given by
\[
P(\Lambda,r)=E_{loc}\left(c\left((\Lambda,r)^{-1}\right)\right)+E_{np}(\Lambda,r),
\]
where
\begin{align*}
&E_{np}(\Lambda,r):=\iint_{\mathcal{D}}U\left(\xi(s,s'),\kappa(s,s'),\Lambda^{-1}\mathbf{d}r(s),\Lambda^{-1}\mathbf{d}r(s')\right)dsds'\\
&\left(\xi(s,s'),\kappa(s,s')\right):=(\Lambda,r)^{-1}(s)(\Lambda,r)(s')\in\mathcal{O}\,\circledS\,E
\end{align*}
and one readily sees that $E_{np}$ is $\mathcal{O}$-invariant. Recall that the cocycle is
\[
c\left((\Lambda,r)^{-1}\right)=\left(\Lambda^{-1}\mathbf{d}\Lambda,\Lambda^{-1}\mathbf{d}r,\Lambda^{-1}r\right).
\]
Thus, a straightforward and maybe useful generalization of $E_{np}$ is
\[
E_{np}(\Lambda,r):=\iint_{\mathcal{D}}U\Big(\xi(s,s'),\kappa(s,s'),c\left((\Lambda,r)^{-1}\right)(s),c\left((\Lambda,r)^{-1}\right)(s')\Big)dsds'.
\]
Using Theorem \ref{coadjoint_lagrangian_reduction} with $L_0$ we obtain the same affine Euler-Poincar\'e equations \eqref{generalized_molec_strand}, where all derivatives are total derivatives. One can equivalently use the modified Euler-Poincar\'e approach and obtain the equations
\begin{equation}\label{generalized_molec_strand_modified}
\left\lbrace\begin{array}{l}
\displaystyle\vspace{0.2cm}\left(\prt_t-\operatorname{ad}^*_{\omega}\right)\frac{\delta l}{\delta\omega}+\lp \operatorname{div} - \ad^*_\Om\rp\frac{\delta l}{\delta\Omega}=\frac{\delta l}{\delta\gamma}\diamond\gamma+\frac{\delta l}{\delta\Gamma}\diamond\Gamma+\frac{\delta l}{\delta\rho}\diamond\rho
\,,\\
\displaystyle\vspace{0.2cm} \qquad+\int \left[\xi(s,s')\frac{\partial U}{\partial\xi}(s',s)-\frac{\partial U}{\partial\xi}(s,s')\xi(s',s)-\kappa(s,s')\diamond\frac{\partial U}{\partial\kappa}(s,s')\right]ds'\\
\displaystyle\left(\prt_t+\omega\right)\frac{\delta l}{\delta\gamma}+\left(\operatorname{div}+ \Omega\right)\frac{\delta l}{\delta\Gamma}=\frac{\delta l}{\delta\rho}+\int\left[\xi(s,s')\frac{\partial U}{\partial\kappa}(s',s)-\frac{\partial U}{\partial\kappa}(s,s')\right]ds'
\,.
\end{array}\right.
\end{equation}
Note that here the derivatives are not total derivatives, see the discussion in \S\ref{sec:recovering_EP}. One can treat the Hamiltonian side in a similar way. As we have seen, the motion is Hamiltonian on affine coadjoint orbits.

\subsubsection{Conservation laws and spatial formulation}

In this paragraph, we generalize the approach of Section~\ref{sec:Conservation} and reformulate the equations \eqref{generalized_molec_strand} for the generalized charged strand as a conservation law.
We first need a $n$-dimensional generalization of formula \eqref{diffAd2}. Given a Lie group $G$, a map $g:\mathcal{D}\rightarrow G$  defined on a $n$-dimensional manifold $\mathcal{D}$, $s \in \mathcal{D}$,  and a $\mathfrak{g}^*$-valued vector field $w$ on $\mathcal{D}$, we have
\begin{equation}\label{Ad*_div_formula}
\operatorname{Ad}^*_{g}\left[\operatorname{div}\left(\operatorname{Ad}^*_{g^{-1}}w\right)\right]=\operatorname{div}w-\operatorname{ad}^*_{\sigma_i}w^i=:\operatorname{div}^\sigma w,\quad\sigma:=g^{-1}\mathbf{d}g\in\Omega^1(\mathcal{D},\mathfrak{g}).
\end{equation}
Using this formula, \eqref{diffAd2}, the expression of $\operatorname{ad}^*$ associated to the semidirect product $\mathcal{O}\,\circledS\,E$, and the equalities
\[
(\omega,\gamma)=(\Lambda,r)^{-1}(\dot\Lambda,\dot r),\quad (\Omega,\Gamma)=(\Lambda,r)^{-1}\mathbf{d}(\Lambda,r),
\]
we find
\begin{align*}
&\operatorname{Ad}^*_{(\Lambda,r)}\frac{\partial}{\partial t}\left[\operatorname{Ad}^*_{(\Lambda,r)^{-1}}\left(\frac{\delta l}{\delta\omega},\frac{\delta l}{\delta\gamma}\right)\right]\\
&\qquad\qquad=\frac{\partial}{\partial t}\left(\frac{\delta l}{\delta\omega},\frac{\delta l}{\delta\gamma}\right)+\left(-\operatorname{ad}^*_\omega\frac{\delta l}{\delta\omega}+\gamma\diamond\frac{\delta l}{\delta\gamma},\omega\frac{\delta l}{\delta\gamma}\right),
\end{align*}
and
\begin{align*}
&\operatorname{Ad}^*_{(\Lambda,r)}\operatorname{div}\left[\operatorname{Ad}^*_{(\Lambda,r)^{-1}}\left(\frac{\delta l}{\delta\Omega},\frac{\delta l}{\delta\Gamma}\right)\right]\\
&\qquad\qquad=\operatorname{div}\left(\frac{\delta l}{\delta\Omega},\frac{\delta l}{\delta\Gamma}\right)+\left(-\operatorname{ad}^*_\Omega\frac{\delta l}{\delta\Omega}+\Gamma\diamond\frac{\delta l}{\delta\Gamma},\Omega\frac{\delta l}{\delta\Gamma}\right).
\end{align*}

Thus, equations \eqref{generalized_molec_strand_modified} can be rewritten in the form of a conservation law, namely
\begin{align}
\frac{\partial}{\partial t} \left[\operatorname{Ad}^*_{(\Lambda,r)^{-1}}\left(\frac{\delta l}{\delta\omega},\frac{\delta l}{\delta\gamma}\right)\right]+
& \operatorname{div}\left[\operatorname{Ad}^*_{(\Lambda,r)^{-1}}\left(\frac{\delta l}{\delta\Omega},\frac{\delta l}{\delta\Gamma}\right)
\right] 
\nonumber 
\\ 
& = 
\operatorname{Ad}^*_{(\Lambda,r)^{-1}}\left(\frac{\delta l}{\delta\rho}\diamond\rho \, , \, \frac{\delta l}{\delta\rho}\right) 
\, . 
\label{consgen0} 
\end{align} 
Using \eqref{coadjoint_action_group}, the right hand side simplifies to 
\begin{align*} 
\operatorname{Ad}^*_{(\Lambda,r)^{-1}}\left(\frac{\delta l}{\delta\rho}\diamond\rho \, , \, \frac{\delta l}{\delta\rho}\right) 
 &= 
\lp 
{\rm Ad}^*_{\Lambda^{-1}} \left(\frac{\delta l}{\delta\rho}\diamond\rho \right) 
+ r \diamond \lp \Lambda \frac{\delta l}{\delta\rho} \rp 
\, , \, 
\Lambda \frac{\delta l}{\delta\rho}
\rp 
\\ 
&= 
\lp
\left(\Lambda \frac{\delta l}{\delta\rho}\diamond \Lambda \rho \right) 
+ r \diamond \lp \Lambda \frac{\delta l}{\delta\rho} \rp 
\, , \, 
\Lambda \frac{\delta l}{\delta\rho}
\rp 
= \left(0, \Lambda \frac{\delta l}{\delta\rho}\right) 
\, , 
\end{align*}
since $\rho = \Lambda^{-1}r $.
Note that this is the exact equivalent of the simplification \eqref{Tffin} 
derived at the beginning of the paper.

Such a conservation law is valid for each solution of the affine Euler-Poincar\'e equation 
\eqref{AEPSD} associated to a $G^c_0$-invariant Lagrangian $L_0:TG\rightarrow\mathbb{R}$. In particular, it is valid for the Kirchhoff's theory as we saw at end of \S \ref{sec:Elastic-Rod}.

A short computation shows that, in general, the previous conservation law reads
\begin{equation}\label{general_conservation_law_0}
\frac{\partial}{\partial t}\left[\operatorname{Ad}^*_{g^{-1}}\frac{\delta l}{\delta\xi}\right]+\mathbf{d}c^T\left(g\frac{\delta l}{\delta a}\right)=0.
\end{equation}
When $a_0$ is not necessarily zero, the previous formula becomes
\begin{equation}\label{general_conservation_law_a_0}
\frac{\partial}{\partial t}\left[\operatorname{Ad}^*_{g^{-1}}\frac{\delta l}{\delta\xi}\right]+\mathbf{d}c^T\left(g\frac{\delta l}{\delta a}\right)=\operatorname{Ad}^*_{g^{-1}}\left(\frac{\delta l}{\delta a}\diamond g^{-1}a_0\right).
\end{equation}

\subsubsection{The fixed filament and its conservation law}\label{cocycle_fixed_filament}

The equations \eqref{rigid_bodies} for a fixed filament can also be obtained by affine Euler-Poincar\'e reduction. It suffices to apply Theorem \ref{coadjoint_lagrangian_reduction} with the group $G=\mathcal{F}(\mathcal{D},\mathcal{O})\ni\Lambda$, acting on the vector space  $\Omega^1(\mathcal{D},\mathfrak{o})\times\mathcal{F}(\mathcal{D},E)\ni(\Omega,\rho)$ by the affine action
\[
(\Omega,\rho)\mapsto \theta_\Lambda(\Omega,\rho):=(\operatorname{Ad}_\Lambda\Omega+\Lambda\mathbf{d}\Lambda^{-1},\Lambda\rho).
\]
Note that the cocycle is $c(\Lambda)=\left(\Lambda\mathbf{d}\Lambda^{-1},0\right)$. Using the expressions
\[
(u,f)\diamond (\Omega,\rho)=\operatorname{ad}^*_{\Omega_i}u^i+f\diamond\rho,
\]
\[
\mathbf{d}c(\omega)=(-\mathbf{d}\omega,0)\quad\text{and}\quad \mathbf{d}c^T(u,f)=\operatorname{div}(u),
\]
the affine Euler-Poincar\'e equations \eqref{AEP} become
\begin{equation}\label{AEP_fixed_fil}
\left(\partial_t-\operatorname{ad}^*_\omega\right)\frac{\delta l}{\delta\omega}+\left(\operatorname{div}-\operatorname{ad}^*_\Omega\right)\frac{\delta l}{\delta\Omega}=\frac{\delta l}{\delta\rho}\diamond\rho
\end{equation}
and the advection equations are
\begin{equation}
\label{AEP_fixed_fil_advected}
\left\{
\begin{array}{l}
\vspace{0.2cm}\partial_t\Omega+\operatorname{ad}_\omega\Omega=\mathbf{d}\omega\,,\\
\partial_t\rho+\omega\rho=0\,.
\end{array}
\right.
\end{equation}
\medskip

Recall from \S\ref{fixedfilament} that the Lagrangian for a fixed filament is of the form
\[
l=l_{loc}(\bom,\bOm)+l_{np}(\xi,\brho),
\]
\[
l_{loc}(\bom,\bOm)=K(\bom)-\frac{1}{2}\int f(\bOm(s))ds, \quad l_{np}(\xi,\brho)=-\iint U(\brho(s),\xi(s,s'))dsds'
\]
where
\[
f:\mathbb{R}^3\rightarrow\mathbb{R},\quad U:\mathbb{R}^3 \times  SO(3)\rightarrow\mathbb{R},\quad \xi(s,s'):=\Lambda^{-1}(s)\Lambda(s').
\]
Using the relations $\omega=\Lambda^{-1}\dot\Lambda$, $\Omega=\Lambda^{-1}\Lambda'$, and $\brho=\Lambda^{-1}\brho_0$, where $\brho_0(s):=\br(s)=(s,0,0)^T$, we obtain that $l$ is induced by a $SO(2)$-invariant Lagrangian $L_{(0,\br)}=L_{(0,\br)}(\Lambda,\dot\Lambda)$. Note that $SO(2)$ is precisely the isotropy group of $(0,\br)$ relative to the affine action.
\medskip

These observations allow us to obtain the equations for the fixed filament by the affine reduction processes described in Theorems \ref{coadjoint_lagrangian_reduction} and \ref{coadjoint_reduction}. Using the general framework involving $\mathcal{D}$ and $\mathcal{O}\,\circledS\,E$, we obtain the equations
\[
\left(\partial_t-\operatorname{ad}^*_\omega\right)\frac{\delta l}{\delta\omega}+\left(\operatorname{div}-\operatorname{ad}^*_\Omega\right)\frac{\delta l}{\delta\Omega}=\frac{\delta l}{\delta\rho}\diamond \rho+\int \left[\xi(s,s')\frac{\partial U}{\partial\xi}(s',s)-\frac{\partial U}{\partial\xi}(s,s')\xi(s',s)\right]ds'
\]
which coincides with \eqref{modified_EP_fixed_filament} in the case of the fixed filament. 
Using total derivatives, these equations can be rewritten as \eqref{AEP_fixed_fil}.
\medskip

The general formula \eqref{general_conservation_law_a_0} yields the conservation law
\[
\frac{\partial}{\partial t}\left[\operatorname{Ad}^*_{\Lambda^{-1}}\frac{\delta l}{\delta\omega}\right]+\operatorname{div}\left[\operatorname{Ad}^*_{\Lambda^{-1}}\frac{\delta l}{\delta\Omega}\right]=\operatorname{Ad}^*_{\Lambda^{-1}}\left(\frac{\delta l}{\delta\rho}\diamond \rho\right).
\]
From the general theory it follows that the solution of the advection equations \eqref{AEP_fixed_fil_advected}   in terms of $\Lambda$ are given by $\Omega=\Lambda^{-1}\mathbf{d}\Lambda$ and $\rho=\Lambda^{-1}\rho_0$.
\medskip

For the fixed filament, we choose $\mathcal{D}=[0,L]$, $E=\mathbb{R}^3$, $\mathcal{O}=SO(3)$, $\brho_0(s)=\br(s)=(s,0,0)^T$ and we get
\begin{equation}
\label{AEP_fixed_fil_conservation}
\frac{\partial}{\partial t}\left[\operatorname{Ad}^*_{\Lambda^{-1}}\frac{\delta l}{\delta\bom}\right]+\frac{\partial}{\partial s}\left[\operatorname{Ad}^*_{\Lambda^{-1}}\frac{\delta l}{\delta\bOm}\right]=\operatorname{Ad}^*_{\Lambda^{-1}}\left(\frac{\delta l}{\delta\brho}\times \brho\right).
\end{equation}
Note that in this case, the torque does not vanish. The explanation is that
the initial value $\brho_0$ of $\brho$ is not zero, so we need to use \eqref{general_conservation_law_a_0} instead of \eqref{general_conservation_law_0}.
..
Observe that we can write
\[
\operatorname{Ad}^*_{\Lambda^{-1}}\left(\frac{\delta l}{\delta\brho}\times \brho\right)=\Lambda\frac{\delta l}{\delta\brho}\times \Lambda\brho=\Lambda\frac{\delta l}{\delta\brho}\times\left(\begin{array}{l} s\\0\\0\end{array}\right).
\]
More generally, the right hand side is 
\[
\left(\Lambda\frac{\delta l}{\delta\brho}\right)\times\br,
\]
where $\br$ describes the fixed filament.

Note that the conservation law \eqref{AEP_fixed_fil_conservation} does not appear in \S\ref{fixedfilament}. It is a particular case of the general formula   \eqref{general_conservation_law_a_0}. We believe that the derivation of this law through the affine Euler-Poincar\'e theory is interesting and shows the breadth of application of our theories.

\section{New variables: Coordinate change and horizontal-vertical split}
\label{sec:Coordinatechange}
In this section, we show that a drastic simplification of the equations arises under a particular change of variables. We shall assume that the Lagrangian $l $ is only local. As far as we know, there is no general theory that can deal with the nonlocal term in the context of field theory, of which the present section is a forerunner; the field theoretic approach is developed for local Lagrangians in the next section.

We first consider the case of strands. This change of variables will then be extended to the general setting of the previous section where $[0,L] $ is replaced by a manifold $\mathcal{D}$ and $\operatorname{SE}(3)$ by an arbitrary semidirect product associated to a representation.

\subsection{Motivation in terms of covariant derivatives}

We can see from \eqref{bundle.coords} that $\brho$, $\bGam$, and $\bgam$ satisfy the following relations 
\begin{equation*}
\lp \prt_s + \bOm \times\rp \brho = \bGam, \qquad \lp \prt_t + \bom \times\rp \brho = \bgam.
\end{equation*}
Thus the reduced variables \eqref{bundle.coords} lead naturally to two differential operators which can be interpreted as covariant derivatives,
\begin{equation}\label{cov.deriv}
\DD{}{s}= \lp \prt_s + \bOm \times\rp, \qquad \DD{}{t}= \lp \prt_t + \bom \times\rp.
\end{equation}

With this interpretation we regard $\bGam$ and $\bgam$ as covariant tangent vectors above $\brho$,
\begin{equation}
\label{cov.tang}
\DD{\brho}{s} = \bGam, \qquad \DD{\brho}{t} = \bgam
\,.
\end{equation}

The operators from \eqref{cov.deriv} also appear in the equations of motion \eqref{Final_Euler_Poincare_equations} since we can write the second Euler-Poincar\'e equation in the form
\begin{equation}
\label{cov.EP}
\DD{}{t}\dede{l}{\bgam} + \DD{}{s}\dede{l}{\bGam}
 -\dede{l}{\brho}=0
 \,.
\end{equation}

When take \eqref{cov.tang} and \eqref{cov.EP} together we see that \eqref{cov.EP} is in the form of the Euler-Lagrange equations where the partial derivatives have been replaced by covariant derivatives.  With this interpretation in mind we can ask whether, by a change of variables, we can transform \eqref{cov.EP} to the canonical Euler-Lagrange form.  In this section we find that such a change of variables does exist, and we give it explicitly.  This line of enquiry leads us to consider in the subsequent Sections how the two sets of coordinates are related from a geometric point of view.

\subsection{The case of charged strands}

Consider the coordinate change 
\begin{align}
& 
\mathcal{F}([0, L], \mathfrak{so}(3)) \times \mathcal{F} ([0,L], \mathbb{R}^3) \times \Omega^1( [0, L], \mathfrak{so}(3)) \times \Omega^1([0,L], \mathbb{R}^3) \times \mathcal{F}([0, L], \mathbb{R}^3) \nonumber  \\ 
&   
\qquad \qquad \qquad  \ni    \left(\bom, \bgam, \bOm,  \bGam, \brho \right) \mapsto  
\left( \brho, \brho_s, \brho_t , \bom , \bOm \right) \in  \label{coordchange} \\
&
\mathcal{F}([0, L], \mathbb{R}^3) \times \Omega^1([0,L], \mathbb{R}^3) \times \mathcal{F}([0, L], \mathbb{R}^3) \times \mathcal{F}([0, L], \mathfrak{so}(3)) \times \Omega^1( [0, L], \mathfrak{so}(3)) \nonumber
\end{align}
where we have defined two new variables
\begin{equation}
 \brho_s = \bGam - \bOm \times \brho, \quad \brho_t = \bgam  - \bom \times \brho
\, .
\label{newvar}
\end{equation}

We shall show that the equations of motion \eqref{hpvarsig} and
\eqref{hpvarpsi} have  simple expressions if one uses  \emph{horizontal} and \emph{vertical} coordinates. As far as we know, this transformation has not been noticed before, in either nonlocal or local setting.

\begin{notation} {\rm Assume that the Lagrangian $l$ is local in the variables $(\bom,\bgam,\bOm,\bGam,\brho)$, that is, $l = l_{loc} $. 
We shall denote by $\bar l$ the integrand of the Lagrangian $l$ in terms of the new variables given by  \eqref{coordchange}, that is, we have}
\end{notation}
\[
l(\bom,\bgam,\bOm,\bGam,\brho)=\int_0^L\bar l(\brho(s),\brho_s(s),\brho_t(s),\bom(s),\bOm(s)) \mbox{d}s.
\]

\subsection{Change of coordinates}

The action principle for the  Lagrangian $l=l_{loc}$ yields
\begin{align}
0&=\delta \int l(\bom,\bgam,\bOm,\bGam,\brho)\, \mbox{d} t 
\nonumber \\
&=
\int
\left[\left\langle 
\frac{\delta l}{\delta \brho}
\, , \,
\delta \brho
\right\rangle
+
\left<
\frac{\delta l}{\delta \bgam}
\, , \,
\delta \bgam
\right>
+
\left<
\frac{\delta l }{\delta \bGam}
\, , \,
\delta \bGam
\right>
+
\left<
\frac{\delta l}{\delta \bom}
\, , \,
\delta \bom
\right>
+
\left<
\frac{\delta l}{\delta \bOm}
\, , \,
\delta \bOm
\right> \right]\mbox{d} t
\nonumber
\\
& = \delta \int \int_0^L \bar{l} (\brho, \brho_s, \brho_t, \bom, \bOm) \mbox{d}s\, \mbox{d}t  \label{changevarminaction} \\
&=
\int \left[
\left<
\frac{\delta \bar l}{\delta \brho}
\, , \,
\delta \brho
\right>
+
\left<
\frac{\delta \bar l}{\delta \brho_s}
\, , \,
\delta \brho_s
\right>
+
\left<
\frac{\delta \bar l }{\delta \brho_t}
\, , \,
\delta \brho_t
\right>
+
\left<
\frac{\delta \bar l}{\delta \bom}
\, , \,
\delta \bom
\right>
+
\left<
\frac{\delta \bar l}{\delta \bOm}
\, , \,
\delta \bOm
\right> \right] \mbox{d} t.
\nonumber
\end{align}
Define free variations $\bpsi(s)=\Lambda(s)^{-1} \delta \br(s)$ and $\Sigma(s)=\Lambda(s)^{-1} \delta \Lambda(s)$. As usual, $\Psi$ denotes the antisymmetric matrix
that is obtained from $\bpsi$ by the hat map. Then, the following theorem holds.
\begin{thm} \label{changevariables} 
The variations in $\de\brho_s$ and $\de\brho_t$ yield  dynamical equations in the following form
\begin{eqnarray}
\label{verlp}
\lp\prt_s + \bOm\times \rp\dede{\bar l}{\bOm} + \lp\prt_t + \bom\times \rp\dede{\bar l}{\bom} &=& 0\,,\\
\label{horlp}
\dede{\bar l}{\brho} - \prt_t\dede{\bar l}{\brho_t} - \prt_s\dede{\bar l}{\brho_s}
&=& 0\,.
\end{eqnarray}
\end{thm}

\begin{remark}{\rm 
The derivatives in the equations \eqref{verlp} and \eqref{horlp}  have now formally decoupled, although the equations themselves must be solved simultaneously because the Lagrangian $l$ depends on all the variables.  Also note that equation \eqref{horlp} is equivalent, for local Lagrangians, to \eqref{hpvarpsi} with the covariant derivatives replaced by partial derivatives (but relative to the new variables).  This gives a new interpretation to the right-hand side of \eqref{hpvarsig} as being terms that arise from the induced covariant derivative.}
\end{remark}

\begin{proof}
First, variations $\delta \brho_t$ and $\delta \brho_s$ are computed from \eqref{newvar} as follows:
\begin{eqnarray}
\delta \brho_t = \de \bgam - \delta \bom \times \brho - \bom \times \delta \brho
\label{deltabrhot}
\,,\\
\delta \brho_s = \de \bGam - \delta \bOm \times \brho - \bOm \times \delta \brho
\label{deltabrhos}
\, .
\end{eqnarray}
Then, using the identities
\begin{eqnarray*}
\delta \bom= \dot \bsigma + \bom \times \bsigma
\,,\\
\delta \bOm= \bsigma' + \bOm \times \bsigma
\,,\\
\delta \brho=- \bsigma \times \brho + \bpsi
\,,\\
\delta \bgam= \dot \bpsi + \bom \times \bpsi  - \bsigma  \times \bgam
\,,\\
\delta \bGam= \bpsi' +\bOm \times \bpsi - \bsigma \times \bGam
\,,
\end{eqnarray*}
we find, for example, from the term involving the derivatives with respect to $\brho_t$,
\begin{align}
\left<
\frac{\delta \bar l }{\delta \brho_t}
\, , \,
\delta \brho_t
\right> = &
\Bigg<
\frac{\delta \bar l }{\delta \brho_t}
\, , \,
\dot \bpsi + \bom \times \bpsi  - \bsigma  \times \bgam
\nonumber
\\
&\quad
- \Big(
 \dot \bsigma + \bom \times \bsigma
\Big) \times \brho
-
\bom \times
\Big(
- \bsigma \times \brho+\bpsi
\Big)
\Bigg>
\nonumber
\\
=&
\left< - \frac{\partial}{\partial t}
\frac{\delta \bar l }{\delta \brho_t}
\, , \,
\bpsi
\right>
\nonumber \\
&\quad
+
\left< \frac{\partial}{\partial t}
\Bigg( \rho \times
\frac{\delta \bar l}{\delta \brho_t}
\Bigg)
- \bgam \times \frac{\delta \bar l}{\delta \brho_t} -
\Big(
\brho \times \bom
\Big)
\times
 \frac{\delta l }{\delta \brho_t}
 \, , \bsigma
\right>
\,,
\label{varrhodot}
\end{align}
where we have used the Jacobi identity simplifying two triple cross products. We now employ the
.. We now employ the
kinematic condition for the derivative of $\rho$,
\[
\partial_t \brho=\bgam - \bom \times \brho
\,,
\]
to simplify the $\bsigma$ term in \eqref{varrhodot} and obtain the following simple condition
\begin{equation}
\left<
\frac{\delta \bar l}{\delta \brho_t}
\, , \,
\delta \brho_t
\right>
=
\left< - \frac{\partial}{\partial t}
\frac{\delta \bar l}{\delta \brho_t}
\, , \,
\bpsi
\right> +
\left< - \brho \times \frac{\partial}{\partial t}
\frac{\delta \bar l}{\delta \brho_t}
\, , \,
\bsigma
\right>
\, .
\label{varrhodot2}
\end{equation}
Analogously,
\begin{equation}
\left<
\frac{\delta \bar l}{\delta \brho_s}
\, , \,
\delta \brho_s
\right>
=
\left< - \frac{\partial}{\partial s}
\frac{\delta \bar l}{\delta \brho_s}
\, , \,
\bpsi
\right> +
\left< - \brho \times \frac{\partial}{\partial s}
\frac{\delta \bar l}{\delta \brho_s}
\, , \,
\bsigma
\right>
\, .
\label{varrhos2}
\end{equation}
On completing the variational principle \eqref{changevarminaction} for all variables, one sees that the only terms containing $\bpsi$ are the derivatives with respect to $\brho$, $\brho_s$ and $\brho_t$. Due to \eqref{varrhodot2} and \eqref{varrhos2}, these remaining terms yield \eqref{horlp}.

On collecting the terms proportional to $\bsigma$, we notice another cancellation. As is evident already from \eqref{varrhodot2} and \eqref{varrhos2}, all the terms involving cross products with respect to $\brho$ will cancel, as they will each be multiplied by the left hand side of \eqref{horlp} which vanishes.
Thus, derivatives with respect to $\brho$, $\brho_s$ and $\brho_t$ will not contribute to the terms proportional to $\bsigma$, so that collecting those terms will yield exactly \eqref{verlp}.
\end{proof}

\medskip

There is another approach to performing the change of variables that highlights the decoupling.  We key point is that we recognize two pieces of information we know about the variations $\de\brho$, $\de\brho_s$ and $\brho_t$.  First we consider the expression for $\de\brho$ in terms of the free variations $\bpsi$ and $\bsigma$.  The relation is given by
\[
\de\brho = \bpsi - \bsigma \times \brho
\,.
\]

This relation can be interpreted as saying that we can select any two of the variations $\bsigma$, $\bpsi$, and $\de \brho$ as a free variation and the third variation is then determined.  We find in practice that there are quantities such as $\de\bOm$ that only depend on $\bsigma$.  Therefore any selection of free variations must include $\bsigma$.  This leaves us with a choice of $\bpsi$ or $\de \brho$ as the choice for the second free variation.  It is interesting to consider the choice of $\de \brho$.  Indeed, since we have the relations
\[
\brho_s = \prt_s \brho
\,, \qquad 
\brho_t = \prt_t \brho
\,.
\]

we can express the variations $\de\brho_s$ and $\brho_t$ in terms of our free variation $\de \brho$.
\[
\de\brho_s = \de\prt_s \brho
= \prt_s \de \brho
\]

Similarly, $\de \brho_t = \prt_t \de \brho$.  Since $\de\bOm$ and $\de \bom$ only depend on $ \bsigma$ we have a complete description of the variations in terms of $\bsigma$ and $\de \brho$ which are given by
\begin{eqnarray} \label{alternativevariations1}
\de \bom = \dot\bsigma + \bom \times \bsigma
\,, &\qquad& 
\de \bom = \bsigma' + \bOm \times \bsigma
\,,\\
\label{alternativevariations2} 
\de\brho_s = \prt_s \de \brho 
\,, &\qquad& \de \brho_t = \prt_t \de \brho
\,,
\end{eqnarray}

which is obviously augmented by the trivial relation $\de \brho = \de\brho$.  An alternative proof of Theorem \ref{changevariables} can be given as follows

\vspace{0.5cm}

\begin{proof}  Using the variations \eqref{alternativevariations1}, \eqref{alternativevariations2} we obtain, for example, the following calculation in the variational principle,
\[
\scp{\dede{\bar l}{\brho_s}}{\de\brho_s} = \scp{\dede{\bar l}{\brho_s}}{\prt_s \de \brho} 
= -\scp{\prt_s\dede{\bar l}{\brho_s}}{\de \brho}.
\]

The terms arising from $\de \bOm$ and $\de\bom$ are identical to before and only depend on $\bsigma$.  Therefore, we obtain the following equation from stationarity under the $\bsigma$ variation,
\[
\lp \prt_t + \bom \times \rp \dede{\bar l}{\bom} + \lp \prt_s + \bOm\times\rp \dede{\bar l}{\bOm} = 0
\,.
\]
The second equation comes from terms proportional to $\de\brho$ which is
\[
\prt_t \dede{\bar l}{\brho_t} + \prt_s \dede{\bar l}{\brho_s} - \dede{\bar l}{\brho} = 0
\,.
\]
These are the required equations in Theorem \ref{changevariables}.
\end{proof}

\begin{rema}{\rm  Notice that this alternative proof does not require any cancelation of terms after the equations are derived.  Thus, the variations do all the work for us.  This opens up an interesting question.  In some sense the choice of $\de \brho$ as a free variation is optimal since no extra terms appear in the resulting equations of motion.  We might also refer to the heavy top at this point and ask whether a similar change of variables might simplify the heavy top equations.  The answer, alas, is negative but is nevertheless instructive.  The crucial property that we used was to regard $\de\brho$ as a free variation.  Now, suppose we have an advected quantity ${\bf a }= \Lambda^{-1}{\bf a }_0$.  This case appears in the heavy top as well as often occurring in fluid dynamics.  Could we consider $\de {\bf a }$ as a free variation?  Unfortunately the variation $\de{\bf a }$ is given by
\[
\de{\bf a } = -\,\bsigma \times {\bf a }
\,.
\]

Therefore $\de {\bf a }$ is determined by $\bsigma$ and we cannot interpret $\de{\bf a }$ as a free variation.  We shall investigate the geometric structure required for this approach in \ref{sec:Lagrange-Poincare}.}
\end{rema}

\begin{remark}{\rm 
This change of variables is not available in the classical Kirchhoff approach because the variable $\brho$ is absent in the classical approach.}
\end{remark}

\subsection{The general case}

We now generalize the previous results to the general situation described in Subsection \ref{n_dimensional_generalization}. Recall that in this case we have $(\Lambda,r)\in\mathcal{F}(\mathcal{D},S)$, $(\Omega,\Gamma)\in\Omega^1(\mathcal{D},\mathfrak{s})$, and $\rho\in\mathcal{F}(\mathcal{D},E)$, where $S=\mathcal{O}\,\circledS\,E$ is the semidirect product of a Lie group $\mathcal{O}$ with a vector space $E$.

Consider the variable $\rho$. Recall from \eqref{kincoc} that we have the kinematic equation
\[
\dot\rho=\gamma-\omega\rho
\,.
\]
Assuming that the initial value of $\rho$ is zero, we have
\[
\mathbf{d}\rho=\mathbf{d}(\Lambda^{-1}r)=-\Lambda^{-1}\mathbf{d}\Lambda\Lambda^{-1}r+\Lambda^{-1}\mathbf{d}r=\Gamma-\Omega\rho
\,.
\]
This motivates us to define the new variables $\rho_s\in\Omega^1(\mathcal{D},E)$ and $\rho_t\in\mathcal{F}(\mathcal{D},E)$ which will play the role of space and time derivatives of $\rho$. They are naturally defined by
\begin{equation}\label{change_variables}
\rho_s=\Gamma-\Omega\rho
\,,\quad\text{and}\quad 
\rho_t=\gamma-\omega\rho
\,.
\end{equation}
This change of variables defines a diffeomorphism from the variables $(\omega,\gamma,\Omega,\Gamma,\rho)$ to the variables $(\omega,\Omega,\rho_s,\rho_t,\rho)$, and generalizes \eqref{newvar}. In terms of the new variables, the local Lagrangian  is denoted by $\bar l$ and we have
\[
\int_{ \mathcal{D}} \bar l(\rho,\rho_s,\rho_t,\om,\Om) \mbox{d} s
=
l(\omega,\gamma,\Omega,\Gamma,\rho)
\,.
\]
There are two equivalent points of view to obtain the equations of motion in terms of $\bar l$.

The first one is to use a variational principle, as done before in the particular case of the charged strand. Using the constrained variations of $\omega,\gamma,\Omega,\Gamma,\rho$ given by the affine Euler-Poincar\'e principle, we obtain the constrained variations
\[
\delta\omega=\dot\Sigma+[\omega,\Sigma]
\,,\quad
\delta\Omega=\mathbf{d}\Sigma+[\Omega,\Sigma]
\,,
\]
\[
\delta\rho_t=\dot\Phi-\dot\Sigma\rho-\Sigma\rho_t
\,,\quad 
\delta\rho_s=\mathbf{d}\Phi-\mathbf{d}\Sigma\rho-\Sigma\rho_s
\,,
\]
and
\[
\delta\rho=\Phi-\Sigma\rho
\,.
\]

The second point of view is to compute the functional derivatives of $l$ in terms of those of $\bar l$. We find
\[
\frac{\delta l}{\delta \omega}=\frac{\delta \bar l}{\delta \omega}-\rho\diamond \frac{\delta \bar l}{\delta \rho_t},\quad \frac{\delta l}{\delta \Omega}=\frac{\delta \bar l}{\delta \Omega}-\rho\diamond \frac{\delta \bar l}{\delta \rho_s},
\]
\[
\frac{\delta l}{\delta \gamma}=\frac{\delta \bar l}{\delta \rho_t},\quad \frac{\delta l}{\delta \Gamma}=\frac{\delta \bar l}{\delta \rho_s}
\,,
\]
and
\[
\frac{\delta l}{\delta \rho}=\frac{\delta \bar l}{\delta \rho}+\Omega_i\frac{\delta \bar l}{\delta \rho_{s\,i}}+\omega\frac{\delta \bar l}{\delta \rho_t}
\,.
\]
These two ways lead to the same equations
\begin{equation}\label{equ_new_variables}
\left\lbrace\begin{array}{l}
\displaystyle\vspace{0.2cm}\left(\frac{d}{dt}-\operatorname{ad}^*_\omega\right)\frac{\delta \bar l}{\delta\omega}+\operatorname{div}^\Omega \frac{\delta \bar l}{\delta\Omega}=0
\,,\\
\displaystyle\vspace{0.2cm}
\frac{d}{dt}\frac{\delta \bar l}{\delta\rho_t}+\operatorname{div}\frac{\delta \bar l}{\delta\rho_s}-\frac{\delta \bar l}{\delta\rho}=0
\,,
\end{array}\right.
\end{equation}
where $\operatorname{div}^\Omega : \mathfrak{X}( \mathcal{D}, \mathfrak{o}^\ast) \rightarrow \mathcal{F}( \mathcal{D}, \mathfrak{o}^\ast) $ is defined by 
\begin{equation*}
\operatorname{div}^\Omega
w:=\operatorname{div}w-\operatorname{ad}^*_{\Omega_i}
w^i\in\mathcal{F}(\mathcal{D},\mathfrak{o}^*)
\,.
\end{equation*}
These equations coincide with \eqref{verlp} and \eqref{horlp} in the particular case $\mathcal{D}=[0,L]$ and $S=SE(3)$. The other equations for the advected variables are computed to be
\[
\left\lbrace\begin{array}{l}
\displaystyle\vspace{0.2cm}\dot\rho_s+\omega\rho_s=\mathbf{d}\rho_t+\omega\mathbf{d}\rho
\,,\\
\displaystyle\vspace{0.2cm}\dot\Omega+\operatorname{ad}_\omega\Omega=\mathbf{d}\omega
\,,\\
\dot\rho=\rho_t
\,.
\end{array}\right.
\]
We also know that $\mathbf{d}\rho=\rho_s$. Therefore, using the third equation, we obtain that the first equation is verified. Thus the last system can be replaced by
\[
\left\lbrace\begin{array}{l}
\displaystyle\vspace{0.2cm}\mathbf{d}\rho=\rho_s
\,,\\
\displaystyle\vspace{0.2cm}\dot\Omega+\operatorname{ad}_\omega\Omega=\mathbf{d}\omega
\,,\\
\dot\rho=\rho_t
\,.
\end{array}\right.
\]

\section{The bundle covariant Lagrange-Poincar\'e approach}
\label{sec:Lagrange-Poincare}
In this section we explain how the decoupled equations discussed above are covariant Lagrange-Poincar\'e equations.  The coordinate change will be interpreted as a transformation from the affine and modified Euler-Poincar\'e perspectives to the covariant Lagrange-Poincar\'e perspective.  The corresponding Lagrange-Poincar\'e equations are derived on principal fiber bundles by introducing a principal connection and splitting the configuration space into horizontal and vertical parts.  Two equations occur, one horizontal and one vertical.  In the case of strands we also have to take into account the continuous dependence of the variables on $s$.  This leads us to consider a {\em covariant} version of the equations.  We shall give various geometric structures that combine to give the required space. These geometric structures are introduced very effectively in the literature and the reviews of the various geometric objects are particularly based on \cite{CaRa03, CeMaRa2001,El2009}.

These field theoretical considerations work only for local Lagrangians. Since the classical infinite dimensional approach applies also to Lagrangians having a nonlocal part, it is clear that an extension of the theory presented below exists also in the field theoretic framework. We defer to future work a development of the Lagrange-Poincar\'e theory of Lagrangians depending on non-local variables.  The local part of the equations of motion in the Lagrange-Poincar\'e framework will, of course, have their non-local counterparts, equivalent to those derived in the Euler-Poincar\'e framework.

\subsection{Covariant state space}

In this paragraph we shall introduce the {\bfi covariant state space}.  The aim is to incorporate all the dynamical information into a single geometric object.   

We begin by noting that equations \eqref{Final_Euler_Poincare_equations} have an exchange symmetry in their $s$ and $t$ dependences.  Therefore, guided by the equations derived so far, we may treat $s$ and $t$ on an equal basis by introducing a {\em spacetime}, $X := I \times\mR$.  The dynamical quantities are then regarded as special  vector bundle maps $\lambda: TX \rightarrow TSE(3)$.  Such objects may be studied by considering the trivial fiber bundle
\[
\pi_{XP} : P:= X \times SE(3) \rightarrow X,\quad \pi_{XP}(x,\Lambda,\br):=x.
\]
The analogue of the state space $T Q$ in field theory is the  {\em first jet bundle}, $J^1P$, of $P$.  

\begin{defi}  Given a locally trivial fiber bundle $\pi_{XP}:P\rightarrow X$, the {\bfi first jet bundle} $\pi: J ^1P \rightarrow P$ of $P$ is the affine bundle over $P$ whose fiber at $p  \in P $ is 
\[
(J^1P)_p = \left\{ \lambda \in L\lp T_xX,T_pP \rp \mid T\pi_{XP} \circ \lambda = \id_{T_x X}\right\},
\]
where $L\lp T_xX,T_pP \rp$ denotes linear maps $T_xX\rightarrow T_pP$ and $x = \pi_{XP}(p)$.
\end{defi}

The space $J^1P$ is called the {\bfi covariant state space}.

\begin{remark} \label{rem:linearmaps} {\rm It might, at first sight, seem unnatural to form the dynamics on a space of linear maps.  After all, in canonical Lagrangian dynamics we consider tangent vectors,  $(q,\dot q)$.  However, in the canonical setting we could consider maps of the form $Tq: T\mR \rightarrow TQ$ where the state space is $TQ$ and $\mR$ is time.  Since $Tq$ is a linear map on each fiber of $T\mR$ we consider a basis of the image of each fiber given by $T_{(t,1)}q =: (q,\dot q)$, then $T_{(t,a)}q = (q,a\dot q)$ for all $a \in \mR$, which is just rescaling of time viewed in a geometric way.  When there is more than one independent variable we wish to capture the entire dynamics, independently of which direction is chosen in spacetime, therefore the notion of a linear map is the idea that generalizes most elegantly.  In what follows one should think of jets as giving the `velocities' in an arbitrary direction on spacetime.  Thus it turns out that the first jet bundle, $J^1P$, is a very natural state space since it is the analogue of the tangent bundle in the case where many independent variables are considered.}
\end{remark}

In field theory one only uses certain sections of $J ^1 P $, namely the {\em holonomic} or {\em first jet extensions of sections of} $P $.

\begin{defi} Let $\sigma: X  \rightarrow P$ be a section of $P$, that is, $\pi_{XP}\circ \sigma = \id_X$.  The {\bfi first jet extension} of $\sigma$ is the map $j^1\sigma: X \rightarrow J ^1P $ defined by $j^1\sigma(x)=T_x\sigma$ for all $x \in X $. 
\end{defi}

We see that $j ^1 \sigma (x) \in (J ^1P)_{ \sigma(x)} $, we differentiate the relation $\pi_{XP}\circ \sigma = \id_X$ to find $T\pi_{XP}\circ T\sigma = \id_{TX }$.  This verifies that $T\sigma \in J^1P$.

Given $X = I \times \mathbb{R}$ and $P = X \times SE(3)$, any section $\sigma$ reads 
\[
\sigma(x) = (x,\Lambda(x), \br(x)) \in \{x\} \times (SO(3) \,\circledS\, \mathbb{R}^3),
\]
where $x : =(s,t) $. In this case we also have $(J ^1 P)_{ \sigma(x)} \cong L( T_x X, T_{(\Lambda (x), \mathbf{r}(x))} SE(3))$. Using this identification, we can write
\[
j^1\sigma(x) = T_x \lp x, \Lambda, \br\rp
= \lp x, \Lambda(x), \br(x), \id_{T_xX},  \Lambda'(x) ds+ \dot \Lambda(x) dt, \br'(x) ds + \dot \br(x) dt \rp.
\]
From \eqref{rhodef} we conclude that the dependent variables that occur in the unreduced Euler-Lagrange dynamics are simply components of a first jet extension of a section of $\pi_{XP}$.

\subsection{Principal bundle structures}\label{Principal_bundle_structures}

Consider the natural principal $SO(3)$-bundle structure on $SE(3)$ given by the projection 
\[
\pi_{SE(3)}:SE(3)\rightarrow\mathbb{R}^3,\quad \pi_{SE(3)}\lp \Lambda, \br\rp = \Lambda^{-1}\br = \brho
\,.
\]
The action of $SO(3)$ is given by
\begin{equation}\label{SO(3)_action}
g(\Lambda, \br) = (g\Lambda, g\br).
\end{equation}

This principal bundle structure induces a principal $SO(3)$-bundle structure on the trivial fiber bundle $P = X \times SE(3)$.   We enforce the relationship $\pi_{XP}\lp g \cdot p\rp = \pi_{XP}\lp p\rp$ for all $p \in P$.  This relationship means that the group does not act on spacetime and this is reasonable because we do not want $SO(3)$ to act on $s$ or $t$ in the applications.  Explicitly, we have the following definition of the action of $SO(3)$ on $P$:
\[
g \lp x, \Lambda, \br\rp = \lp x, g\Lambda, g\br\rp.
\]
Using this action we easily see that the reduced space is given by
\[
\Sigma := P/SO(3)
=\lp X \times SE(3)\rp/SO(3)
= X \times \lp SE(3)/SO(3)\rp
= X \times \mR^3.
\]

The principal $SO(3)$-structure on $P$ is given by the projection induced on $P$ by $\pi_{SE(3)}$.  That is,
\[
\pi_{\Sigma P}:P\rightarrow\Sigma,\quad \pi_{\Sigma P}\lp x, \Lambda, \br\rp = \lp x, \pi_{SE(3)}\lp \Lambda, \br\rp\rp
= \lp x, \brho\rp.
\]

Define the trivial fiber bundle $\pi_{X\Sigma}: \Sigma \to X$ by
\[
\pi_{X\Sigma}\lp x, \brho\rp = x.
\]

To summarize, we are given a principal $SO(3)$-bundle structure $\pi_{SE(3)}$ on $SE(3)$ and a fiber bundle structure $\pi_{XP}$ on $P$.  From these we construct a new principal $SO(3)$-bundle structure, $\pi_{\Sigma P}$, on $P$ and a new fiber bundle structure, $\pi_{X \Sigma}$, on $\Sigma = P/SO(3)$.  We note that 
..  We note that 
\begin{equation} \label{crucialproperty}
\pi_{X\Sigma} \circ \pi_{\Sigma P} = \pi_{X P}.
\end{equation}
We summarize the considerations above in the diagram
\[
\begin{CD}
SO(3) &&  SO(3)\\
@VVV  @VVV \\
SE(3) @>>> P @>\pi_{XP}>> X\\
@V\pi_{SE(3)} VV  @V\pi_{\Sigma P} VV @VV\id V \\
\mR^3 @>>> \Sigma @>>\pi_{X\Sigma} > X. \\
\end{CD}
\]

The section $\sigma(x) = ( x, \Lambda(x), \br (x)) $ of $\pi_{XP}: P \rightarrow X$ induces the section $x \mapsto (x, \brho (x)) $ of $\pi_{X \Sigma} : \Sigma \rightarrow X $.

The tangent lift of the $SO(3)$ action yields a free action on the jet bundle, $J^1P$.
\begin{align*}
&g\lp x, \Lambda, \br, \id_{T_xX}, \Lambda' ds + \dot \Lambda dt, \br' ds + \dot \br dt\rp \\
&\quad=\lp x, g\Lambda, g\br, \id_{T_xX}, g\Lambda'ds + g \dot \Lambda dt, g \br' ds + g\dot \br dt\rp.
\end{align*}

The action is free because the action on $SO(3)$ is free.  Therefore we find that $J^1P$ is also a principal $SO(3)$-bundle.  In particular, $J^1P/SO(3)$ is a manifold.  We note that sometimes, for brevity, we omit the reference to $x$ and $\id_{T_xX}$ in the explicit representation.

\begin{remark}  {\rm At this point we could reduce by the $SO(3)$ action on $SE(3)$ to derive Euler-Poincar\'e equations.  The reduced variables are given as
\begin{eqnarray*}
\Lambda^{-1}T(\Lambda,\br) &=& (e, \Lambda^{-1}\br, \Lambda^{-1}\Lambda^\prime ds+ \Lambda^{-1}\dot\Lambda dt, \Lambda^{-1}\br^\prime ds + \Lambda^{-1}\dot \br dt)\\
&=& (\brho, \bOm ds+ \bom dt, \bGam ds + \bgam dt)
\,.
\end{eqnarray*}

This route is taken in the Euler-Poincar\'e picture and results in the equations derived above.  Again it is interesting to note that all of the dynamical quantities arise as components of a jet.}
\end{remark}

\subsection{Principal Connection}

We introduce a principal connection that is needed to split $J^1P$ into horizontal and vertical parts and discuss its induced geometric structure. Recall that a principal connection on a principal $G$-bundle $P$ is a $\mathfrak{g}$-valued one form on $P$ that satisfies
\[
A(\xi_P(p)) = \xi,\quad A(gv_p)=\operatorname{Ad}_gA(p)
\,,
\]
where $gv_p$ denotes the tangent lifted action of $G$ on $TP$ and $\xi_P$ is the infinitesimal generator
\[
\xi_P(p)=\left.\frac{d}{dt}\right|_{t=0}\operatorname{exp}(t\xi)p
\,.
\]
For our particular $SO(3)$-bundle $\pi_{\Sigma P}:P\rightarrow\Sigma$, we make the choice
\begin{equation}
\label{trivial_connection}
A\lp x, \Lambda, \br, v_x, v_\Lambda, \bu\rp =v_\Lambda\Lambda^{-1}\in\mathfrak{so}(3)
\,,
\end{equation}
for all $v_\Lambda\in T_\Lambda SO(3),\, \bu \in \mR^3, \, v_x \in T_xX$.  Note that this is the Maurer-Cartan connection for the structure group $SO(3)$. The choice of connection is actually arbitrary, but the particular choice above is well suited to the problem since it is not overly complicated.  Also we recall that any vector $v_\Lambda\in T_\Lambda SO(3)$ can be written $v_\Lambda=\Lambda \eta$ where $\eta\in\mathfrak{so}(3)$. The connection decomposes $TP$ into  the {\bfi horizontal} and {\bfi vertical subbundles} as follows: 
\begin{align}
\label{A_vertical}
\ver\, P &= \ker\left(T\pi_{\Sigma P}\right) = \left\{ \lp x, \Lambda, \br; 0, \Lambda \eta, \lp\Ad_\Lambda \eta\rp \cdot \br\rp \mid  \eta \in \mso(3) \right\},
 \\ 
 \label{A_horizontal}
 \hor_A \,P &= \ker A = \left\{(x, \Lambda, \br; v_x, 0, \bu) \mid \bu \in T_{\br}\mR^3,\, v_x \in T_xX \right\}.
\end{align}

\subsection{Splitting $TP/SO(3)$}

In order to take advantage of the horizontal-vertical split of $TP$ we give the induced global splitting of the vector bundle $TP/SO(3)\rightarrow\Sigma$.  This is provided by the following vector bundle isomorphism, $\al_A:TP/SO(3) \longrightarrow T \Sigma \oplus_{ \Sigma} \ad P$, given by:
\begin{displaymath}
\al_A\lp\lsb v_p\rsb_{SO(3)}\rp = T\pi_{\Sigma P}(v_p) \oplus \lsb p, A(v_p)\rsb_{SO(3)}, \quad v_p \in T_pP
\,,
\end{displaymath}
where $\ad P := \lp P \times \mso(3)\rp/SO(3)$ is the {\bfi associated adjoint bundle} to $P$. The quotient is taken relative to the left diagonal action given by
\[
(p,\eta)\mapsto (hp,\operatorname{Ad}_h\eta),
\]
and the elements in the adjoint bundle are written $[p,\eta]_{SO(3)}$.
To check that $\al_A$ is well defined we verify that
\[
T\pi_{\Sigma P}(h v_p)=T\pi_{\Sigma P}(v_p),\\
\]
and
\[
\lsb hp, A(hv_p)\rsb_{SO(3)}=\lsb hp, \Ad_hA(v_p)\rsb_{SO(3)} = \lsb p, A(v_p)\rsb_{SO(3)}.
\]
To show that $\al_A$ is an isomorphism we give its inverse,
\begin{displaymath}
\al_A^{-1}\lp v_{(x,\brho)} \oplus \lsb p,\eta\rsb_{SO(3)} \rp = \lsb \operatorname{Hor}^A_p(v_{(x,\brho)}) + \eta_P(p)\rsb_{SO(3)},
\end{displaymath}
where $p=(x,\Lambda,\br)\in P$ is such that $\pi_{\Sigma P}(p)=(x,\brho)$ and $\operatorname{Hor}^A_p$ denotes the {\bfi horizontal lift} of $v_{(x,\brho)}=(x,\brho,v_x,\bu)\in T_{(x,\brho)}\Sigma$ to $T_pP$ with respect to $A$. It is given by
\[
\operatorname{Hor}^A_{(x,\Lambda,\br)}(x,\brho,v_x,\bu)= \lp x, \Lambda, \br, v_x, 0, \Lambda \bu\rp.
\]

\begin{remark}[The choice of connection]
{\rm As we have seen in \eqref{trivial_connection}, a natural choice of connection is the Maurer-Cartan form for the structure group, $d\Lambda\Lambda^{-1}$.}
\end{remark}

\subsection{Properties of $\ad P$}
\label{properties_adjoint_bundle}

We shall need various properties of the adjoint bundle to derive the Lagrange-Poincar\'e equations; we review them here.
\smallskip

We can give $\ad P$ a Lie algebra structure on each fiber.  The vector space structure is given by
\begin{displaymath}
\lsb p,\eta\rsb_{SO(3)} + a\lsb p,\nu\rsb_{SO(3)} = \lsb p,\eta+a\nu\rsb_{SO(3)}
\end{displaymath}
and the Lie bracket is given by
\begin{displaymath}
\lsb \lsb p, \eta\rsb_{SO(3)}, \lsb p, \nu\rsb_{SO(3)}\rsb = \lsb p, \lsb \eta, \nu \rsb \rsb_{SO(3)}.
\end{displaymath}

The principal connection $A $ induces an affine connection on the adjoint bundle $\ad P$. It is known that the covariant derivative of this affine connection is given by 
\begin{equation}
\label{general_covariant_derivative}
\frac{D^A}{D\tau} \lsb p(\tau), \eta(\tau) \rsb_{SO(3)} = 
 \lsb p(\tau), \dot{\eta}(\tau) - \lsb A\lp \dot{p}(\tau) \rp, \eta(\tau)\rsb  \rsb_{SO(3)} 
\end{equation}
(see, for example, \cite{CeMaRa2001}, Lemma 2.3.4). 
This formula allows us to define a covariant derivative of any section $\zeta: X \rightarrow \ad P $. Note that the adjoint bundle $\ad P $ has base $\Sigma$. We view it now as a bundle over $X $. It is important to note that the composite bundle $\ad P \rightarrow X $ is \textit{not} a vector bundle, in general. The covariant derivative of the section $\zeta$ is defined by using the formula for the covariant derivative of $\ad P \rightarrow \Sigma$ induced by the principal connection $A$ on $\pi_{\Sigma P} : P  \rightarrow \Sigma$. We define
\begin{equation}
\label{strange_covariant_derivative}
\nabla ^A _U \zeta(x) : = \left. \frac{D^A}{D\tau} \right|_{\tau=0} ( \zeta \circ c)( \tau), 
\end{equation}
where $c( \tau) $ is a smooth curve in $X $ such that $c (0) = x$ and $\dot{c}(0) = U \in T_x X$. Concretely, denoting $\zeta( x) = [p( x), \eta(x)]_{SO(3)}$, formula \eqref{general_covariant_derivative} gives 
\begin{equation}
\label{strange_formula}
\nabla ^A _U \zeta(x) = \left[p(x), \mathbf{d}\eta(x)(U) - [A(T_xp(U)), \eta(x)]
\right]_{SO(3)}.
\end{equation}
\smallskip

Let us note that the vector bundle $\ad P \rightarrow \Sigma$ is in our case trivial. Indeed, the map
\begin{equation}
\label{trivialization}
\left[ (x, \Lambda, \br), \eta \right]_{SO(3)} \mapsto \left( (x, \Lambda^{-1} \br ), \operatorname{Ad}_{ \Lambda^{-1}} \eta\right)
\end{equation}
is a vector bundle isomorphism from $\ad P$ to $\Sigma \times \mathfrak{so}(3)$. In this trivialization, using the connection \eqref{trivial_connection}, the formula for the covariant derivative \eqref{general_covariant_derivative} becomes
\[
\frac{D^A}{D\tau}( x ( \tau), \brho( \tau), \xi( \tau) ) = 
( x ( \tau), \brho( \tau), \dot{\xi}( \tau) ).
\]
Similarly, if $U \in T_xX$,  formula \eqref{strange_formula} becomes
\begin{equation}
\label{covariant_trivial}
\nabla ^A _U \zeta(x) = ( x, \brho (x), \mathbf{d}\xi(x) (U)), \quad \text{where} \quad \zeta(x) = (x, \brho (x), \xi(x)).
\end{equation}
Had the bundles been nontrivial, the formulas for the covariant derivatives would be more involved.

\subsection{Splitting $J^1P/SO(3)$}

Having introduced the connection that splits  $TP/SO(3)$ we now wish to use it to split the reduced covariant state space $J^1P/SO(3)$.  This is easily achieved by regarding the jets as linear maps and composing with $\al_A$.  Therefore we split $J^1P/SO(3)$ by splitting the image of the jets in $TP/SO(3)$. 

Consider a section $\sigma$ of the fiber bundle $\pi_{XP}:P\rightarrow X$ and its first jet extension $j^1\sigma(x)=T_x\sigma\in (J^1P)_{\sigma(x)}$.  We compose $\lsb T\sigma\rsb_{SO(3)}: TX \to TP/SO(3)$ with the vector bundle isomorphism $\al_A:TP/SO(3)\rightarrow T\Sigma\oplus_{ \Sigma}\operatorname{ad}P$ over $\Sigma$ and obtain the following equality in the fiber over  $\pi_{\Sigma P}(\sigma(x))$:
\begin{eqnarray*}
\al_A\circ \lsb T_x\sigma \rsb_{SO(3)} &=& \lp T_{\sigma(x)}\pi_{\Sigma P}\circ T_x\sigma\rp \oplus \lsb \sigma(x), A\circ T_x\sigma\rsb_{SO(3)}
\\
&=& T _x\lp\pi_{\Sigma P}\circ \sigma \rp \oplus \lsb \sigma(x), A\circ T_x\sigma\rsb_{SO(3)}
\,.
\end{eqnarray*}
Using $\pi_{X\Sigma}\circ \pi_{\Sigma P} = \pi_{XP}$, we have
\[
\pi_{X\Sigma}\circ (\pi_{\Sigma P}\circ\sigma)=\pi_{XP}\circ\sigma=\id_X.
..
\]
This shows that $\pi_{\Sigma P}\circ\,\sigma$ is a section of the fiber bundle $\pi_{X\Sigma}:\Sigma\rightarrow X$. If we denote
\[
\sigma_1=\pi_{\Sigma P}\circ\,\sigma,\quad \con A\lp v_p\rp: = \lsb p, A\lp v_p\rp \rsb_{SO(3)},\quad\text{and}\quad \sigma_2(x):=\con A\circ T_x\sigma,
\]
for all $v_p \in T_pP$, then the {\bfi reduced jet} $[ j ^1 \sigma (x)]_{SO(3)} \in (J ^1 P)/SO(3)$ is expressed as
\[ 
\alpha_A \circ [ j ^1 \sigma (x)]_{SO(3)} = 
\al_A\circ \lsb T_x\sigma\rsb_{SO(3)} = T_x \sigma_1 \oplus \con A\circ T_x\sigma=T_x\brho \oplus \sigma_2(x),
\]
since $\sigma_1=\brho$. Note that this element lies in the fiber
\[
(J^1\Sigma)_{\sigma_1(x)}\,\times \,L\left(T_xX,\left(\operatorname{ad}P\right)_{\sigma_1(x)}\right)
\]
over $\sigma_1(x)=\pi_{X\Sigma}(\sigma(x))\in \Sigma$. In particular, there is a fiber bundle isomorphism
\[
J^1P/SO(3)\cong J^1\Sigma\times_\Sigma L(TX,\operatorname{ad}P)
\]
over $\Sigma$. Using the equality $\sigma(x)=(x,\Lambda(x),\br(x))$, the explicit description of the quantities appearing in the reduced jet are:
\begin{eqnarray*}
T_x\brho &=&\left(x, \Lambda ^{-1} \br (x), T_x ( \Lambda ^{-1} \br)
\right), \\
&=& \lp x,\brho(x) ,\brho_s(x)ds + \brho_t(x)dt \rp \in (J^1 \Sigma)_{ \sigma(x)},
\\
\\
\sigma_2 &=& \lsb (x,\Lambda(x),\br(x)), T_x\Lambda\Lambda^{-1}(x) \rsb_{SO(3)},
\\
&\cong & \left((x,\Lambda^{-1}\br(x)), \Lambda^{-1}T_x\Lambda \right),\\
&=& \left( x, \brho(x), \bOm(x) ds + \bom(x)dt \right)
\,,
\end{eqnarray*}
by the relations \eqref{rhodef}, where $\cong$ denotes here the vector bundle isomorphism \eqref{trivialization}. Thus, the reduced jet $[j^1 \sigma(x)]_{SO(3)}$ associated to $\sigma(x)=(x,\Lambda(x),\br(x))$ is represented in the trivialization  \eqref{trivialization} by
\begin{align}\label{reduced_jet_extension}
T_x\brho \oplus \sigma_2(x) 
& = \left(x, \Lambda ^{-1} \br (x), T_x ( \Lambda ^{-1} \br), \Lambda^{-1}T_x\Lambda \right)
\nonumber\\
&= \lp x,\brho(x) ,\brho_s(x)ds + \brho_t(x)dt, \bOm(x) ds + \bom(x)dt \rp .
\end{align}
Therefore we have recovered the new coordinates given in \eqref{coordchange} and the reduced Lagrangian reads
\[
\bar l(T\brho,\sigma_2)=\bar l (\brho,\brho_s,\brho_t,\bom,\bOm).
\]

We summarize the spaces involved in the previous discussion in the following diagram:
\[
\begin{CD}
\ad SE(3) && \ad P\\
@VVV  @VVV \\
TSE(3)/SO(3) @>>> TP/SO(3) @>T\pi_{XP}/SO(3)>> TX \\
 @V T\pi_{SE(3)}/SO(3) VV  @VT\pi_{\Sigma P}/SO(3) VV @VV\id V \\
T\mR^3 @>>> T\Sigma @>> T\pi_{X\Sigma} > TX  .
\end{CD}
\]

We recall that the projections $\pi_{XP}$ and $\pi_{\Sigma P}$ are $SO(3)$-invariant, therefore they naturally induce projections $\pi_{\Sigma P}/SO(3):P/SO(3) \to \Sigma$ and $\pi_{XP}/SO(3):P/SO(3) \to X$, respectively. The relationship between the variables in the affine Euler-Poincar\'e and covariant Lagrange-Poincar\'e equations can be illustrated in the following diagram:

{\footnotesize
\[
\!\!\!\!
\begin{CD}
 \lp \brho, \bOm ds + \bom dt\rp_{SO(3)} \in L(TX,\ad SE(3)) && L(TX,\ad P)\\
 @VVV  @VVV \\
\lp \brho, \bGam ds + \bgam dt, \bOm ds + \bom dt\rp \in L(TX, TSE(3)/SO(3)) @>>> J^1P/SO(3) @>T\pi_{XP}/SO(3)>> L(TX, TX)  \\
 @V T\pi_{SE(3)}/SO(3) VV  @V T\pi_{\Sigma P}/SO(3) VV @VV\id V \\
 \lp \brho, \brho_s ds + \brho_t dt\rp \in L(TX,T\mR^3) @>>> J^1\Sigma @>> T\pi_{X\Sigma} > L(TX, TX) .
\end{CD}
\]
}

The affine Euler-Poincar\'e variables $[j ^1\sigma]_{SO(3)} \cong  (\brho, \bOm ds+ \bom dt, \bGam ds + \bgam dt) $ appear in the middle horizontal sequence whereas the covariant Lagrange-Poincar\'e variables  $( \brho, \brho_s ds + \brho_t dt, \bOm d s + \bom d t)$ appear in the top and bottom rows.

\subsection{Reduced Variations}

Let $\sigma: X \rightarrow P $ be a section of the fiber bundle $\pi_{XP}:P \rightarrow X $. If $\sigma_ \varepsilon: X \rightarrow P $ is a curve of sections with $\sigma_0 = \sigma$, that is, $\sigma_ \varepsilon(x) = ( x , \Lambda_ \varepsilon( x), \br_ \varepsilon( x)) $,  $\Lambda_0 = \Lambda$, and $\br_0 = \br$, define the variation
\begin{displaymath}
\delta \sigma(x)= \left.\dd{}{\epsilon}\right|_{\epsilon = 0}(x,\Lambda_\epsilon(x), \br_\epsilon(x)) \in T_{ \sigma(x)} P.
\end{displaymath}
Splitting $\delta\sigma(x)$ into its vertical and horizontal parts relative to the connection $A $ in the principal $SO(3)$-bundle $\pi_{ \Sigma P}: P \rightarrow \Sigma$ (see \eqref{A_vertical}, \eqref{A_horizontal}), we obtain
\begin{align*}
\delta \sigma(x)  &= \lp x, \Lambda, \br,0,\delta\Lambda, \delta \br \rp \\
&= \lp x,\Lambda, \br,0,\delta\Lambda, \de \Lambda \Lambda^{-1} \br \rp +\lp x, \Lambda, \br,0,0, \delta \br - \de \Lambda \Lambda^{-1} \br \rp \in T_{ \sigma(x)} P.
\end{align*}

\smallskip

To compute the vertical variation of $[j^1 \sigma]_{SO(3)}$, we consider curves $\sigma_ \varepsilon$ that perturb  $\sigma_0 = \sigma:X \rightarrow P $ along the group orbits, that is,  
\[
\sigma_ \varepsilon (x) : =  \exp ( \varepsilon\xi(x)) \cdot \sigma(x)
=  \left(x,  \exp ( \varepsilon\xi(x)) \Lambda(x),  \exp ( \varepsilon\xi(x)) \br(x)\right)
\]
where $\xi: X\rightarrow \mathfrak{so}(3)$. By \eqref{reduced_jet_extension}, in the trivialization \eqref{trivialization}, 
\[
\left[ j^1 \sigma_ \varepsilon(x) \right]_{SO(3)} \cong 
 \left(x, \Lambda ^{-1} \br (x); T_x ( \Lambda ^{-1} \br), \left(\exp ( \varepsilon\xi(x)) \Lambda(x) \right)^{-1}T_x \left(\exp ( \varepsilon\xi )\Lambda\right)\right).
\]
Taking the $\varepsilon$-derivative of the right hand side we get 
\begin{align}\label{ver_var}
\delta^v \left[ j^1 \sigma(x) \right]_{SO(3)} :&=
\left.\frac{d}{d \varepsilon}\right|_{\varepsilon=0}\left[ j^1 \sigma_ \varepsilon(x) \right]_{SO(3)} \nonumber\\
& \cong
\left(x, \Lambda^{-1}\br (x); 0, \left( \bsigma' +\bOm\times\bsigma\right) ds + \lp \dot \bsigma + \bom\times\bsigma\rp dt \right),
\end{align}
where $\Sigma(x) : = \operatorname{Ad}_{ \Lambda(x)^{-1}} \xi(x)$.

To compute the horizontal variation of $[j^1 \sigma]_{SO(3)}$, we consider curves $\sigma_ \varepsilon$ that perturb  $\sigma_0 = \sigma:X \rightarrow P $ such that $\delta \sigma $ is horizontal. In view of \eqref{A_horizontal}, 
a curve giving a horizontal $\delta\sigma$ is $\sigma_ \varepsilon(x) : = ( x , \Lambda(x),  \br_ \varepsilon (x))$. Therefore, for such a curve $\sigma_ \varepsilon$ we get
\[
\left[ j^1 \sigma_ \varepsilon(x) \right]_{SO(3)} \cong 
\left(x, \Lambda^{-1} \br _ \varepsilon(x); T_x ( \Lambda^{-1}\br _ \varepsilon) , \Lambda(x)^{-1}T _x \Lambda
\right)
\]
and hence 
\begin{align}\label{hor_var}
\delta^h \left[ j^1 \sigma(x) \right]_{SO(3)} :&=
\left.\frac{d}{d \varepsilon}\right|_{\varepsilon=0}\left[ j^1 \sigma_ \varepsilon(x) \right]_{SO(3)}
 \cong \left(x, \Lambda^{-1}\br (x); T_x ( \Lambda^{-1} \delta\br ), 0 \right) \\
 & = \left( x, \brho (x); T_x( \delta\brho ), 0  \right)
  =\left(x, \brho(x);  \lp\delta\brho\rp_t dt + \lp\delta\brho\rp_s ds, 0 \right)
\end{align}

\begin{rema} {\rm The free variations $\de \brho$ and $\bsigma$ are now recognized as being horizontal and vertical variations.  This is the reason for the decoupled form of the resulting equations.  If we had defects in the strand and therefore our connection had non-zero curvature, then the equations would not decouple completely.}
\end{rema}

\subsection{Variational Principle}
 
Having derived the reduced horizontal and vertical variations, we may now derive the horizontal and vertical Lagrange-Poincar\'e equations.  For the vertical variations, using \eqref{ver_var} and
\[
\dede{\bar l}{\sigma_2}=\lsb \brho, \dede{\bar l}{\bom}\prt_t + \dede{\bar l}{\bOm}\prt_s\rsb_{SO(3)},
\]
we obtain
\begin{align*}
\delta^v S &= \left.\frac{d}{d\varepsilon}\right|_{\varepsilon=0}\int_X\bar l\left(\left[j^1\sigma_\varepsilon\right]_{SO(3)} \right)dx\\
&= \int_X \scp{\lsb \brho, \dede{\bar l}{\bom}\prt_t + \dede{\bar l}{\bOm}\prt_s\rsb_{SO(3)}}{\lsb \brho, \lp \dot \bsigma + \bom\times\bsigma\rp dt  + \lp \bsigma' + \bOm\times\bsigma\rp ds\rsb_{SO(3)}}dx
\\
&= -\int_X \scp{\lsb \brho, \lp \prt_t + \bom\times\rp\dede{\bar l}{\bom} + \lp \prt_s + \bOm\times\rp\dede{\bar l}{\bOm}\rsb_{SO(3)}}{\lsb \brho, \bsigma\rsb_{SO(3)}}dx= 0
\,.
\end{align*}

Therefore, the {\bfi vertical covariant Lagrange-Poincar\'e equation} is
\begin{displaymath}
\lp \prt_t + \bom\times\rp\dede{\bar l}{\bom} + \lp \prt_s + \bOm\times\rp\dede{\bar l}{\bOm} = 0
\,.
\end{displaymath}

Similarly, we derive the variational principal for horizontal variations and obtain, using \eqref{hor_var},
\begin{align*}
\delta^h S &=\left.\frac{d}{d\varepsilon}\right|_{\varepsilon=0}\int_X\bar l\left(\left[j^1\sigma_\varepsilon\right]_{SO(3)} \right)dx=\int_X \scp{\dede{\bar l}{T\brho}}{\de T\brho}dx\\
&= \int_X \scp{\dede{\bar l}{\brho}}{\de\brho} + \scp{\dede{\bar l}{\brho_t}\prt_t + \dede{\bar l}{\brho_s}\prt_s}{\de \brho_t dt + \de \brho_s ds}dx
\\
&= \int_X \scp{\dede{\bar l}{\brho} - \prt_t \dede{\bar l}{\brho_t} - \prt_s\dede{\bar l}{\brho_s}}{\delta\brho}dx= 0
\,.
\end{align*}

Therefore, the {\bfi horizontal covariant Lagrange-Poincar\'e equation} is
\begin{displaymath}
\dede{\bar l}{\brho} - \prt_t \dede{\bar l}{\brho_t} - \prt_s\dede{\bar l}{\brho_s} = 0
\,.
\end{displaymath}

Upon putting these together, we find that the covariant Lagrange-Poincar\'e equations are
\begin{eqnarray}
\lp \prt_t + \bom\times\rp\dede{\bar l}{\bom} + \lp \prt_s + \bOm\times\rp\dede{\bar l}{\bOm} &=& 0\,,  \label{verLaPo}
\\
\dede{\bar l}{\brho} - \prt_t \dede{\bar l}{\brho_t} - \prt_s\dede{\bar l}{\brho_s} &=& 0
\,. \label{horLaPo}
\end{eqnarray}

\subsection{A circulation theorem}\label{circulation_thm}

The Kelvin-Noether Theorem tells us about the solutions to the Euler-Poincar\'e equations in continuum mechanics.  There is an analogue in the covariant picture that is described in this section. Denoting by 
$\operatorname{div}_x $ the divergence relative to the variable $x =(s, t)\in [0,L] \times  \mathbb{R}$, we have
\begin{align*}
&\operatorname{div}_x\left(\operatorname{Ad}^*_{\Lambda^{-1}}\frac{\delta \bar l}{\delta\bOm}\partial_s+\operatorname{Ad}^*_{\Lambda^{-1}}\frac{\delta \bar l}{\delta\bom}\partial_t\right)
=\partial_s\left(\operatorname{Ad}^*_{\Lambda^{-1}}\frac{\delta \bar l}{\delta\bOm}\right)+\partial_t\left(\operatorname{Ad}^*_{\Lambda^{-1}}\frac{\delta \bar l}{\delta\bom}\right)\\
&\qquad \qquad =\operatorname{Ad}^*_{\Lambda^{-1}}\left(\partial_s\frac{\delta \bar l}{\delta\bOm}+\bOm\times\frac{\delta \bar l}{\delta\bOm}+\partial_t\frac{\delta \bar l}{\delta\bom}+\bom\times\frac{\delta \bar l}{\delta\bom}\right)=0
\end{align*}
by \eqref{verLaPo}. Using the divergence theorem, we find
\begin{align*}
0&=\int_S\operatorname{div}_x\left(\operatorname{Ad}^*_{\Lambda^{-1}}\frac{\delta \bar l}{\delta\bOm}\partial_s+\operatorname{Ad}^*_{\Lambda^{-1}}\frac{\delta \bar l}{\delta\bom}\partial_t\right) ds dt \\
&=\int_{\partial S}\left(\operatorname{Ad}^*_{\Lambda^{-1}}\frac{\delta \bar l}{\delta\bOm}\partial_s+\operatorname{Ad}^*_{\Lambda^{-1}}\frac{\delta \bar l}{\delta\bom}\partial_t\right)\!\cdot\!\mathbf{n}\, d\ell\\
&=\int_{\partial S}\operatorname{Ad}^*_{\Lambda^{-1}}\left(\frac{\delta \bar l}{\delta\bOm}dt-\frac{\delta \bar l}{\delta\bom}ds\right),
\end{align*}
where $\mathbf{n}$ is the outward pointing unit normal to the boundary $\partial S$ and we used the identity
\begin{equation}
\label{div_identity}
\left(\frac{\delta \bar l}{\delta\bOm}\partial_s+\frac{\delta \bar l}{\delta\bom}\partial_t\right)\!\cdot\!\mathbf{n}\, d\ell = 
\frac{\delta \bar l}{\delta\bOm}dt-\frac{\delta \bar l}{\delta\bom}ds.
\end{equation}
Thus, we obtain the \textit{circulation theorem}
\begin{equation}
\label{circulation_1}
\int_{\partial S}\operatorname{Ad}^*_{\Lambda^{-1}}\left(\frac{\delta \bar l}{\delta\bOm}dt-\frac{\delta \bar l}{\delta\bom}ds\right)=0.
\end{equation}

\subsection{Generalizations of the molecular strand}\label{sec.strand.gen}

Essentially the problem of the molecular strand is a complex filament.  We have a filament in $\mR^3$ and also specify that the filament has microstructure that we describe by attaching a group element to each point of the filament.  This setup can be generalized in various interesting ways.  First we can consider the multidimensional problem, that is, complex sheets and related structures.  A molecular sheet is an object where molecules are bound in a two dimensional surface.  In this case we can still consider RCCs of charges represented by $SO(3)$, but this time the $SO(3)$ quantities are attached to points on an embedding of $\mR^2$ in $\mR^3$.  We could consider another generalization.  Instead of having a filament with RCCs described by $SO(3)$, some other group could describe the microstructure.  For example, the group $SO(2)$ was used in \cite{Mezic2006}.  We could even extend the group to take account of quantum phenomena.  We might want to make both of the above generalizations and consider, for example, the spin sheet.

\smallskip

Motivated by these considerations, we quickly indicate here how to generalize the covariant Lagrange-Poincar\'e approach to the setting of Subsection \ref{n_dimensional_generalization}, that is, the case of $n$-dimensional strand with an arbitrary Lie group structure $\mathcal{O}$.

Consider the $(n+1)$-dimensional spacetime $X:=\mathcal{D}\times\mathbb{R}$ and the trivial fiber bundle
\[
\pi_{XP}:P:=X\times S\rightarrow X
\,,
\]
where $S=\mathcal{O}\,\circledS\, E$. A section $\sigma$ of $P$ reads
\[
\sigma(x)=(x,\Lambda(x),r(x)),\quad x=(s,t)\in X
\,,
\]
and its first jet extension is
\[
j^1\sigma(x)=(T_x\Lambda,T_xr)
=(\mathbf{d}\Lambda(x)+\dot\Lambda(x)dt,\mathbf{d}r(x)+\dot r(x)dt)
\,,
\]
where $\mathbf{d}$ is the partial derivative with respect to space (that is, the derivative on $\mathcal{D}$), and the dot is the partial derivative with respect to time.

There is a natural $\mathcal{O}$-principal bundle structure on $S$ given by
\[
\pi_{ES}:S\rightarrow E,\quad \pi_{ES}(\Lambda,r)=\Lambda^{-1}r=\rho.
\]
This principal bundle structure on the fiber $S$, induces a principal $\mathcal{O}$-bundle structure on $P$ given by
\[
\pi_{\Sigma P}:P\rightarrow X\times E,\quad \pi_{\Sigma P}(x,\Lambda,r)=(x,\Lambda^{-1}r).
\]
There is a natural connection $A$ on $\pi_{\Sigma P}:P\rightarrow \Sigma:=X\times E$ given by
\[
A(v_x,v_\Lambda,(r,u))=v_\Lambda\Lambda^{-1},
\]
which allows us to identify the reduced jet bundle $J^1P/\mathcal{O}$ with the fiber bundle $J^1\Sigma\times_\Sigma L(TX,\operatorname{ad}P)$. Using the same notations as before, we have
\[
\alpha_A\circ [j^1\sigma(x)]_{\mathcal{O}}=T_x\rho\oplus \con A\circ T_x\sigma\cong (x,\rho(x),\mathbf{d}\rho(x)+\dot\rho(x),\Omega(x)+\omega(x)dt),
\]
by \eqref{rhodef_n_dimensional}. The vertical and horizontal variations being given by
\begin{align*}
\delta^v\left[j^1\sigma(x)\right]_{\mathcal{O}} &=(x,\Lambda^{-1}r(x);0,\mathbf{d}\Sigma+[\Omega,\Sigma]+(\dot\Sigma +[\omega,\Sigma])dt),\\
\delta^h\left[j^1\sigma(x)\right]_{\mathcal{O}} &=(x,\rho(x);\mathbf{d}(\delta\rho)+(\delta\rho)_tdt,0),
\end{align*}
we find that the vertical and horizontal Lagrange-Poincar\'e equations are
\begin{eqnarray*}
\left( \prt_t -\operatorname{ad}^*_\omega\right)\dede{\bar l}{\om} + \left(\operatorname{div} - \operatorname{ad}^*_\Om\right)\dede{\bar l}{\Om} &=& 0
\,,\nonumber\\
\dede{\bar l}{\brho} - \prt_t \dede{\bar l}{\brho_t} - \operatorname{div}\dede{\bar l}{\brho_s} &=& 0
\,\nonumber.
\end{eqnarray*}
Of course, as expected, these equations coincide with equations \eqref{equ_new_variables} obtained from the affine Euler-Poincar\'e equations \eqref{generalized_molec_strand} by the change of variables \eqref{change_variables}. In Section~\ref{sec:ApplicationLP} we will obtain the same equations by applying directly the theory in \cite{CaRa03}, using the group structure of $S$. Note that the approach we have used here does not use the group structure of $S$, and is expected to be applicable to more general situations such as the molecular strand on the sphere. This is explained in the next subsection. We have
\begin{align*}
&\operatorname{div}_x\left(\operatorname{Ad}^*_{\Lambda^{-1}}\frac{\delta \bar l}{\delta\Om}+\operatorname{Ad}^*_{\Lambda^{-1}}\frac{\delta \bar l}{\delta\om}\partial_t\right) =\operatorname{div}\left(\operatorname{Ad}^*_{\Lambda^{-1}}\frac{\delta \bar l}{\delta\Om}\right)+\partial_t\left(\operatorname{Ad}^*_{\Lambda^{-1}}\frac{\delta \bar l}{\delta\om}\right)\\
& \qquad \qquad 
=\operatorname{Ad}^*_{\Lambda^{-1}}\left(\operatorname{div}\frac{\delta \bar l}{\delta\Om}-\operatorname{ad}^*_\Om\frac{\delta \bar l}{\delta\bOm}+\partial_t\frac{\delta \bar l}{\delta\om}-\operatorname{ad}^*_\om\frac{\delta \bar l}{\delta\om}\right)=0.
\end{align*}
Using the divergence theorem, we obtain the zero flux theorem
\[
\int_{\partial V}\left(\operatorname{Ad}^*_{\Lambda^{-1}}\frac{\delta \bar l}{\delta\Om}+\operatorname{Ad}^*_{\Lambda^{-1}}\frac{\delta \bar l}{\delta\om}\partial_t\right) \!\cdot\!\mathbf{n}\, d \sigma=0,
\]
where $\mathbf{n}$ is the outward pointing unit normal to the boundary $\partial V$ of a given domain $V  \subset \mathcal{D} \times \mathbb{R}$ and $d \sigma$ is the induced boundary volume element of $\partial V$.

\begin{remark}
\rm
We remark at this point that we have purposefully not referred to $SE(3)$ in this section. Our entire method of reduction and derivation of \eqref{verlp}, \eqref{horlp} makes no use of the group structure of $SE(3)$.  This is the first application we know that uses this method of reduction.  The closest theory that one finds previously is in \cite{CaRa03} where covariant Lagrange-Poincar\'e equations are derived in the case in which a principal fiber bundle is reduced by a subgroup of the structure group.  Naturally that theory is related to this approach in which we imposed less restrictive assumptions on $P$ and renounced the $SE(3)$ group structure.  Of course, since in this case $P$ is also a principal $SE(3)$-bundle, the charged strand is in the intersection of these two reduction theories.  A treatment of the charged strand by a direct application of \cite{CaRa03} is presented in the next section. The approach used here may be regarded as a covariant formulation of the Lagrange-Poincar\'e equations derived in \cite{CeMaRa2001}.  It also answers the call for the approach used in \cite{CaRa03} to be applied to general field theories.  One finds in general that under certain equivariance conditions on the trivializing maps of the bundle $\pi_{XP}$, a variational approach to field theories may be formulated that reduces to the case considered in \cite{CaRa03} when the fibers of $\pi_{XP}$ are Lie groups.  This new reduction procedure also applies to a wider category of fiber bundles  which includes the case where $\pi_{XP}$ is trivial, but $Q$ is any $G$-bundle, as often occurs in applications.  This generalization can be seen in the present application by varying our assumptions.  We note that the method above is capable of dealing with the problem of the molecular strand on a substrate or on any submanifold of $\mR^3$, whereas the method developed in \cite{CaRa03} does not treat that case.  Of course the application of the above theory to such a system is subject to symmetry breaking. This is an interesting direction for future research, but it is beyond the scope of our present considerations.  
\end{remark}

\section{The subgroup covariant Lagrange-Poincar\'e approach}
\label{sec:ApplicationLP}
In this section we allow ourselves the group structure of $SE(3)$ and apply the results formulated in \cite{CaRa03} to the charged strand. More precisely, we see the principal bundle
\[
\pi_{SE(3)}:SE(3)\rightarrow \mathbb{R}^3, \quad\pi_{SE(3)}(\Lambda,\br)=\Lambda^{-1}\br=\brho,
\]
as being associated to the left subgroup action of $SO(3)\cong SO(3)\times\{0\}$ on $SE(3)$. Using the composition law in $SE(3)$
\[
(\Lambda_1,\br_1)(\Lambda_2,\br_2)=(\Lambda_1\Lambda_2,\br_1+\Lambda_1\br_2),
\]
we obtain that the subgroup action is given by
\[
SO(3)\times SE(3)\rightarrow SE(3),\quad \Lambda_1(\Lambda_2,\br_2)=(\Lambda_1\Lambda_2,\Lambda_1\br_2).
\]
We thus have recovered the action \eqref{SO(3)_action}, and the projection $\pi_{SE(3)}$ identifies an equivalence class $[\Lambda,\br]$ in $SE(3)/SO(3)$ with the vector $\Lambda^{-1}\br\in\mathbb{R}^3$.

\medskip

Therefore, we can obtain the equation of the molecular strand by reducing the principal $SE(3)$-bundle $P$ by the subgroup $SO(3)$. Such a theory is developed in \cite{CaRa03}, and is applied below directly to the $n$-dimensional generalization of the molecular strand. The difference with the approach described in the previous Section lies in the fact that Section~\ref{sec:Lagrange-Poincare} does not use the group structure of $SE(3)$ and the fact that the principal bundle $\pi_{SE(3)}:SE(3)\rightarrow \mathbb{R}^3$ is associated to a subgroup action.

\medskip

Consider the manifold $X=\mathcal{D}\times\mathbb{R}$ and the trivial principal $S$-bundle, $P=  X\times S\rightarrow X$, where $S$ is a semidirect product group $\mathcal{O}\,\circledS\, E$.  Since $P$ is trivial, the first jet bundle is given by $J^1P_{(x,g)}=L(T_xX,T_gS)$. A section $\sigma$ of $P$ reads
\[
\sigma(x)=(x,\Lambda(x),r(x))
\,,\quad
x=(s,t)
\,,
\]
and its first jet extension is
\begin{eqnarray*}
j^1\sigma(x)&=& \lp x, \id_{TxX}, \Lambda(x),r(x), T_x\Lambda,T_xr\rp\\
& =&\lp x, \id_{TxX}, \Lambda(x),r(x), \mathbf{d}\Lambda+\dot\Lambda dt,\mathbf{d}r+\dot rdt\rp
,
\end{eqnarray*}

where $\mathbf{d}$ denotes the derivative with respect to space and the dot denotes the time derivative.

We can see $P$ as a $\mathcal{O}$-principal bundle over $\Sigma:=P/\mathcal{O}=X\times E$, relative to the projection $\iota: (x,\Lambda,r)\mapsto (x,\Lambda^{-1}r)$. Suppose that we have a $\mathcal{O}$-invariant Lagrangian density $\mathcal{L}$ defined on $J^1P$. This Lagrangian induces a reduced Lagrangian density $\bar{l}:J^1P/\mathcal{O}\rightarrow\mathbb{R}$. On the principal bundle $P\rightarrow P/\mathcal{O}$ we consider the flat principal connection
\[
A(v_x,v_\Lambda,v_r)=v_\Lambda\Lambda^{-1}.
\]
where $v_x \in T_xX$, $v_\Lambda \in T_\Lambda \CO$ and $v_r \in T_rE$. Using this connection, we have the fiber bundle isomorphism $J^1P/\mathcal{O}\cong J^1(P/\mathcal{O})\times_{\Sigma}L(TX,\operatorname{ad}P)$ over $\Sigma$. Note that in our particular case, the vector bundle $\operatorname{ad}P\rightarrow\Sigma$ is trivial, see \eqref{trivialization}, and can be identified with $\Sigma\times\mathfrak{o}$. Moreover, the connection $A$ is identified with the trivial connection, see \eqref{covariant_trivial}.

We recall now from \cite{CaRa03} the covariant Lagrange-Poincar\'e reduction, adapted here to the case of a semidirect product $S=\mathcal{O}\,\circledS\,E$, and to the fact that $P$, as a $S$-principal bundle, is trivial.

Given a section $\sigma=(\Lambda,r)$ of $P\rightarrow X$, we consider section $\sigma_1$ of $P/\mathcal{O}\rightarrow X$ defined by $\sigma_1(x):=\iota\circ\sigma(x)=\Lambda(x)^{-1}r(x)=\rho(x)$, and the section $\sigma_2$ of $L(TX,\mathfrak{o})\rightarrow X$ defined by $\sigma_2(x)=\Lambda^{-1}T_x\Lambda=\Omega+\omega dt$. The following are equivalent.
\begin{itemize}
\item $\sigma$ is a critical point for the variational principle
\[
\delta \int_X\mathcal{L}(j^1\sigma)=0
\,;
\]
\item $\sigma$ satisfies the Euler-Lagrange equations for $\mathcal{L}$;
\item the variational principle
\[
\delta \int_X\bar{l}(j^1\sigma_1,\sigma_2)=0
\]
holds for arbitrary variations $\delta\sigma_1$ and variations of the form
\[
\delta\sigma_2=\mathbf{d}^A\eta+[\sigma_2,\eta]
\,,
\]
where $\eta$ is an arbitrary section of $L(TX,\mathfrak{o})\rightarrow X$;
\item the sections $\sigma_1,\sigma_2$ satisfy the covariant Lagrange-Poincar\'e equations
\begin{equation}\label{cov_L-P}
\left\{\begin{array}{l}
\displaystyle\vspace{0.2cm}\frac{\delta\bar{l}}{\delta\sigma_1}-\operatorname{div}_x\left(\frac{\delta\bar{l}}{\delta (T\sigma_1)}\right)=0
\,,\\
\displaystyle\vspace{0.2cm}\operatorname{div}_x^A\frac{\delta\bar{ l}}{\delta\sigma_2}=\operatorname{ad}^*_{\sigma_2}\frac{\delta\bar{ l}}{\delta\sigma_2}
\,,
\end{array}\right.
\end{equation}
where $\operatorname{div}^A_x$ denotes the covariant divergence associated to $A$ and acting on $\mathfrak{X}(X,\mathfrak{o}^*)$. Note that here $\operatorname{div}_x^A=\operatorname{div}_x$.
\end{itemize}
Using the decomposition $X=\mathcal{D}\times\mathbb{R}$, we can write
\[
\mathcal{L}(j^1\sigma)=\mathcal{L}(\dot\Lambda,\mathbf{d}\Lambda,\dot r,\mathbf{d}r)
\]
and
\[
\bar{l}(j^1\sigma_1,\sigma_2)=\bar{l}(\rho,\Omega,\omega)
\,.
\]
Hence, we obtain the equality
\[
\frac{\delta\bar{l}}{\delta\sigma_2}=\frac{\delta\bar{l}}{\delta\Omega}
+\frac{\delta\bar{l}}{\delta\omega}\prt_t
\,.
\]
Since $\operatorname{div}^A_x\pi_1^*=\frac{d}{dt}$ and $\operatorname{div}^A_x\pi_2^*=\operatorname{div}$, the second equation of \eqref{cov_L-P} reads
\[
\frac{d}{dt}\frac{\delta\bar{l}}{\delta\omega}
+\operatorname{div}\frac{\delta\bar{l}}{\delta\Omega}
=
\operatorname{ad}^*_{\omega}\frac{\delta\bar{ l}}{\delta\omega}
+
\operatorname{ad}^*_{\Omega_i}\frac{\delta\bar{ l}}{\delta\Omega_i}
\,.
\]
The first equation reads
\[
\frac{\delta\bar{ l}}{\delta\rho}-\frac{d}{dt}\frac{\delta\bar{ l}}{\delta \rho_t}-\operatorname{div}\frac{\delta\bar{ l}}{\delta \rho_s}=0
\,.
\]
We have thus obtained equations \eqref{equ_new_variables} by covariant Lagrange-Poincar\'e reduction. Of course, when $\mathcal{O}\,\circledS\,E=SO(3)\,\circledS\,\mathbb{R}^3=SE(3)$ and $\mathcal{D}=[0,L]$ we recover equations \eqref{verLaPo} and \eqref{horLaPo} for the molecular strand in the new variables.

\section{Formulation of nonlocal exact geometric rods in terms of quaternions} \label{quatform}

Quaternions allow for a simple, elegant and useful method of describing the local orientation of a curve. It is thus natural to seek a representation of our derivation in previous sections that expresses the strand equations in terms of quaternions. The quaternion representation is natural, for example, in formulating the equations of motion for elastic rods in terms of the corresponding Euler parameters. As far as we are aware, a treatment of continuum rod theory in terms of quaternions in the nonlocal sense presented here does not appear in the literature. We shall see how the nonlocal contribution \eqref{Zdef} appears as an imaginary part of a certain quaternion, thereby making the connection to other work. This is accomplished by mapping quaternions  (elements of $SU(2)$) that describe rotations into purely imaginary quaternions, or vectors, that are elements of $\mathfrak{su}(2)  \simeq \mathfrak{so}(3) \simeq \mathbb{R}^3$.  

\begin{remark}\rm
This section simplifies the formulas and avoids extra factors in the integrals by assuming the strand to be inextensible; so that $| \bGam(s)|$ is identically equal to unity and the parameter $s$ is the arc length.
See \eqref{EPsigma} and \eqref{EPpsi} for the case $| \bGam(s)| \neq 1$.
\end{remark}

 Let us associate a quaternion $\fq=(q_0,\bq)$ with every point on the curve $s$. That quaternion describes the local rotation of an orthogonal frame if the condition $\| \fq\|=q_0^2+ |\bq|^2=1$ is satisfied. Then, $q_0=\cos \alpha/2$,  with rotation angle $\alpha$  and $\bq=\sin ( \alpha/2) \widehat{\bf n}$, where $\widehat{\bf v}$ is a unit vector around which the axis is rotated. These are the Cayley-Klein parameters of the rotation
 \smallskip

\begin{remark}\rm

To  simplify the notation, we use bold symbols for purely imaginary quaternions, considering them as vectors. For example, if $\fq$ is unit quaternion,  then
\[
\boldsymbol{b}=\fq \boldsymbol{a} \fq^\ast
\]
means
\[
\lp  0, \boldsymbol{b} \rp  =\fq \lp  0, \boldsymbol{a} \rp  \fq^\ast
\, .
\]
\end{remark}
 
As before, we assume that the interaction potential depends on the distances between point charges that are attached to  each point $\br(s,t)$ by rigid rods of the length $\boldeta_i(s)$. The new position of the charges will be
$\fq   \boldeta_i   \fq^\ast  $, where $ $ denotes the quaternion multiplication and $\fq^\ast$ is the quaternionis conjugate of $\fq$.  The point charges are then positioned at the coordinates in real space $\br(s,t)+ \fq   \boldeta_i   \fq^\ast$, and the distance between point charges is then%
\footnote{All these variables depend on time $t$ as well as $s$, but the time variable $t$ is suppressed.}
\begin{equation}
d_{k,m}(s,s')=\left|
\br(s)-\br(s')+  \fq(s)    \boldeta_k(s)   \fq^\ast(s) -
\fq(s')    \boldeta_m(s')   \fq^\ast(s')
\right|
\, .
\label{distquaternion}
\end{equation}
This is simply  \eqref{Energy0} written now in its quaternionic form. Following \eqref{LiePoissondist}, we perform the
Lie-Poisson reduction as follows (remember that $\| \fq \|=1$ and $\fq \fq^\ast=\fe$, where $\fe=(1,0)$ is the unit quaternion):
\begin{align}
d_{k,m}(s,s')= &
 \left|
\br(s')-\br(s)+  \fq(s')    \boldeta_k(s')   \fq^\ast(s') -
\fq(s)    \boldeta_m(s)   \fq^\ast(s)
\right|
\nonumber
\\
=&
\left| \fq^\ast(s)   \left( 0,
\br(s')-\br(s)+ \fq(s')     \boldeta_k(s')   \fq^\ast(s') -
\fq(s)    \boldeta_m (s)   \fq^\ast(s)
\right)
  \fq(s)
\right|
\nonumber
\\
=&
\left|
\fz(s,s')   \brho(s')    \fz^\ast(s,s') -\brho(s)
+\fz (s,s')    \boldeta_k(s')   \fz^\ast(s,s')
-  \boldeta_m (s) 
\right|
\nonumber 
\\
=&
\left|
\bkappa(s,s')
+ \boldeta_k(s)    -
\fz (s,s')    \boldeta_m(s')   \fz^\ast(s,s')
\right|
 \, ,
\label{LiePoissondist2}
\end{align}
where
$\fz(s,s') = \fq^\ast (s) \fq(s')$ is the coupling between the frames  and the quantity 
\[
\brho(s)=\fq^\ast(s)
\br(s)    \fq(s)
\,,
\]
 is the distance vector connecting the points $\br(s)$ and $\br(s')$ transformed  according to the inverse rotation of the frame at the point $s$.
We have also defined
\begin{equation}
\bkappa(s,s') = \fz(s,s')   \brho(s')    \fz^\ast(s,s')-\brho(s) 
\, .
\end{equation}

The nonlocal part of the reduced Lagrangian depends on the variables $\brho$ and $\fz$. The local part which describes elastic deformation and inertia, can be reduced to functions of $\fm = \fq^\ast \fq'$ and $\fv= \fq^\ast  \dot{\fq}$ where the prime denotes the derivative with respect to $s$ and the dot is the derivative with respect to $t$. The quaternions $\fm$ and $\fv$ belong to the Lie algebra of the Lie group of all unit quaternions, and is isomorphic to the space of purely imaginary quaternions, or vectors, as can be seen by differentiating
$\fq^\ast \fq = \fe$. The commutator is then mapped into twice the vector product of the imaginary parts of the quaternions.

We can again split the reduced Lagrangian $l_T$ in the local and nonlocal
parts
\begin{equation}
l = \int l_{loc} \lp  \fv(s), \fm(s),  \bgam, \bGam,  \brho \rp  \mbox{d} s
+ \iint  U \lp  \bkappa(s,s'), \fz(s,s')  \rp  \mbox{d} s\mbox{d} s'
:= l_{loc}+l_{np}
\, .
\end{equation}
Here, $\bgam$, $\bGam$, $\brho$ are defined as in \eqref{rhodef}.
The equations of motions then follow from the minimization of the reduced action
\begin{equation}
\delta S=\delta  \int  l \lp  \fv(s), \fm(s), \brho(s),
 \bkappa(s,s'), \fz(s,s') \rp   \mbox{d} s \mbox{d} s' \mbox{d} t =0
\, .
\label{minaction0}
\end{equation}
Here $\fv$ and $\fs$ are elements of Lie algebra of  quaternions with fixed absolute value. They are purely imaginary quaternions, $\fv=(0,\bom/2)$ and $\fs=(0,\bOm/2) $, which are isomorphic to vectors in $\mathbb{R}^3$. The factor of $1/2$ is necessary for $\bom$ and $\bOm$ to be exactly the vector angular velocity and strain rate, respectively, in correspondence with the previous section.  Thus, we can re-write
\eqref{minaction0} using vector quantities instead of quaternions whenever possible:
\begin{equation}
\delta S= \delta \iint  l_{np} \lp  \bom, \bgam, \bOm,  \bGam, \brho \rp
\mbox{d} s \mbox{d} t +
\delta
\iiint
U \lp  \bkappa(s,s') , \fz(s,s') \rp  \mbox{d} s \mbox{d} s' \mbox{d} t =0
\, .
\label{minaction1}
\end{equation}
We obtain
\begin{align}
\delta S & = \iint
\left\langle \frac{\delta l_{loc}}{\delta \bOm} \,,\, \delta \bOm  \right\rangle
+
\left\langle \frac{\delta l_{loc}}{\delta \bom}\,  , \, \delta \bom  \right\rangle
+
\left\langle \frac{\delta l_{loc}}{\delta \brho}  \,,\,  \delta \brho \right\rangle
+
\left\langle \frac{\delta l_{loc}}{\delta \bgam} \,,\, \delta \bgam  \right\rangle
\nonumber
\\
&+
\left\langle \frac{\delta l_{loc}}{\delta \bGam} \, , \, \delta \bGam  \right\rangle
\mbox{d} s \, \mbox{d}  t
+
\iiint
\left\langle \frac{\delta l_{np} }{\delta \fz} \, , \, \delta \fz  \right\rangle
+
\left\langle \frac{\delta l_{np}}{\delta \bkappa} \, , \, \delta \bkappa \right\rangle
\, \mbox{d}  s \,  \mbox{d}  s' \, \mbox{d}   t =0
 \,.
\label{minaction}
\end{align}

If we now define $\fs = \fq^\ast \delta \fq=\fq^{-1} \delta \fq$ as the free variation in $\fq$, we obtain, similarly to previous section:
\begin{equation}
\delta \fv = \fv   \fs - \fs   \fv + \dot{\fs}
= \lsb  \fv \, , \, \fs \rsb  + \dot{\fs}
\,,
\label{sdot}
\end{equation}
for the time derivative and
\begin{equation}
\delta \fm = \fm   \fs - \fs   \fm + \fs '
= \lsb  \fm \, , \, \fs \rsb  + \fs'
\,,
\label{sprime}
\end{equation}
for the space derivative.
Note that since $\fq^* \fq =1$,
\[
\fq^* \dot \fq + \dot \fq^* \fq= \fq^* \dot \fq +\big(\fq^* \dot \fq  \big)^* = 2 \mbox{Re} \,  \fv =0
\]
and, analogously, $ \mbox{Re} \,  \fm=0$, which means that $\fv$ and $\fm$ are purely imaginary quaternions, or vectors.
 This gives the variation of the first two terms in \eqref{minaction}. We now remember that $\mbox{Re}\, \fv=0$ and $ \mbox{Re}\, \fs=0$ since they are elements of the corresponding Lie algebra, so $\fv =(0, \bom/2)$, $\fm=(0, \bOm/2)$ and  $\fs =(0, \bsigma/2)$. Then, \eqref{sdot} and \eqref{sprime} become vector equations:
\begin{equation}
\delta \fv = (0, \delta \bom) = \lp  0 \, , \,
\bom \times \bsigma + \dot{\bsigma}
\rp
\, ,
\label{sdotvec}
\end{equation}
and
 \begin{equation}
\delta \fm
= (0, \delta \bOm)=
\lp  0 \, , \,
\bOm \times \bsigma + \bsigma\,'\,
\rp
\, .
\label{sprimevec}
\end{equation}

Now, computations of the first three variations in \eqref{minaction} can be done analogously to those in the previous section, as they involve vector quantities.
The only exception is the computation of $\delta \brho$ in the nonlocal term as it must be computed in terms of quaternions. We
have
\begin{align}
\delta  \brho &=
\delta \lp  \fq^\ast(s)  \br(s)   \fq  \rp
\nonumber
\\
&=
-\fq^\ast \delta \fq \brho(s) +
 \brho(s)   \fq^\ast \delta \fq+
\fq^\ast(s) \delta \br(s) \fq(s)
=
2  \brho \times \bsigma + \bpsi
\, ,
\label{deltarho2}
\end{align}
where we have defined the free variation
\begin{equation}
 \bpsi(s)  :=
\fq^\ast(s)  \delta \br(s)    \fq(s)
\, .
\label{psidef2}
\end{equation}
and used the fact that for purely imaginary $\fs=(0, \bsigma)$, $\fs \brho-\brho \fs = \brho \times
\bsigma$.
Then, the variation we need to compute the variation of the nonlocal part of $\bkappa(s,s') $ as follows. It is easier to use the alternative expression for $\bkappa$ as
\[
\bkappa(s,s') = \brho(s,s') - \fq^\ast (s) \br(s') \fq(s) \, .
\]
Then,
\begin{align}
\delta \bkappa =&
\delta \brho -   \delta\lp  \fq^\ast (s) \br(s') \fq (s) \rp
\nonumber
\\
= &\
\brho (s) \times \bsigma(s)  + \bpsi(s) -
\delta \lp   \fq^\ast(s,s')    \brho(s')   \fz(s,s') \rp
\nonumber
\\
= &\
 \brho (s) \times \bsigma(s)  + \bpsi(s) -
 \delta \fq^\ast(s)   \br(s')   \fq(s)  -
  \fq^\ast(s)  \delta  \br(s')   \fq(s) -
   \fq^\ast(s)   \br(s')  \delta  \fq(s)
\nonumber
\\
= &\
 \brho (s) \times \bsigma(s)  + \bpsi(s)
\nonumber
\\
&
-
\fs^\ast(s)  \big (\brho(s) -\bkappa(s,s') \big)
-
\big (\brho(s) -\bkappa (s,s') \big) \fs (s)
-
\fz(s,s')  \bpsi(s') \fz^\ast(s,s')
\nonumber
\\
=&\
\bpsi(s)  -
\fz^\ast(s,s')    \bpsi(s')     \fz(s,s')
-
 \, \bkappa(s,s') \times \bsigma
 \, ,
\label{deltarhonl}
\end{align}
which is a direct analogy of \eqref{deltakappa}.
We have used the fact that for purely imaginary $\fs$ we have, again
\[
\fs^\ast(s)  \brho(s) + \brho(s) \fs= -\fs \brho + \brho \fs =  \brho \times \bsigma
\]
and
\[
\fs^\ast(s)  \bkappa(s,s') + \bkappa(s,s') \fs= -\fs  \bkappa(s,s') +  \bkappa(s,s') \fs =
  \bkappa(s,s') \times \bsigma
\, .
\]
Next, let us  define the purely imaginary quaternion
\begin{equation*}
\fz ^\ast(s,s')   \delta \fz(s,s'):=
 2 \bT (s,s')
  \,.
\end{equation*}
The real part of $\fz ^\ast(s,s')   \delta \fz(s,s')$ vanishes since $\fz(s,s')$ is a unit quaternion.
The last step is the  variation with respect to $\bT$, computed as
\begin{align}
\fz ^\ast(s,s')   \delta \fz(s,s') &=
\fz^\ast(s,s')   \delta
\lsb
\fq^\ast (s) \fq(s')
\rsb
\nonumber
\\
& =
\fz^\ast(s,s')
\lsb
-\fq^\ast (s)    \delta \fq (s)     \fz(s,s') + \fq^\ast(s)    \fq(s')     \fq^\ast(s')
\delta \fq (s')
\rsb
\nonumber
\\
&=
-\,\frac{1}{2}  \fz^\ast (s,s')     \bsigma(s)     \fz(s,s') +
\frac{1}{2} \bsigma(s')
\, .
\end{align}
Thus, we find
\begin{align*}
2 \, \bT (s,s')
&:= {\rm Im} \lp  \fz ^\ast(s,s')    \delta \fz(s,s') \rp
  \nonumber
  \\
 &  =
-
\fz^\ast (s,s')    \bsigma(s)    \fz(s,s')
 +
  \bsigma(s')
  \nonumber
  \\
 &  =:
 -{\rm Ad} _{\fz^\ast} \bsigma(s) + \bsigma(s') \, .
\end{align*}
Note the exact correspondence between this formula and \eqref{xivar3} defining
the variation $\xi^{-1} \delta \xi$.
Therefore, the variation with respect to $\delta \bkappa$ gives
\begin{align}
\left<
\frac{\delta l_{np}}{\delta \bkappa}
\, , \,
\delta \bkappa
\right> = &
\int
\Big<
  \int
\frac{\partial U}{\partial \bkappa} (s,s') \times \bkappa (s,s')
\mbox{d} s'
\, , \,
\bsigma(s)
\Big> \mbox{d} s
\nonumber
\\
& +
\int
\Big<
 \int
\frac{\partial U}{\partial \bkappa} (s,s') -
\fz (s,s')     \frac{\partial U}{\partial \bkappa} (s',s)     \fz^\ast(s,s')
\mbox{d} s'
\, , \,
\bpsi(s)
\Big>
\mbox{d} s
\, .
\label{dlnpdrhovar2}
\end{align}
Analogously,
\begin{align}
\iint  \Big<
\fz^\ast   \frac{\partial U}{\partial \fz}
\, , \,
\fz^\ast   \delta \fz
\Big> \mbox{d}s \mbox{d}s'
& =
\iint  \Big<
\fz^\ast  \frac{\partial U}{\partial \fz}
\, , \,
-
\fz^\ast (s,s')    \bsigma(s)    \fz(s,s')
 +
  \bsigma(s')
\Big> \mbox{d}s \mbox{d}s'
\nonumber
\\
&=
\iint
\Big<
 \fz^\ast(s',s)   \frac{\partial U}{\partial \fz}(s',s)
- \frac{\partial U}{\partial \fz}(s,s')   \fz^\ast(s,s')
\, , \,
 \bsigma(s)
\Big> \mbox{d}s \mbox{d}s'
\nonumber
\\
&=
- \int  \Big< \int
 \frac{\partial U}{\partial \fz}(s,s')   \fz^\ast(s,s')
 \mbox{d}s '
\, , \,
 \bsigma(s)
\Big> \mbox{d}s
\nonumber
\\
&=
- \int  \Big<
\int {\rm Im} \left[
 \frac{\partial U}{\partial \fz}(s,s')   \fz^\ast(s,s')
 \right] \mbox{d}s '
 \, , \,
 \bsigma(s)
\Big> \mbox{d}s
\, .
\label{fzvar}
\end{align}
Collecting together the terms proportional to $\bsigma(s)$ and $\bpsi(s)$ in the minimal action principle \eqref{minaction} gives the system \eqref{hpvarsig}, \eqref{hpvarpsi}. The role of
antisymmetric matrix $Z(s,s')$ describing the nonlocal interactions in \eqref{hpvarpsi} is now played by the purely imaginary quaternion
\[
\bZ(s,s') = {\rm Im} \left[
 \frac{\partial  U}{\partial \fz}(s,s')   \fz^\ast(s,s')
 \right].
\]

A Hamiltonian description closely following that of Section~\ref{sec:Hamiltonian} can be developed in quaternionic form, as well. The cross products are then substituted by a corresponding product of the quaternions, with explicit formulas for the Lie-Poisson bracket closely resembling (\ref{convect-brkt}). Since the derivation is analogous to Section~\ref{sec:Hamiltonian}, it will be omitted from the exposition. 

\section{Outlook for further studies}
\label{sec:Conclusions}
This paper formulated the problem of strand dynamics for an arbitrary long-range intermolecular potential in the convective representation \cite{HoMaRa1986} of exact geometric rod theory \cite{SiMaKr1988}. Its methods would also apply in the consideration of Lennard-Jones potentials and the constrained motion of non-self-interacting curves.
\smallskip

After a quick derivation of the equations of motion by the Hamilton-Pontryagin approach, the paper demonstrated and compared three different approaches to deriving the same continuum equations of motion for an elastic strand experiencing nonlocal (for example, electrostatic or Lennard-Jones) interactions. These were: (1) the Euler-Poincar\'e approach; (2) the affine transformation approach and (3) the covariant Lagrange-Poincar\'e formulation. The paper concentrated primarily on the case in which the strand is one-dimensional, which is the main object of interest for biological applications. However, these three approaches possess more significance and applicability than might be suggested by the one-dimensional developments illustrated here. For example, the geometrical considerations and nonlinear context of the present investigation also apply in formulating the dynamics of the higher dimensional case. That is, when $s$ has more than one component, the approaches discussed here still apply. 
\smallskip

A change in dimensionality of $s$ in equations \eqref{hpvarsig} and \eqref{hpvarpsi} requires summing over all components of $s$-derivatives (instead of only the single $s$-derivative for the strand). Additional integrability conditions arise from the equality of cross-derivatives with respect to space and time that generalize equations \eqref{kincond} and \eqref{kincondom}.    (In geometric terms, these are zero curvature conditions.) The extension to higher dimensions was discussed in the general setting treated in Section~\ref{n_dimensional_generalization}. The higher dimensional options also figured in the covariant Lagrange-Poincar\'e formulas \eqref{cov_L-P}, where ${\bf div}_x$ denotes derivative with respect to time and {\it all} dimensions of the space (taken to be one-dimensional in the paper).
The extension to higher dimensions illuminates the geometry underlying the present one-dimensional case and may be expected to produce interesting applications in the dynamical description of biological membranes and other extended physical objects. While the equations take the same geometrical form in higher dimensions, their solutions will possess their own unique features.
\smallskip

Besides passing to higher dimensions, future studies will consider both linear and nonlinear wave propagation on electrostatically charged strands, as well as the description of nontrivial stationary states that arise from nonlocal interactions, such as for the VDF oligomers mentioned in the Introduction.
\smallskip

Yet another interesting question for future studies concerns the possibility of enhancing the internal structure of the rigid charge conformations. This will allow even richer dynamics than we considered here. While the resulting equations may be different (and more complex), the methods developed in this paper will still be applicable when the dynamics takes place in spaces that possess richer conformational structure than rigid rotations.

\section*{Acknowledgements}
DDH and VP were partially supported by NSF grant NSF-DMS-05377891. The work of
DDH was also partially supported  by the Royal Society of London Wolfson Research Merit Award. VP acknowledges the support of  the European Science Foundation for partial support through the MISGAM program. FGB and TSR acknowledge the partial support of a Swiss National Science Foundation grant.

\bibliographystyle{unsrt}
\bibliography{papers}
\end{document}